\documentclass[a4paper,oneside,fleqn,11pt]{article}
\usepackage{amsmath,amssymb,amsfonts,amsthm,type1cm,bm,color,colortbl, mathrsfs}
\usepackage{graphicx,psfrag,epsf,booktabs, subfigure}
\usepackage{xcolor}
\usepackage[margin=1in]{geometry}
\usepackage{setspace}
\usepackage{natbib}
\usepackage{rotating}
\usepackage{authblk}%
\usepackage{comment}
\usepackage{accents}
\usepackage{yfonts}
\usepackage{pdflscape}
\RequirePackage[colorlinks,citecolor=blue,urlcolor=blue]{hyperref}
\usepackage{bigstrut}
\usepackage{mathtools}
\mathtoolsset{showonlyrefs=true}
\usepackage{multibib}
\newcites{suppl}{References for Supplementary Material}
\usepackage{enumerate}
\usepackage{bbm}
\usepackage{pdfpages}
\usepackage{enumitem}
\usepackage{float}
\usepackage{placeins}
\usepackage{here}
\usepackage{threeparttable}
\usepackage{adjustbox}

\DeclareMathOperator*\argmin{\arg\min}
\DeclareMathOperator\diag{diag}

\DeclareMathOperator\supp{supp}

\DeclareMathOperator\E{\mathbb{E}}
\DeclareMathOperator\Var{Var}
\DeclareMathOperator\Cov{Cov}
\DeclareMathOperator\Pro{\mathbb{P}}

\DeclareMathOperator\sgn{sgn}

\def\e{\mathrm{e}}

\def\bA{\mathbf{A}}
\def\bB{\mathbf{B}}
\def\bC{\mathbf{C}}
\def\bD{\mathbf{D}}
\def\bE{\mathbf{E}}
\def\bF{\mathbf{F}}
\def\bG{\mathbf{G}}
\def\bH{\mathbf{H}}
\def\bI{\mathbf{I}}

\def\bM{\mathbf{M}}

\def\bQ{\mathbf{Q}}
\def\bR{\mathbf{R}}
\def\bS{\mathbf{S}}
\def\bT{\mathbf{T}}
\def\bU{\mathbf{U}}
\def\bV{\mathbf{V}}

\def\bX{\mathbf{X}}

\def\ba{\mathbf{a}}
\def\bb{\mathbf{b}}

\def\be{\mathbf{e}}
\def\bff{\mathbf{f}}

\def\bq{\mathbf{q}}
\def\br{\mathbf{r}}

\def\bu{\mathbf{u}}
\def\bv{\mathbf{v}}
\def\bw{\mathbf{w}}
\def\bx{\mathbf{x}}
\def\by{\mathbf{y}}
\def\bz{\mathbf{z}}
\def\cA{\mathcal{A}}
\def\cB{\mathcal{B}}

\def\cN{\mathcal{N}}
\def\cS{\mathcal{S}}
\def\cU{\mathcal{U}}

\def\bbR{\mathbb{R}}

\def\biota{\boldsymbol{\iota}}

\def\bgamma{\boldsymbol{\gamma}}
\def\bbeta{\boldsymbol{\beta}}

\def\bDelta{\boldsymbol{\Delta}}
\def\bdelta{\boldsymbol{\delta}}
\def\bkappa{\boldsymbol{\kappa}}

\def\bmu{\boldsymbol{\mu}}

\def\bTheta{\boldsymbol{\Theta}}
\def\bOmega{\boldsymbol{\Omega}}
\def\bomega{\boldsymbol{\omega}}

\def\bSigma{\boldsymbol{\Sigma}}

\def\ep{\varepsilon}
\def\bep{\boldsymbol{\varepsilon}}
\def\bzero{\mathbf{0}}
\def\bone{\mathbf{1}}
\def\what{\widehat}
\def\wtilde{\widetilde}

\newcommand{\vsp}{\vspace{15pt}}
\newcommand{\vspp}{\vspace{8pt}}

\newlength{\dhatheight}
\newcommand{\doublehat}[1]{%
    \settoheight{\dhatheight}{\ensuremath{\hat{#1}}}%
    \addtolength{\dhatheight}{-0.35ex}%
    \hat{\vphantom{\rule{1pt}{\dhatheight}}%
    \smash{\hat{#1}}}}

\newtheorem{thm}{Theorem}
\newtheorem{lem}{Lemma}
\newtheorem{cor}{Corollary}
\newtheorem{con}{Condition}
\newtheorem{rem}{Remark}
\newtheorem{prop}{Proposition}
\newtheorem{proc}{Procedure}

\makeatletter
\def\section{\@startsection {section}{1}{\z@}{-3.5ex plus -1ex minus-.2ex}{2.3ex plus .2ex}{\large\bf}}
\makeatother

\makeatletter
\def\subsection{\@startsection {subsection}{1}{\z@}{-3.5ex plus -1ex minus-.2ex}{2.3ex plus .2ex}{\normalsize\bf}}
\makeatother

\title{\textbf{Post-Screening Portfolio Selection}}

\author[ ]{\textsc{Yoshimasa Uematsu}\thanks{Correspondence: Yoshimasa Uematsu, Department of Social Data Science, Hitotsubashi University, 2-1 Naka, Kunitachi, Tokyo 186-8601, Japan (E-mail: yoshimasa.uematsu@r.hit-u.ac.jp)}}

\author[ ]{\textsc{Shinya Tanaka$^\dagger$}
}

\affil[$*$]{\textit{Department of Social Data Science, Hitotsubashi University}}
\affil[$\dagger$]{\textit{Department of Economics, Otaru University of Commerce}}

\begin{document}

\renewcommand{\theequation}{\thesection.\arabic{equation}}
\makeatletter
\@addtoreset{equation}{section}
\makeatother

\maketitle

\begin{abstract}
We propose post-screening portfolio selection (PS$^2$), a two-step framework for high-dimensional mean--variance investing. First, assets are screened by Lasso-type regression of a constant on excess returns without an intercept. Second, portfolio weights are estimated on the selected set using standard low-dimensional methods. Because strong factors can destroy sparsity in real data, we further introduce PS$^2$ with factors (FPS$^2$), which defactors returns before screening and allows factor investing in the final step. We establish theoretical guarantees, and simulations and an empirical application show competitive performance, especially when sparse screening is appropriate or strong factors are explicitly accommodated.
\end{abstract}
Keywords: Mean-Variance Portfolio, Large Portfolio Selection, Lasso Screening, Post-Lasso OLS, Factor Investing

\section{Introduction} 

Since the seminal paper of \cite{Markowitz1952}, the so-called “mean–variance approach” has been central for portfolio construction. It prescribes portfolio weights that minimize risk, measured by the variance of returns, among all portfolios that achieve a target expected return. In the population, when the mean vector and an invertible covariance matrix are known, the optimal weights are available in closed form. In reality, however, these population quantities must be estimated from finite samples. This estimation step is particularly challenging in modern applications with an expanding universe of investable assets. While a larger opportunity set should, in principle, expand the mean–variance efficient frontier, the need to estimate high-dimensional parameters makes the resulting optimal portfolio fragile when the cross-sectional dimension $N$ is comparable to or exceeds the sample size $T$. 

This difficulty suggests a more basic question: in a high-dimensional asset universe, must one estimate the optimal portfolio weight vector itself in high dimension in order to construct an efficient portfolio? Our answer is no. Rather than estimating it in high dimension, we show that portfolio construction can be decomposed into two conceptually distinct tasks: identifying which assets matter for the efficient portfolio, and estimating their weights only after the problem has been reduced to a low-dimensional one. This reformulation is the central organizing idea of the paper.

\subsection{PS$^2$: A two-step perspective on high-dimensional portfolio choice}

We propose a two-step framework, called \textit{post-screening portfolio selection} (PS$^2$), for high-dimensional mean-variance investing (Figure \ref{fig:psps}). The starting point is a structural identity that decomposes one-shot portfolio construction into support recovery and subsequent low-dimensional weight estimation. Specifically, we show that the support of the population mean-variance optimal weight vector coincides with that of a regression coefficient vector obtained by regressing a constant on excess returns without an intercept (Proposition \ref{prop:1}). This observation shifts the first-stage target away from the portfolio weight itself and toward a simpler regression target that shares the same support and is directly amenable to screening.

This reformulation has two important consequences. First, the high-dimensional task is no longer to estimate the optimal weight vector accurately in one step, but to identify a set of assets that contains its relevant support. Second, once this screening step has reduced the dimension, portfolio weights can be estimated using standard low-dimensional methods on the screened set. In this sense, PS$^2$ separates support recovery from weight estimation: high-dimensional estimation enters only through screening scores, whereas the final portfolio is constructed only after the problem has been reduced to a low-dimensional one.

Concretely, when the cross-sectional dimension is moderate relative to the sample size, the screening target can be estimated by OLS. When the asset universe is large, it can be estimated by the Lasso or its variants. The second step then uses only the selected assets to construct the portfolio, for example through a plug-in estimator or an OLS-based estimator. While the latter is related in spirit to MAXSER \citep{AoEtAl2019}, the conceptual distinction is important: MAXSER combines selection and weight estimation in a single high-dimensional step, whereas PS$^2$ targets support recovery first and performs weight estimation only after screening. Thus, our contribution is not regularization per se, but a different decomposition of the portfolio problem.

\begin{figure}[!h]
\centering
\includegraphics[width=13cm]{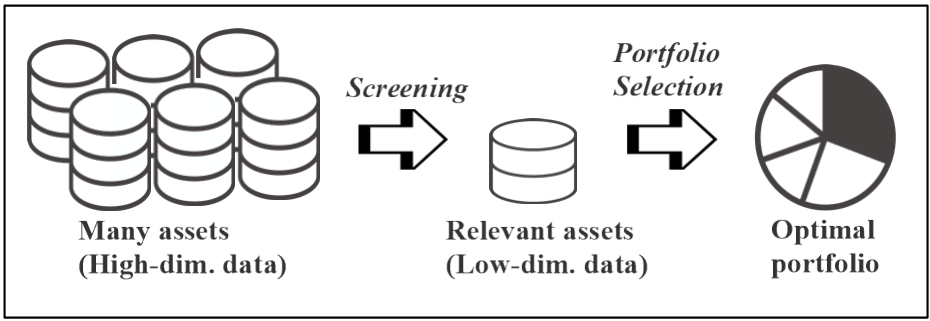}
\caption{Conceptual diagram of PS$^2$: we first screen a large set of potential assets to select eligible candidates, then form the portfolio using only the selected assets.}\label{fig:psps}
\end{figure}

\subsection{Modifying PS$^2$: Strong factors and the failure of sparse screening}

The success of PS$^2$ hinges on an economically meaningful form of sparsity: the population efficient portfolio weight must load on only a relatively small subset of assets, at least approximately. Such a “small-portfolio” structure is attractive not only statistically but also practically, since it promotes stability and helps reduce turnover and trading frictions. At the same time, sparsity is not automatic in financial data.

In particular, as we will see in Section \ref{subsec:obsf}, strong pervasive factors with non-negligible means create two distinct failure channels for sparse screening. First, they can make the screening target effectively dense. If a common factor loads broadly across assets and has a non-negligible mean, then the mean vector inherits a pervasive component, and the resulting screening target need not be sparse even when sparse screening would otherwise be appropriate. Second, the same strong factor induces severe multicollinearity across asset returns. In such environments, Lasso-type methods tend to treat highly redundant assets as substitutes, leading to unstable screening and unreliable discoveries.

These observations motivate a modified procedure, \textit{PS$^2$ with factors} (FPS$^2$), that incorporates a preliminary defactoring step before screening and allows factor investing in the subsequent portfolio-construction step. The role of FPS$^2$ is therefore not cosmetic. It is designed to repair a specific and empirically relevant failure mode of sparse screening in high-dimensional asset returns. This modification changes not only the implementation but also the object of screening. In FPS$^2$, the first step screens the defactored residual component, while the second step forms the final portfolio using the selected assets together with the investable factors.

\subsection{Main contributions}

Our main contributions are threefold.
\begin{itemize}
\item \textbf{A new identification-based framework for high-dimensional portfolio choice.}
We show that the support of the population mean--variance optimal weight vector coincides with that of a coefficient vector obtained by regressing a constant on excess returns without an intercept. This support identity transforms portfolio construction from a high-dimensional weight-estimation problem into a support-recovery problem and leads to PS$^2$, a two-step procedure that combines screening with subsequent low-dimensional weight estimation.

\item \textbf{A comprehensive theoretical analysis in high dimensions.}
We establish high-probability guarantees for the screening step, derive error control for the post-screening portfolio estimator, characterize conditions under which the population optimal weight vector is approximately sparse, and analyze the strong-factor failure mode together with the defactoring-based remedy. These results provide theoretical support for our approach in the regime $N/T \to \infty$.

\item \textbf{Evidence on portfolio performance and practical relevance.}
Monte Carlo experiments show that PS$^2$ and FPS$^2$ deliver reliable screening and stable post-screening weight estimation when sparse screening is appropriate, and clarify the strong-factor environment in which FPS$^2$ becomes essential. In an empirical application to S\&P 500 constituents, FPS$^2$ delivers competitive out-of-sample performance relative to benchmark portfolios, particularly in relatively stable market environments, while transaction costs materially weaken the gains from active rebalancing.
\end{itemize}

\subsection{Related works}

A central difficulty in high-dimensional portfolio choice is the estimation of second-moment objects. \citet{JagannathanMa2003} highlight the sensitivity of optimal portfolios to covariance estimation error, motivating structured and shrinkage-based estimators such as \citet{LedoitWolf17}. Within this line, \citet{DeNardEtAl21} develop a factor-based covariance estimator for portfolio selection in large dimensions, while \citet{KimOh2024} combine high-frequency information with CLIME-type \citep{CaiEtAl2011} precision-matrix estimation for the global minimum-variance portfolio. These approaches stabilize portfolio construction by improving high-dimensional covariance or precision-matrix inputs. By contrast, our procedure uses the high-dimensional sample primarily for screening, so the final weight-estimation step is carried out only after the problem has been reduced to a low-dimensional one.

A related strand studies large-scale portfolio construction by imposing sparsity, regularization, or other forms of dimension reduction directly in the portfolio problem. \citet{AoEtAl2019} propose MAXSER, which uses a regression-based formulation to estimate sparse mean--variance portfolios. \citet{DuEtAl2022} study high-dimensional portfolio selection with cardinality constraints, emphasizing that only a subset of assets may be needed in large investment universes. \citet{FanEtAl12} study minimum-variance portfolios with gross-exposure constraints and show that moderate regularization can improve portfolio performance by mitigating estimation error. Our paper is closest in spirit to this literature. The key distinction, however, is not regularization per se, but the separation between support recovery and weight estimation: PS$^2$ screens assets first and estimates portfolio weights only after dimension reduction.

Our paper is also related to recent work that uses high-dimensional methods to identify economically relevant return components. Building on the double-selection Lasso framework of \citet{BelloniEtAl2014}, \citet{FengEtAl2000} study whether newly proposed factors add incremental information beyond an existing factor set, while \citet{FreybergerEtAl2020} develop a nonparametric characteristic-selection approach for expected returns. Relative to these papers, our contribution is to identify assets that matter for the efficient portfolio itself, including not only assets with strong standalone mean components but also assets that are useful because of their covariance structure.

Most theoretical results in this literature focus on regimes where $N/T \to c$ for a finite constant $c$, although some papers allow $c \geq 1$; see, for example, \citet{CallotEtAl2021} and \citet{BodnarEtAl2020}. In contrast, our analysis targets the high-dimensional regime $N/T \to \infty$, which necessitates different techniques and yields distinct implications for screening and subsequent portfolio construction.

\subsection{Organization}
The remainder of this paper is organized as follows. 
Section \ref{sec:Prelim} reviews the optimal portfolio and three representative methods for estimating the weight. 
Section \ref{ModelandEstimation} proposes our procedures, PS$^2$ and its modified version FPS$^2$. 
Section \ref{sec:theory} gives a comprehensive statistical theory for our procedures. 
Section \ref{MC} demonstrates our procedures in simulation studies. 
Section \ref{emp:mainsec} applies our procedures to a real dataset. 
Section \ref{sec:concl} concludes. All the proofs are collected in Supplementary Material.

\section{Preliminaries} \label{sec:Prelim}

We first review the Markowitz minimum variance portfolio (MVP) in the population in Section \ref{subsubsec:mvp}. This defines an optimal weight of the portfolio, 
which will serve as the \textit{true parameter} to be estimated, for a pre-determined target return. Sections \ref{subsubsec:plug} and \ref{subsec:UR} summarize 
representative estimation methods for the optimal weight. These methods construct optimal portfolios in a ``one-shot'' manner (i.e., estimation is performed in 
a single step rather than multiple stages), which can be fragile in high-dimensional settings.

\subsection{Minimum variance portfolio}\label{subsubsec:mvp}

Suppose there are $N$ risky assets. Let $ \br = (r_{1},\dots,r_{N})'$ 
denote the  $N$-dimensional vector of excess returns with mean $ \bmu $ and variance $ \bSigma$, where $ \bmu \not= \bzero $ and $ \bSigma $ is 
positive definite. Given a weight vector $ \bw = \left(w_{1}, \dots, w_{N} \right)' $, the excess return of this portfolio is $\bw'\br$. 
For any target excess return $ \bar{\rho} $, the MVP weight $\bw^*(\bar{\rho})$ is obtained by
\begin{align}
\bw^*(\bar{\rho})  := \argmin_{\bw\in\mathbb{R}^N,~\bw' \bmu = \bar{\rho}}  \bw' \bSigma \bw \
=\frac{\bar{\rho}}{\theta} \bSigma^{-1} \bmu, \label{OPw1}
\end{align}
where $\theta:=\bmu' \bSigma^{-1} \bmu >0$. 
The Sharpe ratio of a portfolio with its weight $ \bw $ is defined as  $\mathrm{SR}(\bw) 
= \bw' \bmu /\sqrt{ \bw' \bSigma \bw}$. Then we have $\mathrm{SR}(\bw^*(\bar{\rho}))=\sqrt{\theta}$, achieving the maximum among all possible efficient portfolios constructed by $N$ risky assets 
\textit{irrespective} of the value of $\bar{\rho}$, because $ \mathrm{SR}(\bw^*(\bar{\rho})) $ is equivalent to the Sharpe ratio of the tangency  portfolio by definition.

Since $ \bmu $ and $ \bSigma $ are unknown in practice, $ \bw^*(\bar{\rho}) $ is also unknown. Thus our goal is to estimate $ \bw^*(\bar{\rho}) $ as accurate as possible with $T$ stationary observations, 
$ \{\br_t\}_{t=1}^T $. Given $\bar{\rho}$, we introduce two approaches for estimating $ \bw^*(\bar{\rho}) $ 
in the following subsections. Hereafter, when it is unnecessary to display the dependence on $\bar{\rho}$, we write $\bw^*(\bar{\rho})$ simply as $\bw^*$, and we denote its estimators simply by $\what{\bw}$, etc.

\subsection{Existing approaches}

\subsubsection{Plug-in approach}\label{subsubsec:plug}

The first one is the \textit{plug-in} approach. This simply plugs some estimators $ (\hat{\bmu}, \what{\bSigma}^{-1}) $ of $ (\bmu,  \bSigma^{-1}) $  into \eqref{OPw1}: $\what{\bw}_\textsf{plug} := {\bar{\rho}}\what{\bSigma}^{-1} \hat{\bmu}/{\hat{\theta}}  $ with $\hat{\theta}=\hat{\bmu}'\what{\bSigma}^{-1} \hat{\bmu}$. 
Conventionally, we may use a sample mean and sample 
precision (inverse of the sample covariance) of $ \{\br_t\}_{t=1}^T $, respectively, but they will not behave well when $N$ is large. In particular, the inverse sample 
covariance matrix does not exist  when $ N > T$. For instance, \cite{CallotEtAl2021} employ the nodewise regression estimator 
\citep{MeinshausenBuhlmann2006} for $ {\bSigma}^{-1} $ to manage high dimensionality, but it is challenging to precisely estimate large-dimensional parameters in general. 
A series of their studies, including \cite{LedoitWolf17}, also falls into this category of plug-in methods.

\subsubsection{Regression approach}\label{subsec:UR}

The second one is the \textit{unconstrained regression} approach. Define $\bw^\dagger$
as the coefficient vector of the population regression of $ \bar{\rho} $ on $ \br $; 
that is, $\bw^\dagger$ is a minimizer of $ \E \left( \bar{\rho} - \bw'\br \right)^2 $. This is explicitly given as
\begin{align}
\bw^\dagger
= \dfrac{ \bar{\rho}}{1 + \theta}  \bSigma^{-1} \bmu.  \label{regopt}
\end{align} 
Notice that $ \mathrm{SR}(\bw^\dagger)=\mathrm{SR}(\bw^*)=\sqrt{\theta}$ holds (for any choice of $ \bar{\rho} $) as 
$ \bw^\dagger$ is proportional to $ \bw^*$. 
\cite{Britten-Jones1999}, a pioneer of this approach, considered the population regression with $ \bar{\rho} = 1 $. 
We call the approach with arbitrary $ \bar{\rho} $ the \textit{BJ-type} (population) regression. 

Given observations $ \{\br_t\}_{t=1}^T $, we consider the empirical version of the BJ-type regression. 
Let $ \biota_T $ be $ T \times 1 $ vector with ones. The estimator of $\bw^\dagger$, denoted by $\what{\bw}^\dagger$, 
is obtained by the OLS regression of $ \bar{\rho} \biota_T $ on $T\times N$ data matrix $ \bR = (\br_1, \dots, \br_T)'  $: 
\begin{align}
\what{\bw}^\dagger = \bar{\rho} \left(\bR'\bR \right)^{-1} \bR' \biota_T.  \label{olsBJ}
\end{align}
The BJ-type regression has two drawbacks. 
First, $ \bw^\dagger$ achieves the maximum Sharpe ratio, but cannot meet the target return. 
To see this, we have $\E [ \bw^{\dagger\prime}\br ] = \theta \bar{\rho} /(1+\theta) < \bar{\rho} $ 
since $ \theta > 0. $
Namely, the BJ-type portfolio is optimal in terms of efficiency, but conservative. 
Second, it is not applicable in a high-dimensional setting by the construction, especially when $ N>T $. 


\cite{AoEtAl2019} extend the BJ-type regression to the \textit{maximum-Sharpe-ratio estimated \& sparse Regression} (MAXSER), addressing the first drawback. The key idea is to employ 
a well-crafted response, $ r_c \biota_T$ with $r_c:= \bar{\rho} (1+ \theta)/\theta $, instead of $ \bar{\rho} \biota_T $. We call $ r_c $ the \textit{inflated} constant response. In view of the fact that $ \bw^\dagger\left(  \bar{\rho} (1+ \theta)/{\theta}\right) =\bw^*(\bar{\rho}) $ 
for any $ \bar{\rho} $, the BJ-type (population) regression of $r_c$ on $ \br $ yields the optimal portfolio that achieves the target return $ \bar{\rho} $. 
As an empirical analogue, they apply the idea to the $\ell_1$-penalized regression  (Lasso) to obtain the MAXSER portfolio:
\begin{align}
\what{\bw}_\textsf{MAXSER} := ~\argmin_{\bw\in\mathbb{R}^N}  \dfrac{1}{T} \sum_{t=1}^T \left( \hat{r}_c - \bw'\br_t  \right)^2  
+ 2 \lambda \| \bw \|_1,  \label{OPaoLasso_pgin}
\end{align} 
where $ \hat{r}_c = \bar{\rho} (1 + \hat{\theta})/\hat{\theta}$ 
with $ \hat{\theta} $ a consistent estimator of $ \theta$ and $\lambda$ is a positive regularization parameter. 

Despite the appeal of Lasso for handling high dimensionality, the MAXSER does not fully resolve the second drawback noted above, because the construction of $\hat{\theta}$ still relies on low-dimensional ingredients. In particular, the method adopts the bias-corrected estimator of \cite{KanZhou2007},
$ \hat{\theta} = \{(T-N-2)\hat{\theta}_s - N\}/T $ with $\hat{\theta}_s=\hat{\bmu}_s'\what{\bSigma}_s^{-1} \hat{\bmu}_s$, where $\hat{\bmu}_s$ is the sample mean and $\what{\bSigma}_s$ is the sample covariance matrix, and its modified version. This formulation requires $\what{\bSigma}_s$ to be invertible, thereby precluding regimes where $N \ge T$. Furthermore, the consistency of this estimator hinges on the assumption of Gaussianity, and guaranteeing the positivity of $\hat{\theta}$ intrinsically requires $N$ to be much smaller than $T$. Thus, while Lasso regularizes the weight estimation, the required estimation of $\theta$ in MAXSER does not adequately accommodate high dimensionality.

\begin{rem}\normalfont
In addition to the plug-in and the unconstrained approaches, there exist various other approaches to construct the optimal portfolio in high-dimension including 
constrained regression \citep{BrodieEtAl2009}, stable high-dimensional covariance matrix estimation \citep{FanEtAl12, KimOh2024}, 
cardinality constraints \citep{DuEtAl2022},  nonparametric method \citep{FreybergerEtAl2020}, distributionally robust method \citep{BlanchetEtAl22, WuEtAl2025} 
and others.  
\end{rem}

\section{Methodology}
\label{ModelandEstimation}

We propose a two-stage portfolio construction framework broadly applicable in general high-dimensional settings. 
Consider $N$-dimensional vector of excess returns $\br_t$ with mean vector $ \bmu $ and covariance matrix $ \bSigma$. 
For a clear presentation, suppose that $\br_t$ is not driven by any factors and the optimal weight $\bw^*=\bar{\rho}\bSigma^{-1}\bmu/\theta$ given 
in \eqref{OPw1} is exactly sparse at this moment. 
We examine situations in which $\bw^*$ may fail to exhibit sparsity in Section \ref{subsec:obsf}. 
Denoting by $w_j^*$ the $j$th element of $\bw^*$, we define 
\begin{align}
\cS= \left\{j\in[N]:w_j^*\not=0 \right\},
\end{align}
and call $j\in\cS$ the (indices of) \textit{relevant} assets. Our procedure first attempts to \textit{screen} them by some statistical method. Denote by $\what\cS$ the set of screened assets (called \textit{discoveries}). Ideally, it should be the same as $\cS$ or the smallest possible set that completely contains $\cS$. If so, focusing only on $\what\cS$, we may efficiently estimate $\bw^*$ as a low-dimensional problem. That is, we obtain estimates $\hat{w}_j$ for $j\in\what\cS$ by any well-established method, such as the plug-in approach or the BJ-type regression in Section \ref{sec:Prelim}, using only a low-dimensional data set $\{\bR_j:j\in\what\cS\}$, 
where $ \bR_j $ is $j$th column of $ \bR $,  
and simply put $\hat{w}_j=0$ for $j\in\what\cS^c$. Thus the primary task is to obtain a ``nice'' $\what\cS$ that will minimally include $\cS$. 

Define $\bOmega=(\E\br_t\br_t')^{-1}$. 
Our construction of $\what\cS$ relies on the following proposition.
\begin{prop}\label{prop:1}
Suppose that $\bSigma^{-1}$ and $\bOmega$ are well-defined. For any given target return $\bar{\rho}$, the optimal portfolio weight is then expressed as
\begin{align}
\bw^*(\bar{\rho})\equiv \frac{\bar{\rho}}{\theta}\bSigma^{-1}\bmu
=\frac{\bar{\rho}(1+\theta)}{\theta}\bOmega\bmu. \label{SiginvOme}
\end{align}
In particular, if we define $\bbeta(\alpha)=\alpha\bOmega\bmu$ for any $\alpha\in\bbR\backslash\{0\}$, then $\bw^*$ and $\bbeta(\alpha)$ share the exact same support. 
\end{prop}

Proposition~\ref{prop:1} states that $\cS=\{j:\beta_j(\alpha)\not=0\}$ is true. This has two important implications. 
First, if our interim goal is to recover the support of $\bw^*$, $\cS$, we do not need to work with $\bw^*$ itself, which involves the nuisance scaling factor $\theta$; it suffices to study the simpler vector $\bbeta(\alpha)$ for any nonzero constant $\alpha$. This contrasts with MAXSER, which effectively requires accurate estimation of $(1+\theta)/\theta$ as it attempts to perform selection and weight estimation simultaneously. 
Second, when $N$ is small, $\bbeta(\alpha)$ can be easily estimated via the OLS regression of $\alpha \biota_T$ on $\bR$. Under standard regularity conditions, it satisfies
\begin{align}
\what\bbeta_\textsf{OLS}= \left(\frac{1}{T}\bR'\bR\right)^{-1}\left(\frac{\alpha}{T}\bR'\biota_T \right)
\xrightarrow{p} \alpha\bOmega \bmu \equiv \bbeta(\alpha),\label{ex_ols}
\end{align}
leading to a tractable screening step. Remarkably, although \eqref{SiginvOme} shows that $\bSigma^{-1}\bmu$ and $\bOmega\bmu$ differ only by a scalar multiple, they behave very differently from the viewpoint of regression-based identification. The vector $\bOmega\bmu$ is directly identified by regressing $\alpha\biota_T$ on $\bR$ \emph{without} an intercept. By contrast, $\bSigma^{-1}\bmu$ is not identified in an equally simple way. Indeed, if one regresses $\alpha\biota_T$ on $\bR$ \emph{with} an intercept, then the population slope coefficient is exactly zero, because the constant response is fully absorbed by the intercept. Hence $\bSigma^{-1}\bmu$ cannot be recovered through such a regression. Therefore, $\bbeta(\alpha)$ is not merely a convenient surrogate; it is the unique natural target for regression-based screening when one wishes to avoid direct high-dimensional estimation of the centered covariance matrix $\bSigma$.

\subsection{Post-screening portfolio selection}\label{subsec:lasso}

We are now in position to introduce our procedure, Post-Screening Portfolio Selection, abbreviated as PS$^2$.  

\begin{proc}[PS$^2$]\normalfont \label{proc:screening}
Given $\bR$, we obtain $\what{\bw}$, an estimate of the optimal weight vector $\bw^*$, by Steps 1--2:
\begin{enumerate}[label=Step \arabic*., leftmargin=*]
\item[Step 1.] Compute the estimator $\what{\bbeta}$ of the target parameter $\bbeta(\alpha)$ by a regression of 
$\bz := \alpha\biota_T$ on $\bR$, where  
$\alpha>0$ is arbitrary, and set $\what{\cS}=\{j:\hat{\beta}_j\not=0\}$. 
\item[Step 2.] Obtain the optimal weight vector $\what{\bw}$ as follows:
\begin{enumerate}
   \item Set $\hat{w}_{j}=0$ for every $j\in\what\cS^c$.
    \item Estimate $\hat{w}_{j}$ for $j\in \what\cS$ by any method using only discovered assets in $\what\cS$.
\end{enumerate}
\end{enumerate}
\end{proc}

This consists of two stages: screening (Step 1) and estimation (Step 2). Step 1 performs regression of a constant on data. A key advantage of this regression-based formulation is its scalability for high-dimensional settings ($N / T\to\infty$), allowing the use of the Lasso, its variants, or other regularization methods. The selected set $\what\cS$ should be designed to satisfy the sure screening property, $\cS\subseteq \what\cS$, to ensure significant \textit{power}. Simultaneously, suppressing the number of irrelevant selected assets (i.e., \textit{false discoveries}), $|\cS^c\cap\what\cS|$, is important to improve the portfolio's performance. Step 2(a) enforces sparsity, followed by Step 2(b), which performs portfolio construction within the reduced low-dimensional subspace. These steps align with the sparsity assumption on $\bw^*$. Furthermore, even if $\bw^*$ is not exactly sparse, our procedure remains effective provided that the majority of 
components are sufficiently small (approximate sparsity).

\begin{rem}\normalfont
In Step 1, regressing the constant response $\bz = \alpha \biota_T$ on $\bR$ via the Lasso may cause computational instability depending on the software package. In such cases, it is advisable to use a slightly perturbed response $\bz \sim \mathcal{N}(\alpha\biota_T, \tau^2 \bI_T)$ instead of the exact constant $\alpha \biota_T$, where $\tau$ is a small positive constant.
\end{rem}

\subsubsection{Example: Post-Lasso-screening OLS-type construction}\label{subsubsec:post}

In Step~1, one of the simplest constructions of $\what\cS$ is based on the Lasso of $\bz$ on $\bR$:
\begin{align}
\what{\bbeta}_\textsf{L} = \argmin_{\bbeta\in\bbR^N} T^{-1}\|\bz-\bR\bbeta\|_2^2 + 2\lambda_\textsf{L} \|\bbeta\|_1, \label{feasibleLasso}
\end{align}
where $\lambda_\textsf{L}$ is a positive regularization coefficient. We then obtain the set of discoveries as $\what\cS_\textsf{L}=\{j: |\hat{\beta}_j^\textsf{L}| > 0\}$. Other screening methods, such as the adaptive Lasso, can also be used. Theoretically, the Lasso is generally known to yield a set of discoveries larger than the truth due to regularization bias, whereas the adaptive Lasso tends to provide sharper discoveries by mitigating this bias. For additional discussion of the adaptive Lasso, see Section C in Supplementary Material. 

Beyond the fundamental role of the Lasso that copes with high dimensionality, the $\ell_1$-penalty offers a beneficial side effect: it implicitly controls the portfolio's gross exposure during the screening phase \citep{BrodieEtAl2009,FanEtAl12}. Unlike unregularized methods that might select highly correlated pairs with extreme opposing weights, the Lasso tends to avoid such instability. This property ensures that the screened set $\what\cS$ is well-conditioned, facilitating a stable portfolio construction, such as the one via the post-Lasso OLS estimation, in the subsequent Step 2.

Unlike marginal screening, which evaluates each asset in isolation, the Lasso screening step uses all assets jointly and therefore accounts for the covariance structure of returns. As a result, it can select not only assets with strong standalone expected returns (``alpha takers'') but also assets that are valuable mainly because they hedge the risk of other assets (``hedgers''). This is consistent with the fact that the target parameter satisfies $\bbeta(\alpha)\propto \bSigma^{-1}\bmu$, so selection is driven by each asset's contribution to the mean--variance efficient portfolio rather than by its standalone mean alone. This intuition is related to \cite{kozak2020}, who emphasize that regularized regression methods can identify components of the stochastic discount factor by accounting for both expected returns and covariance structure.

\begin{rem}\normalfont
Another possibility is to consider controlling the false discovery rate (FDR) for screening. For instance, we may employ the BH method \citep{benjamini1995controlling} or the knockoff filter \citep{BarberCandes2015} for constructing $\what\cS$ via the regression of $\bz$ on $\bR$. These constructions can lead to a stable portfolio screening. We do not explore this direction further and leave it for future research.
\end{rem}

Regarding Step 2 for the weight estimation, the simplest construction is due to the plug-in approach, employing only the assets contained in $\what{\cS}$ while setting $\hat{w}_j=0$ for all $j\in \what{\cS}^c$, where $\what\cS$ stands for $\what\cS_\textsf{L}$ or $\what\cS_\textsf{aL}$, for instance. Another interesting example is an OLS-based approach, which similarly focuses only on $\what{\cS}$. Specifically, inspired by the MAXSER, we obtain the (non-zero) weight estimates by regressing the inflated response $\hat{r}_c\biota_T$ onto $\bR_{\what\cS}$. Here, $\hat{r}_c=\bar{\rho}(1+\hat{\theta})/\hat{\theta}$, where $\hat{\theta}$ is defined as a simple plug-in estimator consisting of the sample mean vector and inverse of the sample covariance matrix of $\bR_{\what\cS}$. The post-screening OLS estimator is then obtained as $\what{\bw}=(\what{\bw}_{\what\cS}',\bzero_{N-\hat{s}}')'$, where $\hat{s}=|\what\cS|$ and 
\begin{align}
\what{\bw}_{\what\cS}=\hat{r}_c(\bR_{\what{\cS}}'\bR_{\what{\cS}})^{-1}\bR_{\what{\cS}}'\biota_T. \label{post_L_OLS}
\end{align}
Given $\what\cS$, the entire construction relies on low-dimensional methods, thereby avoiding the unstable precision matrix and Sharpe ratio estimation in high dimensions.

\subsection{Modification with defactoring and factor investing}\label{subsec:obsf}

In realistic settings, the screening step of PS$^2$ in Procedure \ref{proc:screening} with Lasso may not necessarily perform well due to the existence of strong factors, and should be modified. To understand the problems, we assume that the $N$ excess returns $\br_t$ are generated by a factor model throughout this subsection:
\begin{align}\label{model:fact000}
\br_t=\bA\bx_t+\bu_t,
\end{align}
where $\bA\bx_t$, formed by $K$ (observable/latent) factors $\bx_t$ with mean $\bmu_x=\E \bx_t$ and factor loadings $\bA$, is the signal component that explains most of the variation in $\br_t$, while $\bu_t$ with mean $\bmu_u=\E \bu_t$ collects the remaining drivers as a noise component. As is widely accepted, we assume that at least one element of $\bx_t$ is a pervasive (globally strong) factor, like the market factor, whose corresponding column of $\bA$ is non-sparse, thereby affecting nearly all asset returns. By contrast, the noise component $\bu_t$ aggregates relatively weak factors 
that influence only subsets of returns together with idiosyncratic disturbances. 

Due to the existence of such strong factors, the PS$^2$ procedure fails to perform reliably. First, a strong factor makes the target $\bbeta(1) = \bOmega \bmu$ non-sparse, though its sparsity is crucial for any screening method to be effective. If $\bbeta(1)$ is sparse, procedures such as the Lasso will recover a low-dimensional set of discoveries. 
Given empirical evidence that $\bA$ is non-sparse and at least one pervasive factor has a nonzero mean, the resulting mean vector $\bmu=\bA\bmu_x+\bmu_u$ is typically non-sparse; hence the target $\bbeta(1)=\bOmega\bmu$ is also non-sparse.  
Second, a pervasive strong factor induces severe multicollinearity in $\br_t$, leading to inaccurate screening. In spite of a non-sparse target $\bbeta(1)$ in such a case as seen above, the Lasso estimate $\widehat{\bbeta}(1)$ tends to be sparse; this is because the Lasso treats highly redundant assets as substitutable, and unstably selects only a few. This is observed from an example in Supplementary Material (Section \ref{sim:multico}).


In light of the above discussion, does the concept of a ``sparse portfolio'' break down by the presence of strong factors? The answer is no. To resolve this issue, we consider a setting where investors can take direct positions not only in individual assets but also in the strong factors themselves (e.g., via index ETFs). In other words, we consider an \textit{augmented} portfolio problem that adds $K$-dimensional vector of \textit{investable} factor returns $\bx_t$ to the $N$-dimensional vector of original asset returns $\br_t$, and combine them as $\br_t^{\mathrm{aug}}=(\br_t',\,\bx_t')'$. The corresponding optimal portfolio weight vector is denoted as $\bw_{\mathrm{aug}}^*=(\bw_r^{*\prime},\bw_x^{*\prime})'$. This augmented weight vector is characterized as the next proposition. 

\begin{prop}\label{prop:aug}
Under the factor model in \eqref{model:fact000} with $\Cov(\bx_t, \bu_t) = \bzero$, we have
\begin{align}
    \bw_r^* = \frac{\bar{\rho}}{\theta_{\mathrm{aug}}} \bSigma_u^{-1} \bmu_u
    ~~~ \text{and}~~~
    \bw_x^* = \frac{\bar{\rho}}{\theta_{\mathrm{aug}}} (\bSigma_x^{-1} \bmu_x - \bA' \bSigma_u^{-1} \bmu_u), \label{eq:w*aug}
\end{align}
where $\theta_{\mathrm{aug}} = \bmu_u' \bSigma_u^{-1} \bmu_u + \bmu_x' \bSigma_x^{-1} \bmu_x$. 
\end{prop}

Perhaps surprisingly, once the augmented portfolio problem is considered, $\bw_r^*$ is proportional to the residual component $\bSigma_u^{-1} \bmu_u$, which can be sparse. This elegant decoupling property is consistent with the factor-portfolio decomposition in \cite{AoEtAl2019}. However, while they leverage this property to perform a one-step Lasso-based estimation procedure, we utilize it to design a highly stable two-step post-screening framework. That is, analogous to Proposition \ref{prop:1}, by treating $\bu_t$ as the data and defining $\bbeta_u = \bOmega_u \bmu_u$ with $\bOmega_u = (\E[\bu_t \bu_t'])^{-1}$, the support of $\bw_r^*$ can be obtained by examining that of $\bbeta_u$. Furthermore, $\bw_x^*$ is generally non-zero. Economically, this is required not only to capture the pure factor premium ($\bSigma_x^{-1} \bmu_x$) but also to hedge the factor exposures induced by the individual asset positions ($-\bA' \bSigma_u^{-1} \bmu_u$). Therefore, to achieve an efficient portfolio, investable factors should be included in the second-step portfolio construction alongside the screened assets. Motivated by these observations, we propose a modified procedure, \textit{PS$^2$ with factors} (FPS$^2$), that includes explicit defactoring and factor investing steps. 

\begin{proc}[FPS$^2$]\normalfont \label{proc:screening_mod}
Given $\bR_{\mathrm{aug}}=(\bR,\bX)$, we obtain $\what{\bw}_{\mathrm{aug}}$, an estimate of the optimal weight vector $\bw_{\mathrm{aug}}^*$, by Steps 1--3:
\begin{enumerate}[label=Step \arabic*., leftmargin=*]
\item[Step 1.] Make residuals $\what{\bU}=\bR-\bX\what{\bA}'$ with $\what{\bA}'=(\wtilde{\bX}'\wtilde{\bX})^{-1}\wtilde{\bX}'\bR$, where $\wtilde{\bX}$ is a column-wise demeaned factor matrix satisfying $\wtilde{\bX}'\biota_T=\bzero$. 
\item[Step 2.] Run Step 1 of Procedure \ref{proc:screening} with $\what\bU$ as data, and obtain $\what\cS$. 
\item[Step 3.] Obtain the optimal weight vector $\what{\bw}_{\mathrm{aug}}=(\hat{w}_1,\dots,\hat{w}_{N+K})'$ as follows:
\begin{enumerate}
   \item Set $\hat{w}_j=0$ for every $j\in\what\cS^c$.
    \item Estimate $\hat{w}_{j}$ for $j\in \what\cS \cup \{N+1,\dots,N+K\}$ by any method using only discovered assets in $\what\cS$ and the investable factors.
\end{enumerate}
\end{enumerate}
\end{proc}
Steps 1 and 3(b) correspond to defactoring and factor investing, respectively. The OLS estimator
$\what{\bA}$ with time-demeaned factor matrix $\wtilde{\bX}$, which is identical to the OLS estimator of regressing $\bR$ on $\bX$ with an intercept, is consistent for $\bA$. Importantly, we form defactored individual returns as $\what{\bU}=\bR-\bX\what{\bA}'$ rather than the conventional OLS residuals $\bR-\biota_T\what{\boldsymbol{\alpha}}'-\bX\what{\bA}'$ with $\what{\boldsymbol{\alpha}}$ estimated intercept,
so as not to remove the mean component $\bmu_u$ of the defactored returns.
This is essential because our subsequent screening step is designed to detect
assets with non-negligible average returns.

The weight estimation in Step 3(b) is similarly performed in low dimensions as demonstrated in Section \ref{subsubsec:post} since the number of 
factors $K$ as well as $|\what\cS|$ is usually small.

\section{Statistical Theory}\label{sec:theory}

This section establishes theoretical guarantees for our methodology. First, assuming the ideal case where defactoring operates without error (i.e., model \eqref{model:fact000} with $\bA=\bzero$), we discuss the theoretical validity of the Lasso screening and the post-Lasso OLS construction for PS$^2$ in Procedure \ref{proc:screening} in Theorems \ref{thm:lasso} and \ref{thm:post}, respectively. Furthermore, we provide sufficient conditions under which the target parameter becomes sparse in Theorem \ref{thm:when}. Finally, we provide justification for the defactoring step in FPS$^2$ of Procedure \ref{proc:screening_mod} in Theorem \ref{thm:defac}.

Motivated by the approximate factor structure in the arbitrage pricing theory (APT) literature \citep{Ross1976, Chamberlain1983}, we assume that the residuals, after controlling for pervasive risks ($\bA=\bzero$), constitute a weak factor process. Suppose the $N$-dimensional excess return vector $\br_t$ admits the factor structure: 
\begin{align}\label{model:fact}
\br_t = \bu_t = \bmu_u + \bB\bff_t + \be_t,
\end{align}
where $\bmu_u$ ($=\bmu$) is a mean vector of $\br_t$, $\bff_t$ is an $r$-dimensional vector of centered stochastic factors, 
$\bB$ is an $N\times r$ full-column rank matrix of deterministic factor loadings, and $\be_t$ is an $N$-dimensional vector of centered idiosyncratic errors independent of $\bff_t$. 
Let the (pre-centering) means of the factors and errors be $\bmu_f$ and $\bmu_e$, respectively. 
The model in \eqref{model:fact} is understood as the model in \eqref{model:fact000} with restricting $\bA=\bzero$ and specifying $\bu_t= \bmu_u + \bB\bff_t+\be_t$, where the signal is weak enough to rule out globally strong factors in $\bu_t$. In this case, the mean vector of $\br_t$ reduces to $\bmu=\bmu_u$ with $\bmu_u=\bB\bmu_f+\bmu_e$. 
To focus on non-trivial effects of the factors, we proceed under the assumption that the factors exist (i.e., $\mathbf{B} \neq \mathbf{0}$) throughout our analysis.

Our model in \eqref{model:fact} requires the boundedness of the maximal squared Sharpe ratio, $\theta=\bmu'\bSigma^{-1}\bmu$, a condition dictated by economic no-arbitrage constraints. If $\theta$ were to diverge as $N$ increases, it would imply the existence of asymptotic arbitrage opportunities, theoretically allowing investors to achieve infinite risk-adjusted returns \citep{Ross1976}. To this end, we assume the following condition. 
Denote by $\bSigma_f$ and $\bSigma_e$ the covariance matrices of $\bff_t$ and $\be_t$, respectively. 

\begin{con}\label{con:WFmodel}\normalfont
(a) $\|\bmu_f\|_2=O(1/\sqrt{\phi})$ and $c \leq \|\bmu_e\|_2\leq 1/c$; (b) $c\leq\lambda_{\min}(\bSigma_\bullet)\leq \lambda_{\max}(\bSigma_\bullet)\leq1/c$ holds for each $\bullet\in\{f,e\}$ for some constant $c>0$; (c) $\|\bB\|_2=O(\sqrt{\phi})$; (d) $\|\bB'\bSigma_e^{-1}\bmu_e\|_2=o(\sqrt{\phi})$. Here, $\phi$ is a non-decreasing positive sequence that can slowly diverge as $N \to \infty$. 
\end{con}

Condition \ref{con:WFmodel}(a)--(c) ensures consistency with the no-arbitrage principle of the APT. Specifically, we allow the factor loading $\|\mathbf{B}\|_2$ to diverge slowly, which captures residual local correlation structures that are not fully removed by the defactoring step. To maintain equilibrium, the expected factor return $\|\boldsymbol{\mu}_f\|_2$ must decay asymptotically to offset this growth. This decay reflects the economic intuition that local mispricing or localized risk premia dissipate as the market scale $N$ expands and efficiency improves. Consequently, the squared Sharpe ratio $\theta$ remains bounded, precluding asymptotic arbitrage opportunities. 

Complementing this, Condition \ref{con:WFmodel}(d) guarantees a strictly positive lower bound for $\theta$. While economic theory does not mandate the decay of $\theta$, this condition ensures that the weight coefficients remain of constant order, thereby avoiding unnecessary technical complexities in the discussion. 
Mathematically, this condition mitigates the structural overlap between the idiosyncratic alpha and the factor loadings by requiring their interaction (through $\boldsymbol{\Sigma}_e^{-1}$) to be asymptotically dominated by $\sqrt{\phi}$. While this formulation is flexible enough to permit mild alignment (or collinearity), it effectively ensures that $\boldsymbol{\mu}_e$ lies outside the column space of $\mathbf{B}$. Economically, this orthogonality is crucial for identifying ``pure'' alpha distinct from systematic factor risks. In the recent high-dimensional asset pricing literature, this restriction prevents the identified alpha from merely being a ``disguised'' risk premium associated with omitted factors \citep{FengEtAl2000, Giglio2021}. Furthermore, any excessive alignment with $\mathbf{B}$ would imply redundancy in the factor structure, as such a component should theoretically be spanned by the stochastic discount factor's loadings on the risk factors \citep{kozak2020}.

\begin{prop}\label{prop:SR}
Suppose that Condition \ref{con:WFmodel} holds. Then, for sufficiently large $N$, the squared Sharpe ratio, $\theta=\bmu'\bSigma^{-1}\bmu$, obtained under model \eqref{model:fact} is bounded from below and above by some positive constants.  
\end{prop}

This proposition shows that Condition \ref{con:WFmodel} is plausible in light of financial economics. Building on it, we develop theoretical guarantees for PS$^2$, adding further conditions as needed.

\subsection{Lasso-based procedures}\label{subsec:lasso-based}

We provide theoretical justification for Lasso-type asset screening and subsequent OLS-based portfolio construction. 
To this end, we assume the sparse optimal weights as well as some additional conditions required for the Lasso-based procedures. While this is a somewhat high-level assumption, Section \ref{subsec:whensp} provides conditions under which the optimal weight approximately inherits the support of the signal vector.

\begin{con}\label{con:sparse}\normalfont
(a) $\bw^*$ is $s$-sparse with $\cS=\supp(\bw^*)$, where $s$ can diverge with $\phi \leq s$ and $s\sqrt{\phi/T}\log N=O(1)$; (b) 
$\|\bB\|_1=O(\phi)$, $\|\bB\|_\infty=O(1)$, and $1/\lambda_{\min}(\bB'\bB)=O(1/\sqrt{\phi})$; (c) $c\leq \|\bSigma_e^{-1}\|_\infty \leq 1/c.$ 
\end{con}

Condition \ref{con:WFmodel} is sufficient for boundedness of $\|\boldsymbol{\mu}\|_2$ and $\|\mathbf{\Sigma}^{-1}\|_2$, directly implying boundedness of the optimal weight vector $\mathbf{w}^* = \bar{\rho}\mathbf{\Sigma}^{-1}\boldsymbol{\mu}/\theta$ in the $\ell_2$-norm. 
This together with Condition \ref{con:sparse}(a) further implies that the magnitude of individual weights decays on the order of $O(1/\sqrt{s})$, which reflects the effect of portfolio diversification. The lower bound of $s$ in (a) is reasonable since $\cS\approx \supp(\bB)\cup \supp(\bmu_e)$ with $\supp(\bB):=\cup_{k=1}^r\supp(\bb_k)$ while the upper bound of $s$ is a restriction for the Lasso working well. 
While the $\ell_2$ norm is constrained to be $O(\sqrt{\phi})$ in Condition \ref{con:WFmodel}(c), the additional $\ell_1$ restriction $\|\mathbf{B}\|_1 = O(\phi)$ in Condition \ref{con:sparse}(b) is critical. It eliminates dense structures---which would satisfy the $\ell_2$ bound but violate the $\ell_1$ limit---and ensures that the factors are distinguished by structural ``spikiness'' rather than diffuse noise. 
The last part of (b) imposes a factor identification condition by requiring the minimum eigenvalue of $\mathbf{B}'\mathbf{B}$ to diverge at a rate of at least $\sqrt{\phi}$. Since the maximum eigenvalue is $O(\phi)$, this formulation accommodates heterogeneous divergence rates. This setup aligns with the recent weak factor context, similar to the assumption in \cite{UY2023E}.

We define a class of random variables for specifying the factors and errors. A centered random variable $\xi$ is said to be sub-Gaussian (subG) with variance proxy $K^2$, denoted by $\xi\sim\text{subG}(K^2)$, if it satisfies 
\begin{align}
    \Pro(|\xi|\geq x) \leq 2\exp(-x^2/K^2)
\end{align}
for all $x\geq 0$ for some $K>0$. The tails of subG distributions decay at least as fast as those of a Gaussian distribution. For example, this class includes a Gaussian distribution and all the distributions with bounded supports. 

\begin{con}\label{con:dgp_fe}\normalfont
(a) $\bff_t$ and $\be_t$ are specified as 
$\ba_r'\bff_t,\, \ba_N'\be_t \sim \text{ind.\ subG}(K^2)$ for any $k$-dimensional vector $\ba_k\in\mathbb{R}^{k}$ such that $\|\ba_k\|_2=1$ and for some constant $K>0$; (b) $z_t\sim\text{ind.}\ \mathcal{N}(\alpha,\tau^2)$, where $c\leq \alpha\leq 1/c$ and $\tau=O(1)$. 
\end{con}
Condition \ref{con:dgp_fe}(a) imposes temporal independence but allows for heteroscedasticity for the factors and errors. We do not consider extensions to mixing processes here, as they would complicate the theory without changing the main implications. Additionally, the strict sub-Gaussian assumption is not essential for our main results, and could be relaxed to allow for heavier tails with minor modifications to the proofs. 
(b) concerns the response variable. Although theoretically it could be a constant, we introduce a perturbation to accommodate empirical constraints and ensure numerical stability in estimation procedures.

\subsubsection{Lasso-type screening}\label{subsec:ss}

For the screening step of PS$^2$, the following theorem establishes theoretical guarantees for Lasso.

\begin{thm}[Lasso screening]\label{thm:lasso}
Suppose that Conditions \ref{con:WFmodel}--\ref{con:dgp_fe} hold and the minimum signal satisfies
\begin{align}
\min_{j\in\cS}|\bbeta_{j}| 
\Bigg/ \sqrt{\frac{ \phi s\log N}{T}}\to \infty. \label{con:w-min}
\end{align}
If the regularization coefficient in the Lasso is set to be
\begin{align}
\lambda_{\textsf{L}} &\asymp \sqrt{\frac{\phi\log N}{T}}, \label{cond:lasso_lamb}
\end{align}
then $\cS\subseteq  \what\cS_{\textsf{L}}$ occurs with high probability. 
\end{thm}

Although our proofs leverage established techniques from high-dimensional statistics on the Lasso, the analysis presents two distinct challenges. First, unlike the conventional framework that assumes a true linear model generating the response, our Lasso screening involves the regression of an artificially generated response vector $\mathbf{z}$ on the returns $\mathbf{R}$. Since $\mathbf{z}$ is independent of $\mathbf{R}$ by construction, standard bounds relying on the residual process cannot be directly applied; instead, we must establish concentration inequalities for the empirical covariance structure around the population target $\boldsymbol{\beta}(\alpha)$. Second, the factor structure modeled in \eqref{model:fact} introduces nontrivial cross-sectional dependence. This dependency violates the standard restricted eigenvalue conditions used in i.i.d.\ settings, requiring us to account for the factor strength $\phi$ in the convergence rates. Consequently, the derived rates are generally slower than those in classical sparse models, reflecting the difficulty of disentangling signal from systematic noise.

\begin{rem}\normalfont
For completeness, we also establish an exact support recovery result for the adaptive Lasso screening under stronger signal and design conditions; this auxiliary result is reported in Section \ref{suppl:adaL} of Supplementary Material.
\end{rem}

\subsubsection{Post-Lasso OLS-type selection}

We construct the post-Lasso OLS estimator as $\what{\bw}=(\what{\bw}_{\what\cS}',\bzero_{N-\hat{s}}')'$, where $\what{\bw}_{\what\cS}$ is given by \eqref{post_L_OLS}.

\begin{thm}[Post-Lasso OLS selection]\label{thm:post}
Suppose that Conditions \ref{con:WFmodel}--\ref{con:dgp_fe} and $\phi s\sqrt{\log N / T} = o(1)$ hold. If $\what\cS_{\textsf{L}}$ in Theorem \ref{thm:lasso} is obtained, then the post-Lasso OLS estimator satisfies
\begin{align}
\|\what{\bw}-\bw^*\|_2
\lesssim \max\{\phi, \sqrt{s}\}\sqrt{\frac{s\log N}{T}}
\end{align}
with high probability.
\end{thm}

This theorem establishes the consistency of the post-Lasso OLS weight estimator under suitable conditions, assuming sufficient sparsity of the weights and weakness of the factors. The result first requires the upper bound of $\hat{s}=|\what\cS|$ to manage the random subset, $\what\cS$, which is indeed obtained as $\hat{s}\lesssim \phi s$. 
The presence of factors amplifies the estimation error for the active set size. To rigorously establish it, we rely on an $\epsilon$-net argument \citep[Ch.\ 4.2]{Vershynin2018} in the proof. The resulting error bound of $\what\bw$ is sufficiently small in spite of the original dimension of the weight vector, $N$. 

The statement of Theorem \ref{thm:post} requires the condition $\phi s\sqrt{\log N / T} = o(1)$, which is a bit stronger than Condition \ref{con:sparse}(a) required for the Lasso screening in Theorem \ref{thm:lasso}. This is essential to guarantee the asymptotic regularity of the sample 
Gram submatrix $\bR_{\hat{\mathcal{S}}}^{\prime}\bR_{\hat{\mathcal{S}}}/T$ in the post-Lasso OLS estimation.

\begin{cor}[Sharpe Ratio of the Post-Lasso OLS Portfolio]\label{cor:sr}
Suppose that the conditions of Theorem \ref{thm:post} hold. Then, the Sharpe ratio of the post-Lasso OLS portfolio satisfies
\begin{equation}
|\mathrm{SR}(\what{\bw}) - \mathrm{SR}(\bw^*)| 
\lesssim \max\{\phi^2, s\} \frac{\phi s \log N}{T}
\end{equation}
with high probability.
\end{cor}
Thus, the post-screening weight error is small enough to preserve the risk-adjusted performance of the optimal portfolio.

\subsection{Sparsity structure of the optimal weight}\label{subsec:whensp}

Thus far, the methodology and theory of PS$^2$ have been developed under the standing assumption that the optimal weight $\bw^*$ is sparse (Condition \ref{con:sparse}). Here, we investigate under what conditions this sparsity holds. Letting $\cA=\supp(\bmu)$ and $\supp(\bB)=\cup_{k=1}^r \supp(\bb_k)$, we naturally observe $\cA \subseteq \supp(\bB) \cup \supp(\bmu_e)$. We allow $|\cA|$ to slowly diverge as $N\to\infty$.

\begin{con}\label{con:when}\normalfont
(a) $1/\lambda_{\min}(\bB'\bB)=O(1/\sqrt{\phi})$; (b) the submatrix of $\boldsymbol{\Sigma}_e^{-1}$ restricted to the index set $\mathcal{A}^c \times \mathcal{A}$, which is denoted by $(\bSigma_e^{-1})_{\cA^c\cA}$, satisfies
(b1) $\|(\bSigma_e^{-1})_{\cA^c\cA}(\bmu_e)_{\cA}\|_{\infty}=o(1)$ and
(b2) $\max_{k\in[r]}\|(\bSigma_e^{-1})_{\cA^c\cA} (\bb_k)_{\cA}\|_{\infty}=O(1)$.
\end{con}

Condition \ref{con:when} guarantees the \textit{support heredity} of $\bw^*$ from $\bmu$.
To achieve this, part (b) imposes strong restrictions on the error precision matrix
$\bSigma_e^{-1}$, which governs the conditional dependence structure of the idiosyncratic
errors. While we allow for an approximate factor structure in the sense of
\citet{Chamberlain1983}, preserving sparsity requires controlling how signals propagate
through this precision matrix.
Specifically, (b1) limits the spread of alpha signals from $\cA$ to $\cA^c$ through the
idiosyncratic network. Without this restriction, optimization could assign substantial
weights to assets with no direct alpha, making the alpha effectively indistinguishable from
a pervasive systematic component \citep{Fama1993}. Similarly, (b2) controls the extent to
which the idiosyncratic network can mimic factor exposures. Collectively, these conditions
ensure that the source of the signal---whether alpha, common-factor exposure, or noise---is
structurally distinguishable in the population weights. The stricter $o(1)$ bound in (b1)
reflects the difficulty of preserving weak alpha signals, whereas the weaker $O(1)$ bound
in (b2) is sufficient for factor components, whose signals are stronger.

\begin{thm}[approximate support heredity]\label{thm:when}
Suppose that Conditions \ref{con:WFmodel} and \ref{con:when} hold. For any $\ep>0$, we have $\|\bw_{\cA^c}^*\|_\infty \leq \ep$ eventually. In particular, $\cS_{\ep}:=\{j\in[N]:|w_j^*|> \ep\} \subseteq  \cA$ holds eventually.
\end{thm}

This theorem shows that the support of $\bmu$ is approximately inherited by that of $\bw^*$.
In other words, if the signal vector $\bmu$ is sparse, then the optimal portfolio weight is
approximately sparse as well.

The degree of approximation is controlled by $\ep>0$, with $\cS_0=\cS$ in the exactly
sparse case. Although the theorem guarantees approximate rather than exact sparsity, this
is both theoretically robust and practically sufficient. In realistic financial markets with
transaction costs and minimum lot sizes, maintaining infinitesimal positions is economically
irrational. Indeed, sparse portfolio selection is often interpreted as a reduced-form way of
capturing proportional trading costs \citep{BrodieEtAl2009}. Any marginal benefit from
weights below $\ep$ would typically be dominated by trading frictions such as commissions
or bid--ask spreads. Thus, truncating such negligible weights to zero is not merely a
mathematical simplification but also a natural implication of implementable portfolio choice.

In this sense, Theorem \ref{thm:when} extends the intuition of \citet{treynor1973how} to a
high-dimensional factor-model setting. In the classical diagonal-residual case, portfolio
weights are driven directly by appraisal ratios such as $\mu_e/\sigma_e^2$. With correlated
residuals, however, optimal weights depend more generally on the precision matrix
$\bSigma_e^{-1}$, and off-diagonal dependence can generate additional hedging positions
\citep{stevens1998inverse}. Theorem \ref{thm:when} shows that, once these cross-asset
hedging effects are sufficiently weak, the optimal portfolio remains concentrated on the
assets carrying the underlying signals. Thus, even in a high-dimensional factor model, the
support of the optimal portfolio continues to reflect the support of economically relevant
signals.

\subsection{Justification of defactoring}

To facilitate the analysis, we restate the full model with investable factors:
\begin{align}
\br_t = \bA \bx_t + \bu_t,~~~~
\bu_t = \bmu_u + \bB\bff_t + \be_t. \label{model:full}
\end{align}
The corresponding matrix notation is the same as before. 
In Section \ref{subsec:lasso-based}, we studied Procedure \ref{proc:screening}, in particular its Lasso-based version, under the ideal case $\bA=\bzero$ (that is, $\bR=\bU$). 
In practice, however, we use Procedure \ref{proc:screening_mod}, which involves defactoring. 
Proposition \ref{prop:aug} shows that, in FPS$^2$, screening should target the residual-side raw-asset block rather than the full augmented weight vector. 
In this subsection, we show that the defactoring step has only a negligible effect on the screening stage. 
Under model \eqref{model:full}, let $\what{\bU}$ denote the residual matrix obtained in Step 1 of Procedure \ref{proc:screening_mod}, let $\what{\bbeta}_{\textsf{L}}^{\textsf{full}}$ be the Lasso estimator from regressing $\bz$ on $\what{\bU}$, let $\hat{\bv}:=\bz-\what{\bU}\what{\bbeta}_{\textsf{L}}^{\textsf{full}}$ be the corresponding residual vector, and let $\lambda_{\textsf{L}}^{\textsf{full}}$ be the regularization parameter. 
For simplicity, we treat the observable factors $\bx_t$ as nonrandom. Extending the result to random factors is straightforward.

\begin{con}\normalfont\label{con:obsx}
(a) $\|\bX/\sqrt{T}\|_2=O(1)$; \quad
(b) $1/\lambda_{\min}(\wtilde\bX'\wtilde\bX/T)=O(1)$. 
\end{con}

Our justification is based on the Lasso score
\begin{equation}
g(\bD):=\frac{1}{T}\bD'\what{\bv},
\end{equation}
defined for an arbitrary design matrix $\bD$. 
This score measures the empirical alignment between the residual vector $\what{\bv}$ and the columns of $\bD$. 
Since $\what{\bv}$ is obtained by solving the Lasso under the feasible design $\what{\bU}$, the Karush--Kuhn--Tucker (KKT) condition implies deterministically that
\[
\|g(\what{\bU})\|_\infty \le \lambda_{\textsf{L}}^{\textsf{full}}.
\]
The main issue is whether the same residual remains nearly orthogonal not only to the feasible design $\what{\bU}$ but also to the oracle design $\bU$. 
To study this, write the defactoring error as $\bDelta:=\what{\bU}-\bU$. Then we have $g(\bU)=g(\what{\bU})-g(\bDelta)$. Thus, even if the feasible score $g(\what{\bU})$ is small by the KKT condition, the oracle score $g(\bU)$ could still be large if the perturbation term $g(\bDelta)$ is not negligible. 
The next theorem shows that this does not happen.

\begin{thm}[KKT stability]\label{thm:defac}
Suppose that Conditions \ref{con:WFmodel}, \ref{con:dgp_fe}, and \ref{con:obsx} hold, and that $K$ is fixed. If $\lambda_{\normalfont{\textsf{L}}}^{\normalfont{\textsf{full}}} \asymp \sqrt{\phi \log N /T}$, then we have
\begin{align}\label{KKT_full_ideal}
\|g(\bU)\|_\infty \leq (1+c) \lambda_{\normalfont{\textsf{L}}}^{\normalfont{\textsf{full}}}
\end{align}
for some constant $c>0$ with high probability. 
\end{thm}

Theorem \ref{thm:defac} shows that, even when the residual vector $\what{\bv}$ is evaluated against the oracle design $\bU$, its score remains of the same order as the feasible KKT bound, up to a constant factor. Thus, relative to the benchmark analysis in Section \ref{subsec:lasso-based}, defactoring inflates the screening threshold only at the level of constants. Consequently, the KKT-score bounds underlying the screening results in Section \ref{subsec:lasso-based} carry over to Procedure~\ref{proc:screening_mod} with at most a change in constants. We stress, however, that this does not imply exact support invariance, which would require an additional margin condition. 
Rather, the theorem justifies transferring the rate-level screening guarantees to the feasible defactored design.

\section{Monte-Carlo Experiment}\label{MC}

We investigate the finite-sample performance of PS$^2$, FPS$^2$, and competing methods under four data-generating processes (DGPs). DGP1 and DGP2 do not contain strong common factors, whereas DGP3 and DGP4 incorporate strong factor structures; in addition, DGP2 and DGP4 include weak factors. We evaluate the methods from two perspectives:
(a) \emph{screening accuracy}, namely, how well the procedure recovers the small set of relevant assets from a large universe of mostly irrelevant ones; and (b) \emph{estimation accuracy}, namely, how accurately it estimates portfolio weights and the
implied Sharpe ratios under each DGP.

\subsection{Experimental design}\label{ED}

Consider an experimental design based on (\ref{model:full}), where asset returns admit two types of factor structures:
\begin{align} 
\br_t = \bA \bx_t + \bu_t,~~~~
\bu_t = \bmu_u + \bB\bff_t + \be_t, ~~~t=1,\dots,T, \label{sim1}
\end{align} 
where $ \bx_t \sim \text{ind.\ } \mathcal{N}\left( \bmu_x, \bI_{r_a} \right)$ and 
$ \bff_t \sim \text{ind.\ } \mathcal{N}\left( \bzero, \bI_{r_b} \right)$ are factors, $ \be_t \sim \text{ind.\ } \mathcal{N}\left( \bzero, \sigma^2_e  \bI_{N} \right)$ is an error term,  
and $ \bA $ and $ \bB $ are $ N \times r_a $ and $ N \times r_b $ factor loading matrices such that
$\bA^{\prime}\bA$ and $ \bB^{\prime} \bB$ are diagonal, respectively. We set $ \bmu_u = \bB \bmu_{f} + \bmu_e $. 
Then $ \E[\br_t] = \bmu = \bA \bmu_x + \bmu_u $ and $ \Var[\br_t] = \bSigma = \bA \bA^{\prime} + \bSigma_u $ with 
$ \bSigma_u = \bB \bB^{\prime} + \sigma^2_e \bI_N. $ Thus, the DGP allows asset returns to be driven by two common components,
$ \bA \bx_t$ and $ \bB \bff_t$, together with an idiosyncratic component $ \be_t$, while cross-sectional dependence is induced through $\bA \bA^{\prime}$ and $\bB \bB^{\prime}$. Here we interpret $ \bA \bx_t $ as a strong signal and $ \bB \bff_t $ as a weak signal. 

Let $ a^{*}_{ik} $ and $ b^{*}_{ik} $ be the $(i,k)$ elements of the $ N \times r_a $ and
$ N \times r_b $ orthogonal matrices obtained by applying the Gram--Schmidt orthogonalization
to $ \{a^0_{ik}\}$ and $ \{b^0_{ik}\}, $ respectively, where $ a^0_{ik} \sim \text{ind.\ } \mathcal{N}(0, 1) $
for $i=1,2,\dots, N $, and $ b^0_{ik} \sim \text{ind.\ } \mathcal{N}(0, 1) $ for
$ i = 1,2,\dots, n_k $, whereas $ b^0_{ik} = 0 $ for $ i=n_k + 1, n_k+2, \dots, N $, with
$ n_k = \left\lfloor N^{\alpha_k}+1/2\right\rfloor $ for $\alpha_k=(5-k+1)/10$.
Further, we assume $ N_1 < N $ and let $ \theta_0 = \bmu^{0\prime}_u \bSigma^{-1}_u \bmu^0_u, $
where $ \bmu^0_u $ is specified below. We consider the following four scenarios.

\begin{description}
\item[DGP1: Idiosyncratic returns only.]
$ \bA = \bzero, $ $ \bB = \bzero, $ and $ \bmu_u = \theta^{-1/2}_{0} \bmu^0_u, $
with $\bmu^0_u = \left( \bone'_{N_1}, \bzero_{N-N_1}' \right)'. $

\item[DGP2: Weak factors and idiosyncratic returns.]
$ \bA = \bzero, $ $ \bB = \{b^{*}_{ik}\}, $ and $ \bmu_u = \theta^{-1/2}_{0} \bmu^0_u, $
with $ \bmu^0_u = \bB \bmu^0_{f} + \bmu^0_e, $ $ \bmu^0_{f} = \bone_{r_b}, $ and
$\bmu^0_e = \left( \bone'_{N_1}, \bzero_{N-N_1}' \right)'. $

\item[DGP3: Strong factors and idiosyncratic returns.]
$ \bA = \{a^{*}_{ik}\}, $ $ \bB = \bzero, $ $\bmu_x = (\theta^0_x/r_a)^{1/2} \, \bone_{r_a}, $
and $ \bmu_u = (1-\theta^0_x)^{1/2}\theta_{0}^{-1/2} \bmu^0_u, $
with $ \bmu^0_u = \left( \bone'_{N_1}, \bzero_{N-N_1}' \right)'. $

\item[DGP4: Strong and weak factors with idiosyncratic returns.]
$ \bA = \{a^{*}_{ik}\}, $ $ \bB = \{b^{*}_{ik}\}, $ $\bmu_x = (\theta^0_x/r_a)^{1/2} \, \bone_{r_a}, $
and $ \bmu_u = (1-\theta^0_x)^{1/2} \theta^{-1/2}_{0} \bmu^0_u, $
with $ \bmu^0_u = \bB \bmu^0_{f} + \bmu^0_e, $ $ \bmu^0_{f} = \bone_{r_b}, $ and
$\bmu^0_e = \left( \bone'_{N_1}, \bzero_{N-N_1}' \right)'. $
\end{description}

DGP1 has the simplest structure: asset returns depend only on mutually independent
idiosyncratic components. Since both $ \bSigma $ and $ \bmu $ are sparse by construction, the true optimal portfolio weight $ \bw^* $ is also sparse, 
with $N_1$ nonzero entries ($s=N_1$). 
DGP2 allows returns to depend on weak common factors in addition to idiosyncratic
components. 
Let $ \bar{n} = \max\{n_1, N_1\} $. Then $ \bw^* $ in DGP2 is sparse with $ \bar{n} $ nonzero entries ($s=\bar{n}$), because
$ \bSigma^{-1} = \diag(\bG, \sigma_e^{-2}\bI_{N-n_1}) $, where $ \bG $ is an $n_1 \times n_1$ matrix whose entries are nonzeros, and $ \bmu $ is sparse with $ \bar{n} $ nonzero 
entries by construction. 
Note that $\bmu_u$ is constructed by dividing $\bmu_u^0$ by $\theta_0^{1/2}$ in both DGP1 and DGP2.
This condition ensures that the Sharpe ratio $ \sqrt{\theta} = \sqrt{\bmu'\bSigma^{-1}\bmu} $ is bounded away from both zero and infinity; in fact, $ \sqrt{\theta} = 1 $ in both DGP1 and DGP2. In addition, both DGP1 and DGP2 satisfy Conditions \ref{con:sparse} and \ref{con:dgp_fe}, so that the theoretical requirements for both the screening and estimation steps of 
the PS$^2$ procedure are met under these DGPs.

DGP3 and DGP4 augment these baseline models with a strong-factor structure, under which the FPS$^2$ procedure is designed to operate effectively. Specifically, 
DGP3 and DGP4 add strong common factors to DGP1 and DGP2, respectively. Thus, apart from the presence of strong factors, DGP3 and DGP4 inherit the same underlying 
structures as DGP1 and DGP2. However, due to the incorporation of factor investing in DGP3 and DGP4, the structure of the true optimal weight $\bw^*$ fundamentally 
differs from that in DGP1 and DGP2 (see Proposition \ref{prop:aug}). Furthermore, for DGP3 and DGP4, we employ an extra pre-specified parameter $ \theta^0_x $ to 
control the signals of the strong factors and the idiosyncratic components along with the weak factors. $ \theta^0_x $ corresponds to the pre-specified squared Sharpe 
ratio of $ \bx_t,$ since $ \bmu'_x \bSigma^{-1}_x \bmu_x = \theta^0_x $ by construction. For both DGP3 and DGP4, we restrict $ \theta^0_x $ to be $ 0 < \theta^0_x < 1$. 
This restriction ensures that the squared Sharpe ratio of the augmented portfolio, $\theta_{\mathrm{aug}}, $ is normalized to unity as in DGP1 and DGP2, 
regardless of the choice of $\theta^0_x$. It is because $ \theta_{\mathrm{aug}} = \theta^0_x + \theta_u $ (See Proposition \ref{prop:aug}) 
and $ \theta_u = 1 - \theta^0_x $ by construction in both DGP3 and DGP4.

\subsection{Portfolio construction} \label{sim:portfolio construction}

We consider four portfolios in this experiment: the (F)PS$^2$ portfolio; the MAXSER portfolio of \cite{AoEtAl2019} (see Section \ref{subsec:UR}); 
the QIS portfolio, a plug-in portfolio based on the Quadratic-Inverse Shrinkage estimator of \cite{LedoitWolf22}; and the Graphical Lasso (GL) 
portfolio, a plug-in portfolio based on the precision-matrix estimator of \cite{FriedmanEtAl07}.

The PS$^2$ portfolio is constructed according to Procedure \ref{proc:screening}, where
the screening step (Step 1) uses the Lasso to obtain $\what{\cS}$ and the estimation step
(Step 2) uses the post-screening OLS estimator in (\ref{post_L_OLS}).
To choose the Lasso regularization parameter, we use $k$-fold cross-validation based on
the out-of-sample prediction error of the post-screening OLS.
Specifically, $\what{\cS}$ is obtained from the Lasso with regularization parameter
$\lambda^*$ selected from a prespecified grid $\{\lambda_1,\dots,\lambda_M\}$ as follows:
\begin{enumerate}[label=Step \arabic*., leftmargin=*]
\item Split $\{1,2,\dots,T\}$ into $k$ subsamples $I_1,\dots,I_k$ such that
$\cup_{\ell=1}^k I_\ell=\{1,2,\dots,T\}$, and let
$I_\ell^c=\{1,2,\dots,T\}\setminus I_\ell$.

\item For each $\lambda_m\in\{\lambda_1,\dots,\lambda_M\}$ and each $\ell$, run the Lasso
of the pseudo-response $z_t\sim\mathcal{N}(\alpha,\tau^2)$ on $\br_t$ using the training sample
$t\in I_\ell^c$, and obtain the active set $\what{\cS}_\ell^c(\lambda_m)$.

\item For each $\lambda_m$ and each $\ell$, run the OLS regression of $z_t$ on
$\{r_{jt}: j\in\what{\cS}_\ell^c(\lambda_m)\}$ using the training sample $t\in I_\ell^c$,
and obtain the OLS estimator $\what{\bbeta}^{\mathrm{OLS}}_\ell(\lambda_m)$.

\item Choose $\lambda^*$ by
\begin{align*}
\lambda^*=\argmin_{\lambda_m\in\{\lambda_1,\dots,\lambda_M\}}
\sum_{\ell=1}^k \frac{1}{|I_\ell|}
\sum_{t\in I_\ell}
\left(
\alpha-\sum_{j\in\what{\cS}_\ell^c(\lambda_m)}
r_{jt}\hat{\beta}^{\mathrm{OLS}}_{j\ell}(\lambda_m)
\right)^2,
\end{align*}
where $\hat{\beta}^{\mathrm{OLS}}_{j\ell}(\lambda_m)$ denotes the $j$th element of
$\what{\bbeta}^{\mathrm{OLS}}_\ell(\lambda_m)$.

\item Run the Lasso of $z_t$ on $\br_t$ using the full sample $t=1,\dots,T$ with
regularization parameter $\lambda^*$, and obtain $\what{\cS}$.
\end{enumerate}

Note that the above procedure shares the same spirit as \cite{UematsuTanaka2018}, but differs from the conventional $k$-fold CV. 
Its goal is to choose the regularization parameter by minimizing the cross-validated mean squared prediction error of the post-screening 
OLS regression. By contrast, the conventional $k$-fold CV selects $\lambda^*$ by minimizing the cross-validated mean squared residuals of 
the Lasso in Step 2, and therefore does not incorporate the post-screening OLS steps in Steps 3 and 4.

The MAXSER portfolio in Section \ref{subsec:UR} requires a preliminary \textit{sub-pool selection}
step \citep{AoEtAl2019} whenever $T<N$. Specifically, we first draw 1,000 random subsets of
size $N_{\mathrm{sub}}$ (with $N_{\mathrm{sub}}<T$) from the $N$ available assets. For each subset,
we compute its sample Sharpe ratio. We then select the subset whose Sharpe ratio ranks at the
95th percentile among the 1,000 random draws (equivalently, the 50th highest value), and apply
MAXSER to that subset of size $N_{\mathrm{sub}}$ rather than to the full cross section of size $N$.

In addition, for comparability with PS$^2$, we implement MAXSER so that it targets a given
\textit{return}, rather than a given \textit{risk}, using the relation
$ \bar{\rho} = \bar{\sigma} \sqrt{\theta} $, where $\bar{\sigma}$ denotes the target risk.
The QIS and GL portfolios are constructed using the plug-in approach described in
Section \ref{subsubsec:plug}. For QIS, $\hat{\bSigma}^{-1}$ is given by the inverse of the
QIS covariance estimator of \cite{LedoitWolf22}; for GL, it is given by the precision-matrix
estimator obtained from the Graphical Lasso of \cite{FriedmanEtAl07}.

\subsection{Performance measures}
\label{PM}

We evaluate each method using five performance measures:
($i$) Nonzeros, the total number of selected variables;
($ii$) FDR, the false discovery rate of screening;
($iii$) Power, the screening power;
($iv$) MSE, the mean squared error of the portfolio-weight estimator; and
($v$) SR, the Sharpe ratio.
Measures ($i$)--($iii$) assess screening performance, measure ($iv$) assesses estimation accuracy,
and measure ($v$) summarizes overall portfolio performance.

More precisely, let $(rep)$ index the replication and let $\bar{R}$ denote the total number
of replications. For each procedure, Nonzeros is defined by
$\bar{R}^{-1}\sum_{rep=1}^{\bar{R}} |\what\cS^{(rep)}|$.
FDR and Power are defined by
$ \mathrm{FDR} = \bar{R}^{-1} \sum_{rep=1}^{\bar{R}} \mathrm{fdp}^{(rep)} $
and
$ \mathrm{Power} = \bar{R}^{-1} \sum_{rep=1}^{\bar{R}} \mathrm{pwr}^{(rep)} $,
respectively, where
$ \mathrm{fdp}^{(rep)} = | \what{\cS}^{(rep)} \cap \cS^c |/(|\what{\cS}^{(rep)}| \vee 1) $
and
$ \mathrm{pwr}^{(rep)} = |\what{\cS}^{(rep)} \cap \cS|/s $.
Further,
$ \mathrm{MSE} = \bar{R}^{-1}\sum_{rep=1}^{\bar{R}} \mathrm{mse}^{(rep)} $,
where
$ \mathrm{mse}^{(rep)} = \| \what{\bw}^{(rep)} - \bw^* \|^2_2 $.
Finally, SR is defined by
$ \mathrm{SR} = \bar{R}^{-1} \sum_{rep=1}^{\bar{R}} \mathrm{sr}^{(rep)} $,
where
$\mathrm{sr}^{(rep)} =
\bar{\br}^{(rep)\prime}_{\what{\cS}}
\what{\bw}^{(rep)}_{\what{\cS}}
/
\sqrt{
\what{\bw}^{(rep)\prime}_{\what{\cS}}
\what{\bSigma}^{(rep)}_{\what{\cS}}
\what{\bw}^{(rep)}_{\what{\cS}}
}$,
with $ \bar{\br}^{(rep)}_{\what{\cS}} $ and $ \what{\bSigma}^{(rep)}_{\what{\cS}} $
denoting the sample mean vector and sample covariance matrix of the selected variables,
respectively.
 
\subsection{Hyperparameter setup}

Throughout the simulations, we set $ r_a = r_b = 3 $, $ \sigma_e^2 = 2 $, $ \theta_x^0 = 0.01 $,
$ \bar{\rho} = 0.05 $, $ \alpha = 0.05, $ $ \tau = 10^{-4} $ and $ N_1 = \left\lfloor 1.415\sqrt{N} \right\rfloor. $
We set the sub-pool size for MAXSER to $ N_{\mathrm{sub}} = 100 $ and the total number of
replications to $ \bar{R} = 100 $.
The regularization parameter $ \lambda $ for MAXSER is selected by a conventional 10-fold 
cross-validation, whereas that for PS$^2$ is selected by the procedure described in Steps 1--5 of 
Section \ref{sim:portfolio construction} with $k=10$. 
While we initially set $M=100$, we restrict the candidate set to $\lambda_m$ that satisfy $\what{\cS}_\ell^c(\lambda_m) \leq T(k-2)/k $ 
to ensure the stability of the post-screening OLS steps; thus, the effective number of valid candidates is less than $M$.
By contrast, the regularization parameter for the Graphical Lasso is selected by the extended
Bayesian information criterion (EBIC) of \cite{FoygelDrton10} with tuning parameter 0.5.


\subsection{Simulation results}

Tables \ref{tab:simScreening1}--\ref{tab:simEstimation2} report (i) Nonzeros, (ii) FDR, (iii) Power, (iv) MSE, and (v) SR for the portfolio procedures under DGPs 1 to 4, 
for $ T \in \{200, 500, 1000, 2000\}$ and $ N \in \{200, 500, 1000 \}$. Given that $ \bar{n} = N_1$ in these settings, $ s= \left\lfloor 1.415\sqrt{N} \right\rfloor$ across 
all DGPs.

We begin with the screening properties of the estimators. The values of Nonzeros, FDR, and Power are reported in Table \ref{tab:simScreening1} for DGP1 and DGP2, and in Table \ref{tab:simScreening2} for DGP3 and DGP4. In the row for Nonzeros, standard errors of $ \hat{s} = |\what{\cS}| $ across replications are reported in parentheses. The results for QIS and GL are omitted from these tables, since both methods use all assets in portfolio construction and therefore satisfy $\hat{s}=N$ by definition.

For DGP1 and DGP2, PS$^2$ exhibits strong screening performance. In particular, the number of selected variables closely tracks the true sparsity level, and the empirical power approaches one as $T$ increases. Although PS$^2$ shows a slight tendency to under-select the true active set $\cS$
when $T$ is much smaller than $N$, this finite-sample effect weakens as $T$ grows. Moreover, the FDR of PS$^2$ remains reasonably controlled except when $T$ is very small relative to $N$. As expected, screening is more difficult and less stable in DGP2 than in DGP1.
By contrast, MAXSER delivers unstable screening performance across designs. It tends to over-select when $T>N$ and under-select when $T\le N$, leading to inflated FDR values in both regimes.
Taken together, these findings suggest that the screening behavior of MAXSER is substantially less reliable than that of PS$^2$.

Turning to DGP3 and DGP4, where strong factors are present, FPS$^2$ successfully screens relevant variables, controls FDR, and attains reasonable power, whereas the original PS$^2$ deteriorates because the screening step is contaminated by the strong common component. At the same time, echoing the contrast between DGP1 and DGP2, the additional weak-factor structure in DGP4 makes screening more challenging than in DGP3. 

Next, we examine the estimation properties. The values of MSE and SR are reported in
Table \ref{tab:simEstimation1} (for DGP1 and DGP2) and Table \ref{tab:simEstimation2}
(for DGP3 and DGP4). The MSEs are scaled by those of MAXSER in Table \ref{tab:simEstimation1}, and by those of the oracle estimator in Table \ref{tab:simEstimation2}, where the oracle estimator is defined as the FPS$^2$ estimator under perfect screening ($\hat{s}=s$). In Table \ref{tab:simEstimation2}, we
therefore focus on comparing FPS$^2$ with the oracle estimator.

For DGP1 and DGP2, the MSE and SR results indicate that PS$^2$ generally outperforms
MAXSER. In particular, the SR of MAXSER exhibits a substantial downward bias and is lower
than that of PS$^2$ when $T \leq N$, except in the case $T=200$. This suggests that the
sub-pool selection step does not work sufficiently well for estimating the optimal portfolio
in high dimensions: many truly relevant variables are excluded from the selected sub-pool,
which in turn leads to large estimation errors. As a result, the SR of MAXSER deviates more
substantially from its population value ($=1$) than that of PS$^2$. We also find that PS$^2$
tends to outperform QIS and GL in terms of both MSE and SR.\footnote{We observe that the MSE of QIS becomes particularly large when $N=T$. A plausible explanation is that, although QIS is asymptotically optimal, numerical instability in the boundary region $T/N\approx 1$ amplifies the variability of the plug-in weights.}

Turning to DGP3 and DGP4, we find that the MSE of FPS$^2$ is larger than that of the
oracle estimator, especially when $T$ is small relative to $N$. Nevertheless, this gap
shrinks as $T$ increases. Interestingly, the two estimators differ little in terms of SR.
These results suggest that, although the relative MSE ratios indicate a noticeable loss in
estimation precision for FPS$^2$, the resulting effect on portfolio performance is much
smaller. Finally, echoing the earlier findings, DGP4---which includes a weak-factor
structure---provides a more challenging estimation environment than DGP3.

In summary, the simulation results show that, when $T$ is reasonably large, the PS$^2$
framework---including FPS$^2$ in the presence of strong factors---delivers favorable
screening and estimation performance across all four DGPs, although performance becomes
weaker when weak factors are present or when $T$ is small relative to $N$.







\begin{table}[h!]

\vspace{-70pt}

 \centering

  \caption{Simulation Results for Screening Properties: DGP1 and DGP2}
  
\begin{adjustbox}{scale=0.65}

 \begin{threeparttable}

\vsp

    \begin{tabular}{rrrrrrrrr}
    \midrule
    \midrule
          &       &       &       &       &       &       &       &  \\
          &       & \multicolumn{1}{l}{DGP1} &       &       &       & \multicolumn{1}{l}{DGP2} &       &  \\
          &       &       &       &       &       &       &       &  \\
    \multicolumn{1}{l}{($i$) Nonzeros} &       &       &       &       &       &       &       &  \\
          &       &       & \multicolumn{1}{l}{PS$^2$} & \multicolumn{1}{l}{MAXSER} &       &       & \multicolumn{1}{l}{PS$^2$} & \multicolumn{1}{l}{MAXSER} \\
          & \multicolumn{1}{l}{$N=200$} &       &       &       &       &       &       &  \\
          & \multicolumn{1}{l}{ ($ s = 20 $)} & \multicolumn{1}{l}{$ T = 200 $} & 16.58 & 31.87 &       & \multicolumn{1}{l}{$ T = 200 $} & 12.78 & 29.97 \\
          &       &       & (7.70) & (10.17) &       &       & (8.37) & (10.79) \\
          &       & \multicolumn{1}{l}{$ T = 500 $} & 19.72 & 48.60 &       & \multicolumn{1}{l}{$ T = 500 $} & 16.90 & 54.64 \\
          &       &       & (2.26) & (11.41) &       &       & (6.08) & (13.53) \\
          &       & \multicolumn{1}{l}{$ T = 1000 $} & 20.31 & 46.55 &       & \multicolumn{1}{l}{$ T = 1000 $} & 19.23 & 51.68 \\
          &       &       & (1.13) & (11.73) &       &       & (2.83) & (12.24) \\
          &       & \multicolumn{1}{l}{$ T = 2000 $} & 20.43 & 47.20 &       & \multicolumn{1}{l}{$ T = 2000 $} & 18.65 & 49.26 \\
          &       &       & (1.10) & (11.61) &       &       & (1.37) & (11.29) \\
          & \multicolumn{1}{l}{$ N=500 $} &       &       &       &       &       &       &  \\
          & \multicolumn{1}{l}{ ($ s = 31 $)} & \multicolumn{1}{l}{$ T = 200 $} & 18.03 & 25.79 &       & \multicolumn{1}{l}{$ T = 200 $} & 10.63 & 21.62 \\
          &       &       & (13.50) & (10.61) &       &       & (9.59) & (11.12) \\
          &       & \multicolumn{1}{l}{$ T = 500 $} & 27.94 & 29.76 &       & \multicolumn{1}{l}{$ T = 500 $} & 24.38 & 28.11 \\
          &       &       & (5.71) & (9.41) &       &       & (9.66) & (9.74) \\
          &       & \multicolumn{1}{l}{$ T = 1000 $} & 30.97 & 81.45 &       & \multicolumn{1}{l}{$ T = 1000 $} & 31.12 & 101.50 \\
          &       &       & (2.02) & (17.24) &       &       & (7.05) & (20.21) \\
          &       & \multicolumn{1}{l}{$ T = 2000 $} & 31.18 & 78.53 &       & \multicolumn{1}{l}{$ T = 2000 $} & 31.24 & 94.74 \\
          &       &       & (0.59) & (19.52) &       &       & (2.93) & (18.28) \\
          & \multicolumn{1}{l}{$ N = 1000 $} &       &       &       &       &       &       &  \\
          & \multicolumn{1}{l}{ ($ s = 44 $)} & \multicolumn{1}{l}{$ T = 200 $} & 22.11 & 20.05 &       & \multicolumn{1}{l}{$ T = 200 $} & 8.19  & 17.39 \\
          &       &       & (34.93) & (12.90) &       &       & (11.90) & (11.91) \\
          &       & \multicolumn{1}{l}{$ T = 500 $} & 35.05 & 25.15 &       & \multicolumn{1}{l}{$ T = 500 $} & 22.55 & 22.23 \\
          &       &       & (14.76) & (9.41) &       &       & (10.73) & (10.04) \\
          &       & \multicolumn{1}{l}{$ T = 1000 $} & 41.46 & 24.98 &       & \multicolumn{1}{l}{$ T = 1000 $} & 35.82 & 22.62 \\
          &       &       & (4.37) & (9.37) &       &       & (7.48) & (9.23) \\
          &       & \multicolumn{1}{l}{$ T = 2000 $} & 43.84 & 119.93 &       & \multicolumn{1}{l}{$ T = 2000 $} & 36.73 & 128.71 \\
          &       &       & (1.25) & (21.54) &       &       & (3.66) & (26.74) \\
    \multicolumn{1}{l}{($ii$) FDR} &       &       &       &       &       &       &       &  \\
          &       &       & \multicolumn{1}{l}{PS$^2$} & \multicolumn{1}{l}{MAXSER} &       &       & \multicolumn{1}{l}{PS$^2$} & \multicolumn{1}{l}{MAXSER} \\
          & \multicolumn{1}{l}{$N=200$} &       &       &       &       &       &       &  \\
          &       & \multicolumn{1}{l}{$ T = 200 $} & 0.16  & 0.63  &       & \multicolumn{1}{l}{$ T = 200 $} & 0.15  & 0.62 \\
          &       & \multicolumn{1}{l}{$ T = 500 $} & 0.04  & 0.57  &       & \multicolumn{1}{l}{$ T = 500 $} & 0.12  & 0.65 \\
          &       & \multicolumn{1}{l}{$ T = 1000 $} & 0.02  & 0.54  &       & \multicolumn{1}{l}{$ T = 1000 $} & 0.06  & 0.61 \\
          &       & \multicolumn{1}{l}{$ T = 2000 $} & 0.02  & 0.55  &       & \multicolumn{1}{l}{$ T = 2000 $} & 0.02  & 0.60 \\
          & \multicolumn{1}{l}{$ N=500 $} &       &       &       &       &       &       &  \\
          &       & \multicolumn{1}{l}{$ T = 200 $} & 0.23  & 0.69  &       & \multicolumn{1}{l}{$ T = 200 $} & 0.14  & 0.64 \\
          &       & \multicolumn{1}{l}{$ T = 500 $} & 0.10  & 0.66  &       & \multicolumn{1}{l}{$ T = 500 $} & 0.17  & 0.68 \\
          &       & \multicolumn{1}{l}{$ T = 1000 $} & 0.03  & 0.60  &       & \multicolumn{1}{l}{$ T = 1000 $} & 0.13  & 0.69 \\
          &       & \multicolumn{1}{l}{$ T = 2000 $} & 0.01  & 0.58  &       & \multicolumn{1}{l}{$ T = 2000 $} & 0.04  & 0.66 \\
          & \multicolumn{1}{l}{$ N = 1000 $} &       &       &       &       &       &       &  \\
          &       & \multicolumn{1}{l}{$ T = 200 $} & 0.27  & 0.73  &       & \multicolumn{1}{l}{$ T = 200 $} & 0.14  & 0.66 \\
          &       & \multicolumn{1}{l}{$ T = 500 $} & 0.16  & 0.70  &       & \multicolumn{1}{l}{$ T = 500 $} & 0.15  & 0.71 \\
          &       & \multicolumn{1}{l}{$ T = 1000 $} & 0.06  & 0.67  &       & \multicolumn{1}{l}{$ T = 1000 $} & 0.12  & 0.69 \\
          &       & \multicolumn{1}{l}{$ T = 2000 $} & 0.01  & 0.62  &       & \multicolumn{1}{l}{$ T = 2000 $} & 0.04  & 0.68 \\
          &       &       &       &       &       &       &       &  \\
    \multicolumn{1}{l}{($iii$) Power} &       &       &       &       &       &       &       &  \\
          &       &       & \multicolumn{1}{l}{PS$^2$} & \multicolumn{1}{l}{MAXSER} &       &       & \multicolumn{1}{l}{PS$^2$} & \multicolumn{1}{l}{MAXSER} \\
          & \multicolumn{1}{l}{$N=200$} &       &       &       &       &       &       &  \\
          &       & \multicolumn{1}{l}{$ T = 200 $} & 0.64  & 0.55  &       & \multicolumn{1}{l}{$ T = 200 $} & 0.48  & 0.50 \\
          &       & \multicolumn{1}{l}{$ T = 500 $} & 0.94  & 1.00  &       & \multicolumn{1}{l}{$ T = 500 $} & 0.70  & 0.91 \\
          &       & \multicolumn{1}{l}{$ T = 1000 $} & 1.00  & 1.00  &       & \multicolumn{1}{l}{$ T = 1000 $} & 0.89  & 0.95 \\
          &       & \multicolumn{1}{l}{$ T = 2000 $} & 1.00  & 1.00  &       & \multicolumn{1}{l}{$ T = 2000 $} & 0.91  & 0.94 \\
          & \multicolumn{1}{l}{$ N=500 $} &       &       &       &       &       &       &  \\
          &       & \multicolumn{1}{l}{$ T = 200 $} & 0.38  & 0.22  &       & \multicolumn{1}{l}{$ T = 200 $} & 0.24  & 0.19 \\
          &       & \multicolumn{1}{l}{$ T = 500 $} & 0.80  & 0.29  &       & \multicolumn{1}{l}{$ T = 500 $} & 0.61  & 0.26 \\
          &       & \multicolumn{1}{l}{$ T = 1000 $} & 0.97  & 1.00  &       & \multicolumn{1}{l}{$ T = 1000 $} & 0.85  & 0.98 \\
          &       & \multicolumn{1}{l}{$ T = 2000 $} & 1.00  & 1.00  &       & \multicolumn{1}{l}{$ T = 2000 $} & 0.96  & 1.00 \\
          & \multicolumn{1}{l}{$ N = 1000 $} &       &       &       &       &       &       &  \\
          &       & \multicolumn{1}{l}{$ T = 200 $} & 0.21  & 0.08  &       & \multicolumn{1}{l}{$ T = 200 $} & 0.12  & 0.08 \\
          &       & \multicolumn{1}{l}{$ T = 500 $} & 0.63  & 0.15  &       & \multicolumn{1}{l}{$ T = 500 $} & 0.41  & 0.12 \\
          &       & \multicolumn{1}{l}{$ T = 1000 $} & 0.88  & 0.17  &       & \multicolumn{1}{l}{$ T = 1000 $} & 0.71  & 0.14 \\
          &       & \multicolumn{1}{l}{$ T = 2000 $} & 0.99  & 1.00  &       & \multicolumn{1}{l}{$ T = 2000 $} & 0.80  & 0.89 \\
          &       &       &       &       &       &       &       &  \\
    \bottomrule
    \bottomrule
    \end{tabular}%

\vsp 

\begin{tablenotes}
\item This table reports the estimated number of Nonzeros, false discovery rate (FDR), and Power for the PS$^2$ and MAXSER across 
varying $N$ and $T$ under DGP1 and DGP2. Standard errors for the nonzero estimates are reported in parentheses. 
\end{tablenotes}

\label{tab:simScreening1}%

\end{threeparttable}

\end{adjustbox}
 
\end{table}%

\begin{table}[h!]

\vspace{-70pt}

 \centering
 
  \caption{Simulation Results for Screening Properties: DGP3 and DGP4}

\begin{adjustbox}{scale=0.65}

 \begin{threeparttable}

\vsp

    \begin{tabular}{rrrrrrrrr}
    \midrule
    \midrule
          &       &       &       &       &       &       &       &  \\
          &       & \multicolumn{1}{l}{DGP3} &       &       &       & \multicolumn{1}{l}{DGP4} &       &  \\
          &       &       &       &       &       &       &       &  \\
    \multicolumn{1}{l}{($i$) Nonzeros} &       &       &       &       &       &       &       &  \\
          &       &       & \multicolumn{1}{l}{FPS$^2$} & \multicolumn{1}{l}{PS$^2$} &       &       & \multicolumn{1}{l}{FPS$^2$} & \multicolumn{1}{l}{PS$^2$} \\
          & \multicolumn{1}{l}{$N=200$} &       &       &       &       &       &       &  \\
          & \multicolumn{1}{l}{ ($ s = 20 $)} & \multicolumn{1}{l}{$ T = 200 $} & 16.40 & 20.37 &       & \multicolumn{1}{l}{$ T = 200 $} & 11.10 & 12.68 \\
          &       &       & (7.94) & (11.30) &       &       & (8.10) & (7.87) \\
          &       & \multicolumn{1}{l}{$ T = 500 $} & 19.90 & 29.64 &       & \multicolumn{1}{l}{$ T = 500 $} & 18.81 & 22.71 \\
          &       &       & (2.66) & (11.32) &       &       & (6.27) & (8.19) \\
          &       & \multicolumn{1}{l}{$ T = 1000 $} & 20.26 & 23.78 &       & \multicolumn{1}{l}{$ T = 1000 $} & 19.22 & 36.44 \\
          &       &       & (0.95) & (7.64) &       &       & (3.97) & (12.75) \\
          &       & \multicolumn{1}{l}{$ T = 2000 $} & 20.36 & 60.80 &       & \multicolumn{1}{l}{$ T = 2000 $} & 21.18 & 44.83 \\
          &       &       & (1.01) & (13.62) &       &       & (2.53) & (11.04) \\
          & \multicolumn{1}{l}{$ N=500 $} &       &       &       &       &       &       &  \\
          & \multicolumn{1}{l}{ ($ s = 31 $)} & \multicolumn{1}{l}{$ T = 200 $} & 23.20 & 24.13 &       & \multicolumn{1}{l}{$ T = 200 $} & 10.77 & 10.88 \\
          &       &       & (26.62) & (21.39) &       &       & (12.64) & (10.70) \\
          &       & \multicolumn{1}{l}{$ T = 500 $} & 27.15 & 34.54 &       & \multicolumn{1}{l}{$ T = 500 $} & 22.88 & 24.54 \\
          &       &       & (6.05) & (15.54) &       &       & (8.03) & (11.04) \\
          &       & \multicolumn{1}{l}{$ T = 1000 $} & 30.45 & 34.50 &       & \multicolumn{1}{l}{$ T = 1000 $} & 27.22 & 31.85 \\
          &       &       & (1.78) & (7.96) &       &       & (5.88) & (8.32) \\
          &       & \multicolumn{1}{l}{$ T = 2000 $} & 31.10 & 56.97 &       & \multicolumn{1}{l}{$ T = 2000 $} & 27.28 & 27.88 \\
          &       &       & (0.59) & (19.32) &       &       & (2.50) & (3.24) \\
          & \multicolumn{1}{l}{$ N = 1000 $} &       &       &       &       &       &       &  \\
          & \multicolumn{1}{l}{ ($ s = 44 $)} & \multicolumn{1}{l}{$ T = 200 $} & 34.33 & 18.11 &       & \multicolumn{1}{l}{$ T = 200 $} & 18.30 & 14.93 \\
          &       &       & (46.30) & (19.56) &       &       & (34.51) & (24.01) \\
          &       & \multicolumn{1}{l}{$ T = 500 $} & 35.69 & 40.81 &       & \multicolumn{1}{l}{$ T = 500 $} & 24.27 & 28.77 \\
          &       &       & (14.27) & (16.62) &       &       & (13.02) & (15.91) \\
          &       & \multicolumn{1}{l}{$ T = 1000 $} & 41.29 & 43.28 &       & \multicolumn{1}{l}{$ T = 1000 $} & 37.22 & 42.64 \\
          &       &       & (4.71) & (7.45) &       &       & (8.70) & (10.98) \\
          &       & \multicolumn{1}{l}{$ T = 2000 $} & 43.97 & 45.93 &       & \multicolumn{1}{l}{$ T = 2000 $} & 42.29 & 48.60 \\
          &       &       & (1.47) & (6.62) &       &       & (5.22) & (10.29) \\
          &       &       &       &       &       &       &       &  \\
    \multicolumn{1}{l}{($ii$) FDR} &       &       &       &       &       &       &       &  \\
          &       &       & \multicolumn{1}{l}{FPS$^2$} & \multicolumn{1}{l}{PS$^2$} &       &       & \multicolumn{1}{l}{FPS$^2$} & \multicolumn{1}{l}{PS$^2$} \\
          & \multicolumn{1}{l}{$N=200$} &       &       &       &       &       &       &  \\
          &       & \multicolumn{1}{l}{$ T = 200 $} & 0.17  & 0.29  &       & \multicolumn{1}{l}{$ T = 200 $} & 0.15  & 0.29 \\
          &       & \multicolumn{1}{l}{$ T = 500 $} & 0.05  & 0.33  &       & \multicolumn{1}{l}{$ T = 500 $} & 0.13  & 0.25 \\
          &       & \multicolumn{1}{l}{$ T = 1000 $} & 0.02  & 0.14  &       & \multicolumn{1}{l}{$ T = 1000 $} & 0.09  & 0.48 \\
          &       & \multicolumn{1}{l}{$ T = 2000 $} & 0.02  & 0.66  &       & \multicolumn{1}{l}{$ T = 2000 $} & 0.08  & 0.54 \\
          & \multicolumn{1}{l}{$ N=500 $} &       &       &       &       &       &       &  \\
          &       & \multicolumn{1}{l}{$ T = 200 $} & 0.26  & 0.35  &       & \multicolumn{1}{l}{$ T = 200 $} & 0.15  & 0.17 \\
          &       & \multicolumn{1}{l}{$ T = 500 $} & 0.09  & 0.30  &       & \multicolumn{1}{l}{$ T = 500 $} & 0.15  & 0.23 \\
          &       & \multicolumn{1}{l}{$ T = 1000 $} & 0.02  & 0.13  &       & \multicolumn{1}{l}{$ T = 1000 $} & 0.09  & 0.19 \\
          &       & \multicolumn{1}{l}{$ T = 2000 $} & 0.00  & 0.43  &       & \multicolumn{1}{l}{$ T = 2000 $} & 0.04  & 0.05 \\
          & \multicolumn{1}{l}{$ N = 1000 $} &       &       &       &       &       &       &  \\
          &       & \multicolumn{1}{l}{$ T = 200 $} & 0.35  & 0.39  &       & \multicolumn{1}{l}{$ T = 200 $} & 0.22  & 0.30 \\
          &       & \multicolumn{1}{l}{$ T = 500 $} & 0.18  & 0.34  &       & \multicolumn{1}{l}{$ T = 500 $} & 0.18  & 0.29 \\
          &       & \multicolumn{1}{l}{$ T = 1000 $} & 0.06  & 0.14  &       & \multicolumn{1}{l}{$ T = 1000 $} & 0.13  & 0.25 \\
          &       & \multicolumn{1}{l}{$ T = 2000 $} & 0.01  & 0.08  &       & \multicolumn{1}{l}{$ T = 2000 $} & 0.07  & 0.19 \\
          &       &       &       &       &       &       &       &  \\
    \multicolumn{1}{l}{($iii$) Power} &       &       &       &       &       &       &       &  \\
          &       &       & \multicolumn{1}{l}{FPS$^2$} & \multicolumn{1}{l}{PS$^2$} &       &       & \multicolumn{1}{l}{FPS$^2$} & \multicolumn{1}{l}{PS$^2$} \\
          & \multicolumn{1}{l}{$N=200$} &       &       &       &       &       &       &  \\
          &       & \multicolumn{1}{l}{$ T = 200 $} & 0.62  & 0.63  &       & \multicolumn{1}{l}{$ T = 200 $} & 0.41  & 0.39 \\
          &       & \multicolumn{1}{l}{$ T = 500 $} & 0.93  & 0.90  &       & \multicolumn{1}{l}{$ T = 500 $} & 0.78  & 0.79 \\
          &       & \multicolumn{1}{l}{$ T = 1000 $} & 1.00  & 0.96  &       & \multicolumn{1}{l}{$ T = 1000 $} & 0.85  & 0.85 \\
          &       & \multicolumn{1}{l}{$ T = 2000 $} & 1.00  & 0.98  &       & \multicolumn{1}{l}{$ T = 2000 $} & 0.97  & 0.97 \\
          & \multicolumn{1}{l}{$ N=500 $} &       &       &       &       &       &       &  \\
          &       & \multicolumn{1}{l}{$ T = 200 $} & 0.40  & 0.39  &       & \multicolumn{1}{l}{$ T = 200 $} & 0.23  & 0.23 \\
          &       & \multicolumn{1}{l}{$ T = 500 $} & 0.78  & 0.70  &       & \multicolumn{1}{l}{$ T = 500 $} & 0.60  & 0.56 \\
          &       & \multicolumn{1}{l}{$ T = 1000 $} & 0.96  & 0.93  &       & \multicolumn{1}{l}{$ T = 1000 $} & 0.78  & 0.80 \\
          &       & \multicolumn{1}{l}{$ T = 2000 $} & 1.00  & 0.95  &       & \multicolumn{1}{l}{$ T = 2000 $} & 0.84  & 0.85 \\
          & \multicolumn{1}{l}{$ N = 1000 $} &       &       &       &       &       &       &  \\
          &       & \multicolumn{1}{l}{$ T = 200 $} & 0.25  & 0.18  &       & \multicolumn{1}{l}{$ T = 200 $} & 0.14  & 0.14 \\
          &       & \multicolumn{1}{l}{$ T = 500 $} & 0.63  & 0.57  &       & \multicolumn{1}{l}{$ T = 500 $} & 0.41  & 0.41 \\
          &       & \multicolumn{1}{l}{$ T = 1000 $} & 0.88  & 0.83  &       & \multicolumn{1}{l}{$ T = 1000 $} & 0.72  & 0.70 \\
          &       & \multicolumn{1}{l}{$ T = 2000 $} & 0.99  & 0.94  &       & \multicolumn{1}{l}{$ T = 2000 $} & 0.89  & 0.87 \\
          &       &       &       &       &       &       &       &  \\
    \bottomrule
    \bottomrule
    \end{tabular}%

    \vsp 

\begin{tablenotes}
\item This table reports the estimated number of Nonzeros, false discovery rate (FDR), and Power for the FPS$^2$ and PS$^2$ across 
varying $N$ and $T$ under DGP3 and DGP4. Standard errors for the nonzero estimates are reported in parentheses. 
\end{tablenotes}

  \label{tab:simScreening2}%

  \end{threeparttable}

\end{adjustbox}

\end{table}%


\begin{table}[h!]

 \centering
  \caption{Simulation Results for Estimation Properties: DGP1 and DGP2}

\begin{adjustbox}{scale=0.65}

 \begin{threeparttable}
 
\vsp 
    
    \renewcommand{\arraystretch}{1.2}
    
    \begin{tabular}{rrrrrrrrrrrrr}
          &       &       &       &       &       &       &       &       &       &       &       &  \\
    \midrule
    \midrule
          &       &       &       &       &       &       &       &       &       &       &       &  \\
          &       & \multicolumn{1}{l}{DGP1} &       &       &       &       &       & \multicolumn{1}{l}{DGP2} &       &       &       &  \\
          &       &       &       &       &       &       &       &       &       &       &       &  \\
    \multicolumn{1}{l}{($iv$) MSE} &       &       &       &       &       &       &       &       &       &       &       &  \\
          &       &       & \multicolumn{1}{l}{$ \textsf{PS}^2 $ } & \multicolumn{1}{l}{MAXSER} & \multicolumn{1}{l}{QIS} & \multicolumn{1}{l}{GL} &       &       & \multicolumn{1}{l}{$ \textsf{PS}^2 $ } & \multicolumn{1}{l}{MAXSER} & \multicolumn{1}{l}{QIS} & \multicolumn{1}{l}{GL} \\
          & \multicolumn{1}{l}{$N=200$} &       &       &       &       &       &       &       &       &       &       &  \\
          &       & \multicolumn{1}{l}{$ T = 200 $} & 0.91  & 1.00  & 7.92  & 0.60  &       & \multicolumn{1}{l}{$ T = 200 $} & 1.13  & 1.00  & 10.62 & 0.68 \\
          &       & \multicolumn{1}{l}{$ T = 500 $} & 0.85  & 1.00  & 1.52  & 1.52  &       & \multicolumn{1}{l}{$ T = 500 $} & 1.34  & 1.00  & 1.24  & 1.51 \\
          &       & \multicolumn{1}{l}{$ T = 1000 $} & 0.57  & 1.00  & 1.89  & 1.89  &       & \multicolumn{1}{l}{$ T = 1000 $} & 0.79  & 1.00  & 1.87  & 2.06 \\
          &       & \multicolumn{1}{l}{$ T = 2000 $} & 0.57  & 1.00  & 2.15  & 2.15  &       & \multicolumn{1}{l}{$ T = 2000 $} & 0.52  & 1.00  & 2.32  & 2.41 \\
          & \multicolumn{1}{l}{$ N=500 $} &       &       &       &       &       &       &       &       &       &       &  \\
          &       & \multicolumn{1}{l}{$ T = 200 $} & 1.00  & 1.00  & 0.45  & 0.45  &       & \multicolumn{1}{l}{$ T = 200 $} & 1.11  & 1.00  & 0.52  & 0.57 \\
          &       & \multicolumn{1}{l}{$ T = 500 $} & 0.24  & 1.00  & 6.91  & 0.31  &       & \multicolumn{1}{l}{$ T = 500 $} & 0.53  & 1.00  & 12.03 & 0.46 \\
          &       & \multicolumn{1}{l}{$ T = 1000 $} & 0.71  & 1.00  & 2.16  & 2.16  &       & \multicolumn{1}{l}{$ T = 1000 $} & 1.20  & 1.00  & 1.66  & 1.95 \\
          &       & \multicolumn{1}{l}{$ T = 2000 $} & 0.47  & 1.00  & 2.81  & 2.81  &       & \multicolumn{1}{l}{$ T = 2000 $} & 0.71  & 1.00  & 2.44  & 2.80 \\
          & \multicolumn{1}{l}{$ N = 1000 $} &       &       &       &       &       &       &       &       &       &       &  \\
          &       & \multicolumn{1}{l}{$ T = 200 $} & 1.57  & 1.00  & 0.41  & 0.41  &       & \multicolumn{1}{l}{$ T = 200 $} & 1.55  & 1.00  & 0.45  & 0.48 \\
          &       & \multicolumn{1}{l}{$ T = 500 $} & 0.26  & 1.00  & 0.24  & 0.25  &       & \multicolumn{1}{l}{$ T = 500 $} & 0.42  & 1.00  & 0.30  & 0.34 \\
          &       & \multicolumn{1}{l}{$ T = 1000 $} & 0.08  & 1.00  & 6.20  & 0.17  &       & \multicolumn{1}{l}{$ T = 1000 $} & 0.17  & 1.00  & 8.90  & 0.27 \\
          &       & \multicolumn{1}{l}{$ T = 2000 $} & 0.54  & 1.00  & 3.08  & 3.07  &       & \multicolumn{1}{l}{$ T = 2000 $} & 0.88  & 1.00  & 2.71  & 2.93 \\
          &       &       &       &       &       &       &       &       &       &       &       &  \\
    \multicolumn{1}{l}{($v$) SR} &       &       &       &       &       &       &       &       &       &       &       &  \\
          &       &       & \multicolumn{1}{l}{$ \textsf{PS}^2 $ } & \multicolumn{1}{l}{MAXSER} & \multicolumn{1}{l}{QIS} & \multicolumn{1}{l}{GL} &       &       & \multicolumn{1}{l}{$ \textsf{PS}^2 $ } & \multicolumn{1}{l}{MAXSER} & \multicolumn{1}{l}{QIS} & \multicolumn{1}{l}{GL} \\
          & \multicolumn{1}{l}{$N=200$} &       &       &       &       &       &       &       &       &       &       &  \\
          &       & \multicolumn{1}{l}{$ T = 200 $} & 1.10  & 1.11  & 1.80  & 1.25  &       & \multicolumn{1}{l}{$ T = 200 $} & 1.05  & 1.12  & 1.83  & 1.25 \\
          &       & \multicolumn{1}{l}{$ T = 500 $} & 1.04  & 1.15  & 1.18  & 1.09  &       & \multicolumn{1}{l}{$ T = 500 $} & 0.99  & 1.14  & 1.24  & 1.09 \\
          &       & \multicolumn{1}{l}{$ T = 1000 $} & 1.02  & 1.07  & 1.10  & 1.03  &       & \multicolumn{1}{l}{$ T = 1000 $} & 1.02  & 1.07  & 1.13  & 1.05 \\
          &       & \multicolumn{1}{l}{$ T = 2000 $} & 1.01  & 1.03  & 1.05  & 1.00  &       & \multicolumn{1}{l}{$ T = 2000 $} & 1.01  & 1.03  & 1.07  & 1.05 \\
          & \multicolumn{1}{l}{$ N=500 $} &       &       &       &       &       &       &       &       &       &       &  \\
          &       & \multicolumn{1}{l}{$ T = 200 $} & 1.09  & 0.86  & 1.87  & 1.64  &       & \multicolumn{1}{l}{$ T = 200 $} & 0.89  & 0.82  & 2.00  & 1.53 \\
          &       & \multicolumn{1}{l}{$ T = 500 $} & 1.06  & 0.69  & 1.67  & 1.29  &       & \multicolumn{1}{l}{$ T = 500 $} & 0.99  & 0.71  & 1.70  & 1.28 \\
          &       & \multicolumn{1}{l}{$ T = 1000 $} & 1.03  & 1.13  & 1.22  & 1.14  &       & \multicolumn{1}{l}{$ T = 1000 $} & 1.01  & 1.14  & 1.27  & 1.12 \\
          &       & \multicolumn{1}{l}{$ T = 2000 $} & 1.02  & 1.06  & 1.12  & 1.06  &       & \multicolumn{1}{l}{$ T = 2000 $} & 1.01  & 1.06  & 1.16  & 1.06 \\
          & \multicolumn{1}{l}{$ N = 1000 $} &       &       &       &       &       &       &       &       &       &       &  \\
          &       & \multicolumn{1}{l}{$ T = 200 $} & 1.38  & 0.69  & 2.45  & 2.13  &       & \multicolumn{1}{l}{$ T = 200 $} & 0.74  & 0.69  & 2.55  & 1.96 \\
          &       & \multicolumn{1}{l}{$ T = 500 $} & 1.09  & 0.55  & 1.74  & 1.58  &       & \multicolumn{1}{l}{$ T = 500 $} & 0.91  & 0.56  & 1.77  & 1.54 \\
          &       & \multicolumn{1}{l}{$ T = 1000 $} & 1.04  & 0.48  & 1.59  & 1.32  &       & \multicolumn{1}{l}{$ T = 1000 $} & 1.00  & 0.50  & 1.61  & 1.28 \\
          &       & \multicolumn{1}{l}{$ T = 2000 $} & 1.02  & 1.10  & 1.22  & 1.16  &       & \multicolumn{1}{l}{$ T = 2000 $} & 1.00  & 1.09  & 1.25  & 1.16 \\
          &       &       &       &       &       &       &       &       &       &       &       &  \\
    \bottomrule
    \bottomrule
    \end{tabular}%

\vsp 

\begin{tablenotes}
\item This table reports the mean squared error (MSE) and Sharpe ratio (SR) for the PS$^2$, MAXSER, QIS, and GL across different combinations 
of $N$ and $T$ under DGP1 and DGP2. The MSE values are normalized by those of the MAXSER. The true SR is set to unity.
\end{tablenotes}

  \label{tab:simEstimation1}%

\end{threeparttable}

\end{adjustbox}

\end{table}%

\begin{table}[h!]

 \centering
  \caption{Simulation Results for Estimation Properties: DGP3 and DGP4}

\begin{adjustbox}{scale=0.65}

 \begin{threeparttable}
 
\vsp 
    
    \renewcommand{\arraystretch}{1.2}
  
    \begin{tabular}{rrrrrrrrr}
          &       &       &       &       &       &       &       &  \\
    \midrule
    \midrule
          &       &       &       &       &       &       &       &  \\
          &       & \multicolumn{1}{l}{DGP3} &       &       &       & \multicolumn{1}{l}{DGP4} &       &  \\
          &       &       &       &       &       &       &       &  \\
    \multicolumn{1}{l}{($iv$) MSE} &       &       &       &       &       &       &       &  \\
          &       &       & \multicolumn{1}{l}{FPS$^2$} & \multicolumn{1}{l}{Oracle} &       &       & \multicolumn{1}{l}{FPS$^2$} & \multicolumn{1}{l}{Oracle} \\
          & \multicolumn{1}{l}{$N=200$} &       &       &       &       &       &       &  \\
          &       & \multicolumn{1}{l}{$ T = 200 $} & 3.17  & 1.00  &       & \multicolumn{1}{l}{$ T = 200 $} & 2.81  & 1.00 \\
          &       & \multicolumn{1}{l}{$ T = 500 $} & 1.89  & 1.00  &       & \multicolumn{1}{l}{$ T = 500 $} & 3.28  & 1.00 \\
          &       & \multicolumn{1}{l}{$ T = 1000 $} & 1.23  & 1.00  &       & \multicolumn{1}{l}{$ T = 1000 $} & 3.83  & 1.00 \\
          &       & \multicolumn{1}{l}{$ T = 2000 $} & 1.18  & 1.00  &       & \multicolumn{1}{l}{$ T = 2000 $} & 2.03  & 1.00 \\
          & \multicolumn{1}{l}{$ N=500 $} &       &       &       &       &       &       &  \\
          &       & \multicolumn{1}{l}{$ T = 200 $} & 34.44 & 1.00  &       & \multicolumn{1}{l}{$ T = 200 $} & 3.79  & 1.00 \\
          &       & \multicolumn{1}{l}{$ T = 500 $} & 2.96  & 1.00  &       & \multicolumn{1}{l}{$ T = 500 $} & 3.58  & 1.00 \\
          &       & \multicolumn{1}{l}{$ T = 1000 $} & 1.80  & 1.00  &       & \multicolumn{1}{l}{$ T = 1000 $} & 3.31  & 1.00 \\
          &       & \multicolumn{1}{l}{$ T = 2000 $} & 1.15  & 1.00  &       & \multicolumn{1}{l}{$ T = 2000 $} & 2.06  & 1.00 \\
          & \multicolumn{1}{l}{$ N = 1000 $} &       &       &       &       &       &       &  \\
          &       & \multicolumn{1}{l}{$ T = 200 $} & 5.80  & 1.00  &       & \multicolumn{1}{l}{$ T = 200 $} & 4.29  & 1.00 \\
          &       & \multicolumn{1}{l}{$ T = 500 $} & 3.32  & 1.00  &       & \multicolumn{1}{l}{$ T = 500 $} & 4.79  & 1.00 \\
          &       & \multicolumn{1}{l}{$ T = 1000 $} & 2.54  & 1.00  &       & \multicolumn{1}{l}{$ T = 1000 $} & 3.31  & 1.00 \\
          &       & \multicolumn{1}{l}{$ T = 2000 $} & 1.29  & 1.00  &       & \multicolumn{1}{l}{$ T = 2000 $} & 3.50  & 1.00 \\
          &       &       &       &       &       &       &       &  \\
    \multicolumn{1}{l}{($v$) SR} &       &       &       &       &       &       &       &  \\
          &       &       &       &       &       &       &       &  \\
          &       &       & \multicolumn{1}{l}{FPS$^2$} & \multicolumn{1}{l}{Oracle} &       &       & \multicolumn{1}{l}{FPS$^2$} & \multicolumn{1}{l}{Oracle} \\
          & \multicolumn{1}{l}{$N=200$} &       &       &       &       &       &       &  \\
          &       & \multicolumn{1}{l}{$ T = 200 $} & 1.12  & 1.11  &       & \multicolumn{1}{l}{$ T = 200 $} & 1.02  & 1.12 \\
          &       & \multicolumn{1}{l}{$ T = 500 $} & 1.05  & 1.04  &       & \multicolumn{1}{l}{$ T = 500 $} & 1.04  & 1.05 \\
          &       & \multicolumn{1}{l}{$ T = 1000 $} & 1.03  & 1.02  &       & \multicolumn{1}{l}{$ T = 1000 $} & 1.02  & 1.02 \\
          &       & \multicolumn{1}{l}{$ T = 2000 $} & 1.01  & 1.01  &       & \multicolumn{1}{l}{$ T = 2000 $} & 1.01  & 1.01 \\
          & \multicolumn{1}{l}{$ N=500 $} &       &       &       &       &       &       &  \\
          &       & \multicolumn{1}{l}{$ T = 200 $} & 1.36  & 1.17  &       & \multicolumn{1}{l}{$ T = 200 $} & 0.93  & 1.18 \\
          &       & \multicolumn{1}{l}{$ T = 500 $} & 1.06  & 1.07  &       & \multicolumn{1}{l}{$ T = 500 $} & 1.02  & 1.07 \\
          &       & \multicolumn{1}{l}{$ T = 1000 $} & 1.03  & 1.03  &       & \multicolumn{1}{l}{$ T = 1000 $} & 1.01  & 1.03 \\
          & \multicolumn{1}{l}{$ N = 1000 $} &       &       &       &       &       &       &  \\
          &       & \multicolumn{1}{l}{$ T = 200 $} & 2.02  & 1.27  &       & \multicolumn{1}{l}{$ T = 200 $} & 1.27  & 1.27 \\
          &       & \multicolumn{1}{l}{$ T = 500 $} & 1.10  & 1.10  &       & \multicolumn{1}{l}{$ T = 500 $} & 0.94  & 1.09 \\
          &       & \multicolumn{1}{l}{$ T = 1000 $} & 1.05  & 1.05  &       & \multicolumn{1}{l}{$ T = 1000 $} & 1.02  & 1.04 \\
          &       & \multicolumn{1}{l}{$ T = 2000 $} & 1.02  & 1.02  &       & \multicolumn{1}{l}{$ T = 2000 $} & 0.99  & 1.02 \\
          &       &       &       &       &       &       &       &  \\
    \bottomrule
    \bottomrule
    \end{tabular}%

    \vsp 

\begin{tablenotes}
\item This table reports the mean squared error (MSE) and Sharpe ratio (SR) for the FPS$^2$ and the Oracle across different combinations 
of $N$ and $T$ under DGP3 and DGP4. The Oracle denotes the FPS$^2$ under perfect screening ($\hat{s}=s$), and the MSE values are normalized 
by those of the Oracle. The true SR is set to unity.
\end{tablenotes}

  \label{tab:simEstimation2}%

\end{threeparttable}

\end{adjustbox}

\end{table}%

\section{Empirical Application} \label{emp:mainsec}

We illustrate the empirical performance of FPS$^2$ using weekly returns on S\&P 500 constituents. Specifically, 
we construct portfolios based on FPS$^2$ and several competing procedures, and evaluate their out-of-sample performance in terms of the Sharpe ratio.

\subsection{Data}

We obtain weekly adjusted closing prices for current and historical S\&P 500 constituents from Bloomberg. 
Our sample runs from the first week of January 2000 to the third week of November 2025. 
We first compile the union of all tickers that appeared in the S\&P 500 during this period and then retrieve historical prices for this 
full superset of stocks. The adjusted closing prices, sampled on Fridays, are converted into weekly log returns.

While Bloomberg provides broad historical coverage, a small fraction of price histories could not be reliably retrieved, mainly because of identifier discontinuities associated with mergers, acquisitions, or major corporate restructurings. We therefore apply the following cleaning procedure to the log-return series: $(i)$ remove stocks with no valid observations; $(ii)$ treat price spells that remain constant for three or more consecutive weeks as missing
values, in order to filter out stale prices; $(iii)$ exclude stocks with internal gaps of five weeks or longer, including missing-value or zero-return spells; and $(iv)$ exclude stocks with more than 10\% missing or zero observations during their active period, or with a
return history shorter than two years.

After this cleaning step, the master dataset contains 1{,}350 weekly observations ($T=1350$) and 968 distinct tickers. The full list of tickers is reported 
in Section \ref{emp:stock list} of Supplementary Material. In the empirical portfolio analysis, however, the investable universe is
not fixed at 968 stocks; it varies over time according to the historical constituent list and the rolling-window data availability requirement described below.

We also obtain the Fama--French three factors (Mkt, SMB, and HML), called the FF3 factors, and the one-month Treasury bill rate from the Ken French Data Library.\footnote{\url{https://mba.tuck.dartmouth.edu/pages/faculty/ken.french/data_library.html}} Consistent with the FPS$^2$ framework, the FF3 factors are used both for defactoring and as investable factors in portfolio construction. The risk-free rate is subtracted from raw log returns, yielding weekly excess returns.

\subsection{Portfolio construction} \label{emp:portfolio construction}

We consider five portfolio construction procedures in the empirical analysis: FPS$^2$, which combines defactoring with factor investing; FMAXSER, which augments MAXSER with factor investing; QIS, the plug-in portfolio based on the Quadratic-Inverse Shrinkage
covariance estimator; FF3, the plug-in portfolio formed from the three Fama--French factors; and EW, the equal-weighted portfolio over the investable stock universe.

FPS$^2$ is the factor-augmented version of PS$^2$ studied in Section \ref{subsec:obsf}. Because strong factors are prominent in stock-return data, defactoring is likely to play a central role in practical portfolio construction. Figure \ref{fig:emp_cor2} illustrates this
point: the left and right panels display the correlation matrices of 462 stock excess returns in a 500-week sample before and after FF3-based defactoring, respectively. The strong positive dependence visible in the raw returns is substantially weakened after defactoring, suggesting that the screening step may become more stable in the residualized data. Factor investing is also essential in this setting, since strong common factors are empirically important and should therefore be treated as investable assets in portfolio construction.

We make one additional implementation choice for the (F)PS$^2$ procedures. At each portfolio formation date, screening is performed using the most recent $J$ observations available at that date, whereas the post-screening weight estimation uses only the most recent $\ell<J$ observations within the same window. This choice is intended to reduce the influence of structural breaks, while keeping the second-step estimation problem low-dimensional. It is also meant to mitigate the bias induced by using the same data in both
steps.

\begin{figure}[t!]
 \centering
 \includegraphics[keepaspectratio, scale=0.2, clip]{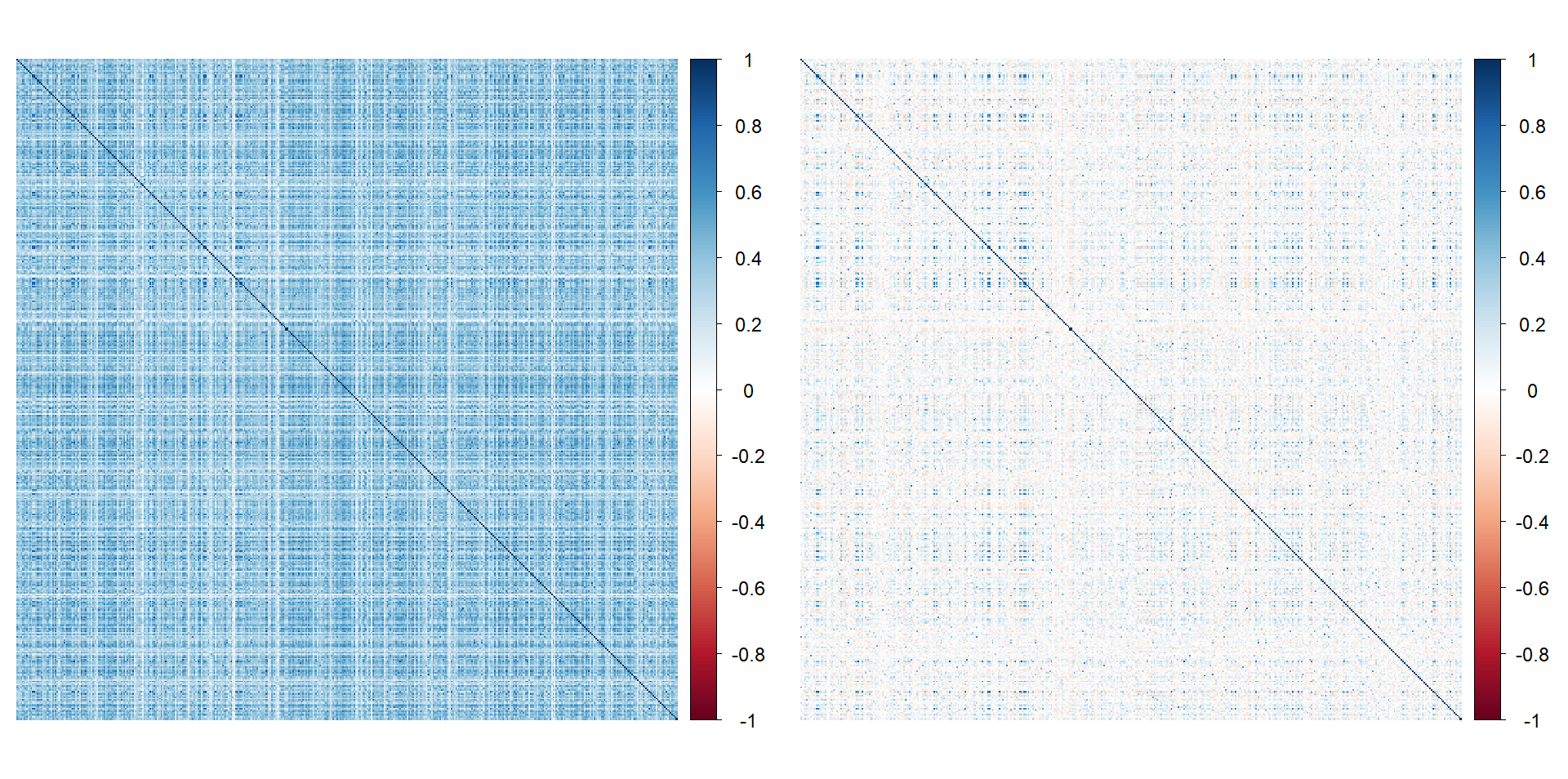}
\caption{Correlation Matrix before/after defactoring: Original (Left), defactored (Right) }\label{fig:emp_cor2}
\end{figure}

Specifically, we construct the FPS$^2$ portfolio as follows. Let $\bR^{(\ell)}_{\textsf{FPS}^2}$ denote the submatrix consisting of the most recent $\ell$ rows of $[\bR_{\what{\cS}_{\textsf{L}}}, \bX]$, where $\what{\cS}_{\textsf{L}}$ is the set selected by the Lasso and $\bX$ is formed from the FF3 factors. Then the FPS$^2$ portfolio, $\what{\bw}_{\textsf{FPS}^2}$, is constructed as
\begin{align}
\what{\bw}_{\textsf{FPS}^2}
=\dfrac{\bar{\rho}\bigl(1+\hat{\theta}^{(\ell)}_{\textsf{FPS}^2}\bigr)}
{\hat{\theta}^{(\ell)}_{\textsf{FPS}^2}}
\left(\bR^{(\ell)\prime}_{\textsf{FPS}^2}\bR^{(\ell)}_{\textsf{FPS}^2}\right)^{-1}
\bR^{(\ell)\prime}_{\textsf{FPS}^2}\bone_{\ell},
\end{align}
where $\hat{\theta}^{(\ell)}_{\textsf{FPS}^2}$ is the bias-corrected estimator of \citeauthor{KanZhou2007} (\citeyear{KanZhou2007}) based on $\hat{\theta}^{(\ell)}_{s,\textsf{FPS}^2}
= \what{\bmu}^{(\ell)\prime}_{\textsf{FPS}^2}
\what{\bSigma}^{(\ell)-1}_{\textsf{FPS}^2}
\what{\bmu}^{(\ell)}_{\textsf{FPS}^2}$ with $\what{\bmu}^{(\ell)}_{\textsf{FPS}^2}$ and $\what{\bSigma}^{(\ell)}_{\textsf{FPS}^2}$ denoting the sample mean vector and sample covariance matrix of $\bR^{(\ell)}_{\textsf{FPS}^2}$, respectively.

The FMAXSER \citep{AoEtAl2019} is an extension of MAXSER that allows for factor investing. Although FMAXSER is motivated by the same practical consideration as FPS$^2$, it performs screening and weight estimation jointly rather than sequentially. We construct the FMAXSER portfolio largely following \cite{AoEtAl2019}. To match our target-return formulation, we implement FMAXSER in terms of the target return $\bar{\rho}$, as in the construction of MAXSER in Section \ref{MC}, using the relation $\bar{\rho}=\bar{\sigma}\sqrt{\theta_{\mathrm{aug}}}$. Moreover, when $N+K$ exceeds the estimation sample size, we apply a sub-pool selection of size 50 to $\bR_{\mathrm{aug}}$.

The QIS portfolio is the plug-in portfolio based on the Quadratic-Inverse Shrinkage (QIS) covariance estimator of \cite{LedoitWolf22}. It is constructed as $\what{\bw}_{\textsf{QIS}}=\bar{\rho}\,\hat{\theta}_{\textsf{QIS}}^{-1}\hat{\bSigma}^{-1}_{\textsf{QIS}}\hat{\bmu}_{\textsf{aug}}$, 
where $\hat{\bSigma}_{\textsf{QIS}}$ is the QIS covariance estimator based on
$\bR_{\textsf{aug}}$, and $\hat{\theta}_{\textsf{QIS}}
=\hat{\bmu}_{\textsf{aug}}'
\hat{\bSigma}^{-1}_{\textsf{QIS}}
\hat{\bmu}_{\textsf{aug}}$. By construction, $\what{\bw}_{\textsf{QIS}}$ has all nonzero entries, so the QIS portfolio invests in all $N+K$ assets.

The FF3 portfolio is the plug-in portfolio based only on the FF3 factors, defined by $\what{\bw}_{\textsf{FF3}}=\bar{\rho}\,\hat{\theta}_{x}^{-1}\what{\bSigma}_{x}^{-1}\what{\bmu}_{x}$ with $\hat{\theta}_{x}=\what{\bmu}_{x}'\what{\bSigma}_{x}^{-1}\what{\bmu}_{x}$. The equal-weighted (EW) portfolio \citep{DeMiguelEtAl09} allocates equal capital across
the $N$ stocks in the investable universe and is defined as $\what{\bw}_{\textsf{EW}}=(1/N)\bone_{N}$. The FF3 and EW portfolios serve as benchmark portfolios.

To avoid survivorship bias, we construct all portfolios from the investable universe at each date using the historical constituent list obtained from a 
public repository.\footnote{GitHub/fja05680/sp500} Specifically, let $C_t$ denote the set of S\&P 500 constituents from the master dataset at time $t$. 
An asset $i\in C_t$ is included in the portfolio problem if and only if it satisfies the data-availability requirement, namely, it has a complete return 
history over the preceding rolling screening window of length $J$ (i.e., from $t-J+1$ to $t$). Hence, the dimension $N$ varies over time,
ranging from 361 to 498 in this application. Screening is performed on this $J$-week window, whereas post-screening weight estimation uses only the most 
recent $\ell<J$ observations within the same window. Out-of-sample performance is then evaluated using the realized return at $t+1$.

\subsection{Evaluation strategy}

We evaluate one-step-ahead out-of-sample portfolio performance in a rolling-window design.
Given data of length $T$, portfolios are formed recursively and performance is assessed over
the remaining $T-J$ observations. In line with the portfolio-construction procedure in
Section \ref{emp:portfolio construction}, screening is based on a rolling window of length $J$, 
whereas post-screening
weight estimation uses only the most recent $\ell<J$ observations within that window.
Specifically, following \cite{CallotEtAl2021}, we compare both gross and net empirical
Sharpe ratios.

The gross empirical Sharpe ratio is defined as $\widehat{\mathrm{SR}}=\doublehat{\mu}/\doublehat{\sigma}$,
where $\doublehat{\mu}=(T-J)^{-1}\sum_{t=J+1}^{T}\what{\bw}'_{t-1}\br^{\mathrm{aug}}_t$ and
$\doublehat{\sigma}^2=(T-J)^{-1}\sum_{t=J+1}^{T}\left(\what{\bw}'_{t-1}\br^{\mathrm{aug}}_t-\doublehat{\mu}\right)^2$.
Here, $\what{\bw}_t=(\hat{w}_{1,t},\dots,\hat{w}_{N+K,t})'$ denotes the estimated portfolio
weights at time $t\in\{J,\dots,T-1\}$ obtained by a given method. However, we set $ \hat{w}_{1,t} = \hat{w}_{2,t} = \dots = \hat{w}_{N,t} = 0 $
for FF3 and  $ \hat{w}_{N+1,t} = \hat{w}_{N+2,t} = \dots = \hat{w}_{N+K,t} = 0 $ for EW, respectively.

The net empirical Sharpe ratio incorporates transaction costs. Following \cite{CallotEtAl2021},
we define the excess portfolio return with transaction cost $\tau_c$ at time $t$ by
$\doublehat{r}^{\mathrm{aug}}_{\tau_c,t}=\what{\bw}'_{t-1}\br^{\mathrm{aug}}_t-\tau_c(1+\what{\bw}'_{t-1}\br^{\mathrm{aug}}_t)\sum_{i=1}^{N}
|\hat{w}_{i,t}-\hat{w}^{+}_{i,t-1}|$,
where $\hat{w}^{+}_{i,t-1}=\hat{w}_{i,t-1}(1+r_{i,t}+r^f_t)/(1+\what{\bw}'_{t-1}\br^{\mathrm{aug}}_t+r^f_t)$,
with $r^{\mathrm{aug}}_{i,t}$ and $r^f_t$ denoting the $i$th element of $\br^{\mathrm{aug}}_t$ and the risk-free rate at
time $t$, respectively. We then define
$\doublehat{\mu}_{\tau_c}=(T-J)^{-1}\sum_{t=J+1}^{T}\doublehat{r}^{\mathrm{aug}}_{\tau_c,t}$ and
$\doublehat{\sigma}_{\tau_c}^2=(T-J)^{-1}\sum_{t=J+1}^{T}
\left(\doublehat{r}^{\mathrm{aug}}_{\tau_c,t}-\doublehat{\mu}_{\tau_c}\right)^2$.
The net empirical Sharpe ratio with transaction cost $\tau_c$ is
$\widehat{\mathrm{SR}}_{\tau_c}=\doublehat{\mu}_{\tau_c}/\doublehat{\sigma}_{\tau_c}$.

\subsection{Parameter setup}

We set the target return to $\bar{\rho}=(1+0.10)^{1/52}-1=0.001835$, which corresponds
to an annual return of 10\% under 52 weeks per year. In the screening step, we set
$\alpha=\bar{\rho}$ and $\tau=10^{-10}$.

For FMAXSER, the sub-pool size is set to $N_{\mathrm{sub}}=50$, following Section 3.2 of
\cite{AoEtAl2019}. For FPS$^2$, the second-step estimation window is set to
$\ell=50+|\what{\cS}_{\textsf{L}}|$, so as to preserve sufficient degrees of freedom for
low-dimensional estimation regardless of the number of screened variables.\footnote{In all
our empirical specifications, this choice satisfies $\ell<J$.}
The Lasso regularization parameter in FMAXSER is selected by 10-fold cross-validation, whereas for FPS$^2$ we use the selection procedure described in Section \ref{sim:portfolio construction}. 

We set the transaction-cost parameter to $\tau_c=0.001$ (10 bps), following
\cite{LeeSeregina24}. This magnitude is plausible for our large-cap stock universe and is
broadly consistent with empirical estimates of trading costs for large-cap equities
\citep[e.g.,][]{NovyMarxVelikov16,FrazziniEtAl14}.

\subsection{Empirical results}


\begin{table}[t!]

   \centering
  \caption{Empirical Sharpe Ratios for the Full Sample and Subperiods}

\begin{adjustbox}{scale=0.65}

 \begin{threeparttable}

  \renewcommand{\arraystretch}{1.8}
    \begin{tabular}{rrrrrrrrlrrrrr}
          &       &       &       &       &       &       &       &       &       &       &       &       &  \\
    \multicolumn{6}{c}{($i$) $\widehat{\mathrm{SR}} $ in the overall out-of-sample period (-2025/11)} &       &       & \multicolumn{6}{c}{($iv$) $\widehat{\mathrm{SR}}_{\tau_c}$  in the overall out-of-sample period (-2025/11)} \\
          &       &       &       &       &       &       &       &       &       &       &       &       &  \\
          & \multicolumn{1}{l}{FPS$^2$} & \multicolumn{1}{l}{FM} & \multicolumn{1}{l}{QIS} & \multicolumn{1}{l}{FF3} & \multicolumn{1}{l}{EW} &       &       &       & \multicolumn{1}{l}{FPS$^2$} & \multicolumn{1}{l}{FM} & \multicolumn{1}{l}{QIS} & \multicolumn{1}{l}{FF3} & \multicolumn{1}{l}{EW} \\
\cmidrule{1-6}\cmidrule{9-14}    \multicolumn{1}{l}{$ J = 100 $} & 0.082 & 0.082 & 0.046 & 0.051 & 0.029 &       &       & $ J = 100 $ & 0.019 & 0.009 & 0.010 & 0.038 & 0.028 \\
    \multicolumn{1}{l}{$ J = 200 $} & 0.102 & 0.053 & 0.066 & -0.036 & 0.033 &       &       & [Turnover] & [0.826] & [0.412] & [0.108] & [0.403] & [0.027] \\
    \multicolumn{1}{l}{$ J = 300 $} & 0.100 & 0.082 & 0.073 & 0.042 & 0.030 &       &       & $ J = 200 $ & 0.042 & -0.031 & 0.031 & -0.047 & 0.032 \\
          &       &       &       &       &       &       &       & [Turnover] & [0.813] & [0.757] & [0.126] & [1.122] & [0.026] \\
    \multicolumn{6}{c}{($ii$) $\widehat{\mathrm{SR}} $ in the pre-pandemic period (-2019/12)} &       &       & $ J = 300 $ & 0.044 & -0.010 & 0.040 & 0.035 & 0.029 \\
          &       &       &       &       &       &       &       & [Turnover] & [0.852] & [1.042] & [0.144] & [0.414] & [0.026] \\
          & \multicolumn{1}{l}{FPS$^2$} & \multicolumn{1}{l}{FM} & \multicolumn{1}{l}{QIS} & \multicolumn{1}{l}{FF3} & \multicolumn{1}{l}{EW} &       &       &       &       &       &       &       &  \\
\cmidrule{1-6}    \multicolumn{1}{l}{$ J = 100 $} & 0.119 & 0.100 & 0.040 & 0.043 & 0.032 &       &       & \multicolumn{6}{c}{($v$) $\widehat{\mathrm{SR}}_{\tau_c}$ in the pre-pandemic period (-2019/12)} \\
    \multicolumn{1}{l}{$ J = 200 $} & 0.122 & 0.056 & 0.063 & -0.051 & 0.038 &       &       &       &       &       &       &       &  \\
    \multicolumn{1}{l}{$ J = 300 $} & 0.117 & 0.082 & 0.055 & 0.032 & 0.034 &       &       &       & \multicolumn{1}{l}{FPS$^2$} & \multicolumn{1}{l}{FM} & \multicolumn{1}{l}{QIS} & \multicolumn{1}{l}{FF3} & \multicolumn{1}{l}{EW} \\
\cmidrule{9-14}          &       &       &       &       &       &       &       & $ J = 100 $ & 0.050 & 0.012 & 0.004 & 0.030 & 0.031 \\
    \multicolumn{6}{c}{($iii$) $\widehat{\mathrm{SR}} $ in the post-pandemic period (2020/1-2025/11)} &       &       & [Turnover] & [0.844] & [0.446] & [0.104] & [0.426] & [0.026] \\
          &       &       &       &       &       &       &       & $ J = 200 $ & 0.057 & -0.041 & 0.028 & -0.063 & 0.037 \\
          & \multicolumn{1}{l}{FPS$^2$} & \multicolumn{1}{l}{FM} & \multicolumn{1}{l}{QIS} & \multicolumn{1}{l}{FF3} & \multicolumn{1}{l}{EW} &       &       & [Turnover] & [0.868] & [0.729] & [0.120] & [1.449] & [0.025] \\
\cmidrule{1-6}    \multicolumn{1}{l}{$ J = 100 $} & -0.009 & 0.047 & 0.060 & 0.082 & 0.021 &       &       & $ J = 300 $ & 0.052 & -0.015 & 0.020 & 0.023 & 0.033 \\
    \multicolumn{1}{l}{$ J = 200 $} & 0.049 & 0.054 & 0.072 & 0.051 & 0.021 &       &       & [Turnover] & [0.932] & [0.902] & [0.137] & [0.520] & [0.026] \\
    \multicolumn{1}{l}{$ J = 300 $} & 0.068 & 0.087 & 0.109 & 0.096 & 0.021 &       &       &       &       &       &       &       &  \\
          &       &       &       &       &       &       &       & \multicolumn{6}{c}{($vi$) $\widehat{\mathrm{SR}}_{\tau_c}$ in the post-pandemic period (2020/1-2025/11)} \\
          &       &       &       &       &       &       &       &       &       &       &       &       &  \\
          &       &       &       &       &       &       &       &       & \multicolumn{1}{l}{FPS$^2$} & \multicolumn{1}{l}{FM} & \multicolumn{1}{l}{QIS} & \multicolumn{1}{l}{FF3} & \multicolumn{1}{l}{EW} \\
\cmidrule{9-14}          &       &       &       &       &       &       &       & $ J = 100 $ & -0.057 & 0.004 & 0.025 & 0.069 & 0.020 \\
          &       &       &       &       &       &       &       & [Turnover] & [0.773] & [0.307] & [0.121] & [0.330] & [0.028] \\
          &       &       &       &       &       &       &       & $ J = 200 $ & 0.002 & -0.015 & 0.039 & 0.044 & 0.020 \\
          &       &       &       &       &       &       &       & [Turnover] & [0.664] & [0.833] & [0.145] & [0.229] & [0.028] \\
          &       &       &       &       &       &       &       & $ J = 300 $ & 0.030 & -0.003 & 0.077 & 0.092 & 0.020 \\
          &       &       &       &       &       &       &       & [Turnover] & [0.659] & [1.378] & [0.160] & [0.158] & [0.028] \\
    \end{tabular}%
  \vsp 

\begin{tablenotes}
\item This table reports the empirical Sharpe ratios for the FPS$^2$, FMAXSER (FM), QIS, FF3, and equally weighted (EW) portfolios for various 
rolling windows $J$ across the full out-of-sample period (–2025/12), the pre-pandemic period (–2019/12), and the post-pandemic period (2020/01–2025/11).
The left panel (Panels $i$–$iii$) presents the Sharpe ratios based on gross excess returns without transaction costs. 
The right panel (Panels $iv$–$vi$) reports the Sharpe ratios based on net excess returns, incorporating a transaction cost of 10 basis 
points per trade. Empirical turnover is reported in square brackets.
\end{tablenotes}

  \label{tab:meanrSR}%

\end{threeparttable}

\end{adjustbox}

\end{table}%

We evaluate the out-of-sample performance of five portfolio strategies—FPS$^2$, FMAXSER
(FM), QIS, FF3, and the equal-weighted (EW) portfolio—over rolling window lengths
$J \in \{100,200,300\}$. Table \ref{tab:meanrSR} reports the empirical gross Sharpe ratios
$\widehat{\mathrm{SR}}$, net Sharpe ratios $\widehat{\mathrm{SR}}_{\tau_c}$, and empirical
turnovers, where $\mathrm{Turnover}=(T-J)^{-1}\sum_{t=J+1}^{T}\sum_{i=1}^{N}
|\hat{w}_{i,t}-\hat{w}^{+}_{i,t-1}|$. To examine the stability of the strategies across different
market environments, we report results for the overall out-of-sample period (January 2000 to
November 2025, with evaluation length $1350-J$), the pre-pandemic period (January 2000 to
December 2019, with evaluation length $1042-J$), and the post-pandemic period (January 2020
to November 2025, with evaluation length $308$).

Panels ($i$) and ($ii$) of Table \ref{tab:meanrSR} report the gross Sharpe ratios for the overall and
pre-pandemic periods. In both periods, FPS$^2$ performs strongly relative to the competing
strategies and uniformly dominates FMAXSER over all reported window lengths. It also exceeds
the FF3 and the EW benchmarks throughout these two periods. These findings are consistent with the
interpretation that the two-step FPS$^2$ construction benefits from avoiding the ad hoc
sub-pool selection step in FMAXSER and from carrying out the final weight estimation only
after dimension reduction. More broadly, the results suggest that screening followed by
low-dimensional estimation can be effective in relatively stable market environments.

The performance ranking changes in the post-pandemic period; see Panel ($iii$). In this more
volatile regime, both screening-based procedures become less competitive, and QIS and FF3
often deliver higher gross Sharpe ratios. In particular, FPS$^2$ underperforms FMAXSER for
all three values of $J$ in this subsample. At the same time, FPS$^2$ is not uniformly dominated
by the passive benchmarks, since it still exceeds EW for $J=200$ and $J=300$. Overall, these
results are consistent with the view that large structural shifts can impair both screening and
post-screening estimation, especially for strategies that rely more heavily on cross-sectional
signal extraction.

Panels ($iv$)–($vi$) report the net Sharpe ratios. As expected, transaction costs materially reduce
the performance of the dynamic strategies. Even after accounting for trading frictions, FPS$^2$
continues to dominate FMAXSER in the overall and pre-pandemic periods, and it also remains
strong relative to QIS. In the pre-pandemic sample, FPS$^2$ continues to outperform all other
benchmarks. In the overall sample, however, its advantage is no longer uniform: for example,
FF3 and EW deliver slightly higher net Sharpe ratios than FPS$^2$ at $J=100$. In the
post-pandemic period, transaction costs further weaken dynamic strategies, and QIS and FF3
become more competitive than FPS$^2$. These net results therefore support a more qualified
conclusion: FPS$^2$ performs well after costs in relatively stable environments, but its
advantage becomes less robust when market conditions are highly unstable.

Our baseline analysis focuses on $J \leq 300$. With weekly data, windows longer than 300
weeks use more than five years of observations and are therefore likely to mix multiple market
regimes, which may distort the screening step. Since the purpose of screening is to identify
assets relevant for the current investment environment, we treat $J=500$ as a long-window
stress test and report the corresponding results in Section \ref{emp:add emp} of Supplementary Material. 
These additional results show that FPS$^2$ remains competitive in gross performance, but that its
net performance deteriorates further in the long-window design, as expected.

Section \ref{emp:add emp} of Supplementary Material also compares FPS$^2$ with the excess market return (Mkt).
In terms of the gross Sharpe ratio, FPS$^2$ outperforms the market benchmark in the overall
and pre-pandemic samples, but not in the post-pandemic period. In net terms, however, the
market benchmark dominates FPS$^2$ throughout. This pattern is naturally explained by the
fact that the market portfolio is passive and therefore does not incur the rebalancing costs that
penalize active strategies such as FPS$^2$.

\section{Conclusion}\label{sec:concl}

This paper proposes a new approach to high-dimensional mean--variance portfolio choice by separating asset screening from weight estimation. Our starting point is a support identity that reformulates portfolio construction as a support-recovery problem, which leads to the post-screening portfolio selection procedure, PS$^2$. In this framework, high-dimensional techniques are used only to screen relevant assets, while the final portfolio weights are estimated only after the problem has been reduced to a low-dimensional one.

We develop a comprehensive theoretical analysis for this approach in high dimensions. The results provide probabilistic guarantees for the screening step, error control for the post-screening portfolio estimator, and conditions under which approximate sparsity of the optimal portfolio is economically plausible. These results clarify why screening-based portfolio construction can be effective in high dimensions even when direct estimation of the full optimal weight vector is unstable.

We also show that strong pervasive factors fundamentally change the screening problem. When one works with the high-dimensional asset vector, strong common components can make the screening target effectively dense and induce severe multicollinearity. Our modified procedure, FPS$^2$, addresses this issue by combining defactoring with factor investing. In particular, once the augmented portfolio problem is considered, the optimal weight vector of the relevant assets is governed by the residual component, which provides a theoretical justification for residual-based screening and for including investable factors in the second-step portfolio construction of FPS$^2$.

Monte Carlo experiments and an empirical application to S\&P 500 constituents illustrate the practical relevance of the proposed framework. The results show that PS$^2$ yields stable post-screening portfolio construction when sparse screening is appropriate, while FPS$^2$ becomes essential in the presence of strong factor structures. In the empirical analysis, FPS$^2$ performs strongly relative to competing strategies in relatively stable market environments, although its advantage becomes less robust after transaction costs and during highly unstable periods.

Several extensions are left for future work. On the theoretical side, it would be useful to extend the analysis beyond temporally independent settings and to allow for heavier-tailed returns. On the empirical side, it would be of interest to investigate broader asset classes, alternative rebalancing frequencies, and richer forms of factor specification. We hope that the support-recovery perspective developed here provides a useful foundation for further work on high-dimensional portfolio construction.

\section*{Supplementary Material}
Section A: Proofs of the theoretical results in Sections \ref{ModelandEstimation} and \ref{sec:theory}. Section B: Lemmas for Section A and their proofs. 
Section C: Theory of exact recovery via adaptive Lasso screening. 
Section D: A simulation result of multicollinearity caused by a strong signal. 
Section E: Additional results and list of the stock tickers analyzed in Section \ref{emp:mainsec}.

\section*{Acknowledgments}

\bibliographystyle{chicago}
\bibliography{portfolio}

\clearpage
\appendix
\setcounter{page}{1}
\counterwithin{equation}{section}  

\setcounter{table}{0}
\renewcommand{\thetable}{E.\arabic{table}}

\renewcommand{\theequation}{\thesection.\arabic{equation}}  
\begin{center}
	{\Large Supplementary Material for} \\[7mm]
	\textbf{\Large Post-Screening Portfolio Selection} \\[10mm]
	\textsc{\large Yoshimasa Uematsu$^*$} {\large and} \textsc{\large Shinya Tanaka$^\dagger$} \\[5mm]
	$*$\textit{\large Department of Social Data Science, Hitotsubashi University} \\[1mm]
	$\dagger$\textit{\large Department of Economics, Otaru University of Commerce} 
\end{center}

\setcounter{section}{0}
\renewcommand{\thesection}{\Alph{section}}

\section{Proofs of the Theoretical Results}\label{append:proofs}

\subsection{Proof of Proposition \ref{prop:1}}

\begin{proof}
The covariance matrix of $\br_t$ is given by $\bSigma = \bOmega^{-1} - \bmu\bmu'$, where $\bOmega^{-1}=\bE[\br_t\br_t']$. Applying the Woodbury matrix identity gives
\begin{align}
\bOmega = (\bSigma + \bmu\bmu')^{-1}
= \bSigma^{-1} - \frac{\bSigma^{-1}\bmu\bmu'\bSigma^{-1}}{1+\bmu'\bSigma^{-1}\bmu}.
\end{align}
Recall $\theta=\bmu'\bSigma^{-1}\bmu$. Multiplying $\bmu$ from the right to the above equation and collecting the terms, we obtain
\begin{align}
\bSigma^{-1}\bmu = (1+\theta)\bOmega\bmu. 
\end{align}
Plugging this into the optimal weight, $\bw^*(\bar{\rho})\equiv (\bar{\rho}/\theta)\bSigma^{-1}\bmu$, immediately yields the first result. From this observation, the second result is trivial. 
\end{proof}

\subsection{Proof of Proposition \ref{prop:aug}}

\begin{proof}
For $\br_t^{\mathrm{aug}} = (\br_t', \bx_t')'$ with $\br_t = \bA \bx_t + \bu_t$ and the assumption $\Cov(\bx_t, \bu_t) = \bzero$, we have
\begin{align*}
    \bmu_{\mathrm{aug}} 
    &= 
    \begin{pmatrix} 
    \bA \bmu_x + \bmu_u \\ 
    \bmu_x 
    \end{pmatrix}, \\
    \bSigma_{\mathrm{aug}} 
    &= 
    \begin{pmatrix} 
    \bA \bSigma_x \bA' + \bSigma_u & \bA \bSigma_x \\ 
    \bSigma_x \bA' & \bSigma_x 
    \end{pmatrix}
    =
    \begin{pmatrix} 
    \bI_N & \bA \\ 
    \bzero & \bI_K 
    \end{pmatrix} 
    \begin{pmatrix} 
    \bSigma_u & \bzero \\ 
    \bzero & \bSigma_x 
    \end{pmatrix} 
    \begin{pmatrix} 
    \bI_N & \bzero \\ 
    \bA' & \bI_K 
    \end{pmatrix}.
\end{align*}
Note that $\bw_{\mathrm{aug}}^* = (\bar{\rho}/\theta_{\mathrm{aug}}) \bSigma_{\mathrm{aug}}^{-1} \bmu_{\mathrm{aug}}$ with $\theta_{\mathrm{aug}} = \bmu_{\mathrm{aug}}' \bSigma_{\mathrm{aug}}^{-1} \bmu_{\mathrm{aug}}$, where the inverse covariance (precision) matrix is obtained as
\begin{align*}
    \bSigma_{\mathrm{aug}}^{-1} 
    &= 
    \begin{pmatrix} 
    \bI_N & \bzero \\ -\bA' & \bI_K 
    \end{pmatrix} 
    \begin{pmatrix} 
    \bSigma_u^{-1} 
    & \bzero \\ \bzero & \bSigma_x^{-1} 
    \end{pmatrix} 
    \begin{pmatrix} 
    \bI_N & -\bA \\ 
    \bzero & \bI_K 
    \end{pmatrix} 
    = 
    \begin{pmatrix} 
    \bSigma_u^{-1} & -\bSigma_u^{-1}\bA \\ 
    -\bA'\bSigma_u^{-1} & \bSigma_x^{-1} + \bA'\bSigma_u^{-1}\bA 
    \end{pmatrix}.
\end{align*}
Multiplying $\bSigma_{\mathrm{aug}}^{-1}$ by $\bmu_{\mathrm{aug}}$ yields 
\begin{align*}
    \bSigma_{\mathrm{aug}}^{-1} \bmu_{\mathrm{aug}} 
    = 
    \begin{pmatrix} 
    \bSigma_u^{-1}\bmu_u \\ 
    \bSigma_x^{-1}\bmu_x - \bA'\bSigma_u^{-1}\bmu_u 
    \end{pmatrix}.
\end{align*}
Finally, we compute $\theta_{\mathrm{aug}} = \bmu_{\mathrm{aug}}' \bSigma_{\mathrm{aug}}^{-1} \bmu_{\mathrm{aug}}$ by a direct calculation, leading to 
\begin{align*}
    \theta_{\mathrm{aug}} 
    = \bmu_u' \bSigma_u^{-1} \bmu_u + \bmu_x' \bSigma_x^{-1} \bmu_x.
\end{align*}
Substituting $\bSigma_{\mathrm{aug}}^{-1} \bmu_{\mathrm{aug}}$ and $\theta_{\mathrm{aug}}$ into $\bw_{\mathrm{aug}}^*$ provides the target equations. This completes the proof.
\end{proof}

\subsection{Proof of Proposition \ref{prop:SR}}

\begin{proof}
First of all, the inverse matrices in the proof are well-defined under Condition \ref{con:WFmodel}(b). 

(i) Consider the upper bound of $\theta$. Note that $\bmu=\bB\bmu_f+\bmu_e$ and $\bSigma = \bB\bSigma_f\bB' +\bSigma_e$. By the property of the Rayleigh quotient and Weyl's inequality with Condition \ref{con:WFmodel}(a)--(c), the upper bound is obtained as
\begin{align*}
\theta \leq \frac{\|\bmu\|_2^2}{\lambda_{\min}(\bSigma)}
\lesssim \frac{\|\bB\|_2^2\|\bmu_f\|_2^2+\|\bmu_e\|_2^2}{\lambda_{\min}(\bB\bSigma_f\bB')+\lambda_{\min}(\bSigma_e)}
\lesssim 1. 
\end{align*}

(ii) Consider the lower bound of $\theta$. Let $\bM=(\bSigma_f^{-1} + \bB'\bSigma_e^{-1}\bB)^{-1}$ and $\mathbf{h}=\bB'\bSigma_e^{-1}\bmu_e$. By the Sherman–Morrison–Woodbury (SMW) formula, we have
\begin{align}
\theta &= \bmu' \bSigma^{-1}\bmu \\
&= \bmu_e'\bSigma_e^{-1}\bmu_e
+ \bmu_f'\bB' \bSigma^{-1}\bB\bmu_f 
+ \mathbf{h}'\bM\mathbf{h}
+ 2\mathbf{h}'\bmu_f + 2\mathbf{h}'\bM\bB'\bSigma_e^{-1}\bB\bmu_f \\
&\geq \frac{\|\bmu_e\|_2^2}{\lambda_{\max}(\bSigma_e)}
-2\|\mathbf{h}\|_2\|\bmu_f\|_2 \left(1 + \|\bM\bQ\|_2\right), \label{theta_low}
\end{align}
where $\mathbf{Q} = \mathbf{B}'\mathbf{\Sigma}_e^{-1}\mathbf{B}$ and the last inequality holds by the property of the Reighlay quotient and Cauchy-Schwarz inequality. Because $\|\bmu_e\|_2$ and $\lambda_{\max}(\bSigma_e)$ are lower- and upper-bounded by some positive constants under Condition \ref{con:WFmodel}(a) and (b), the first term in \eqref{theta_low} is bounded away from zero. 

Recall $\bM=(\bSigma_f^{-1}+\bQ)^{-1}$. By Condition \ref{con:WFmodel}(b), we have $\mathbf{Q} \prec \mathbf{\Sigma}_f^{-1} + \mathbf{Q}$ in the Loewner order. Applying a congruence transformation by multiplying both sides by $(\mathbf{\Sigma}_f^{-1} + \mathbf{Q})^{-1/2}$ from the left and right yields $(\mathbf{\Sigma}_f^{-1} + \mathbf{Q})^{-1/2} \mathbf{Q} (\mathbf{\Sigma}_f^{-1} + \mathbf{Q})^{-1/2} \prec \mathbf{I}$. 
Since the matrix on the left-hand side is symmetric and similar to $(\mathbf{\Sigma}_f^{-1} + \mathbf{Q})^{-1}\mathbf{Q} = \mathbf{M}\mathbf{Q}$, they share the same eigenvalues. Thus we obtain 
\begin{align}
\|\mathbf{M}\bQ\|_2 < 1.
\end{align}

Therefore, by $\|\bmu_f\|_2=O(1/\sqrt{\phi})$ and $\|\mathbf{h}\|_2=o(\sqrt{\phi})$ of Condition \ref{con:WFmodel}(a) and (d), we conclude from \eqref{theta_low} that $\theta\gtrsim 1$ eventually. This completes the proof. 
\end{proof}

\subsection{Proof of Theorem \ref{thm:lasso}}

\begin{proof}
Fix $\alpha\not=0$ and simply write $\bbeta(\alpha)$ as $\bbeta$. 
The KKT condition for the Lasso optimization in \eqref{feasibleLasso} is given by
\begin{align}\label{KKT}
\begin{cases}
\bR_j'(\bz-\bR\what{\bbeta}_\textsf{L})/T = \sgn(\hat{\beta}_j^\textsf{L})\lambda_\textsf{L}, & \hat{\beta}_j^\textsf{L}\not=0,\\
|\bR_j'(\bz-\bR\what{\bbeta}_\textsf{L})|/T \leq \lambda_\textsf{L}, & \hat{\beta}_j^\textsf{L}=0.
\end{cases}
\end{align}
Using $\bbeta=\alpha\bOmega\bmu$, we may equivalently rewrite \eqref{KKT}  as the single equation,
\begin{align}\label{KKT1}
\bR'\bz/T -\alpha \bmu &- \left(\bR'\bR/T-\bOmega^{-1}\right)(\what{\bbeta}_\textsf{L} - \bbeta) \\
&\qquad -\left(\bR'\bR/T-\bOmega^{-1}\right)\bbeta
- \bOmega^{-1}(\what{\bbeta}_\textsf{L} - \bbeta)
= \lambda_\textsf{L} \hat{\bkappa},
\end{align}
where $\hat{\bkappa}=(\hat{\kappa}_j)$ with $\hat{\kappa}_j=\sgn(\hat{\beta}_j^\textsf{L})$ for $\hat{\beta}_j^\textsf{L}\not=0$ and $\hat{\kappa}_j\in[-1,1]$ otherwise. Multiplying $\bOmega$ from the left to \eqref{KKT1} and arranging terms yield
\begin{align}
\what{\bbeta}_\textsf{L} - \bbeta
= \bOmega\bR'\bz/T-\bbeta 
- \left(\bOmega\bR'\bR/T-\bI\right)(\what{\bbeta}_\textsf{L}-\bbeta)
- \left(\bOmega\bR'\bR\bbeta/T-\bbeta \right) -\lambda_\textsf{L} \bOmega\hat{\bkappa}.
\end{align}
We take the $\ell_\infty$ norm and use the triangle inequality. Furthermore, by H\"older's inequality and Lemmas \ref{lem:concent_1}(a)--(c), \ref{lem:L_bound}, and \ref{lem:Omega}(a), we obtain
\begin{align}
\|\what{\bbeta}_\textsf{L} - \bbeta \|_\infty 
&\leq \|\bOmega\bR'\bz/T-\bbeta \|_\infty + \| \bOmega\bR'\bR/T-\bI\|_{\max} \|\what{\bbeta}_\textsf{L}-\bbeta\|_1 \\
&\qquad \qquad \qquad \qquad \qquad \qquad + \| \bOmega\bR'\bR\bbeta/T-\bbeta \|_\infty + \|\lambda_\textsf{L} \bOmega\hat{\bkappa}\|_\infty \\
&\lesssim  \sqrt{\frac{\phi\log N}{T}}
+ \frac{s\phi \log N}{T} + \sqrt{\frac{\phi^2 \log N}{T}}
+ \sqrt{\frac{\phi\log N}{T}}\|\bOmega\|_\infty \\
&\asymp \left(\sqrt{\frac{s\phi\log N}{T}} 
+ \sqrt{\frac{\phi}{s}} + 1\right)\sqrt{\frac{s\phi\log N}{T}} =: (i)
\end{align}
with probability at least $1-O(N^{-\nu})$ for any fixed constant $\nu>2$ determined in Lemma \ref{lem:concent_1}. 
We further bound (i) by $\phi/s\leq 1$ and $\sqrt{s\phi\log N /T}=o(1/\sqrt{s})$ in Condition \ref{con:sparse}(a), leading to
\begin{align}
(i)&\lesssim \sqrt{\frac{s\phi\log N}{T}}=:u_{NT}.\label{infty_norm}
\end{align}

The event, $\{\cS\subseteq \what\cS_\textsf{L}\}$, is equivalent to the event that $\hat{\beta}_j^\textsf{L}\not=0$ holds whenever $j\in\cS$. This together with the triangle inequality implies that
\begin{align}
\Pro\left( \cS\subseteq \what\cS_\textsf{L} \right) 
&= \Pro\left( \min_{j\in\cS}|\hat{\beta}_j^\textsf{L}| >0 \right) 
\geq \Pro\left( \min_{j\in\cS}\left\{|\beta_j|-|\hat{\beta}_j^\textsf{L}-\beta_j |\right\} >0 \right) \\
&\geq \Pro\left( \max_{j\in\cS}|\hat{\beta}_j^\textsf{L}-\beta_j | <\min_{j\in\cS}|\beta_j | \right)
\geq \Pro\left( \|\what{\bbeta}_\textsf{L}-\bbeta \|_\infty < \min_{j\in\cS}|\beta_j| \right).
\end{align}
Condition \eqref{con:w-min} and \eqref{infty_norm} give 
\begin{align}
\|\what{\bbeta}_\textsf{L}-\bbeta \|_\infty 
\lesssim u_n 
\ll \min_{j\in\cS}|\beta_j| 
\end{align}
with probability at least $1-O(N^{-\nu})$. This establishes the result. 
\end{proof}

\subsection{Proof of Theorem \ref{thm:post}}

\begin{proof}
Throughout the proof, we work on the event, $\{\cS\subseteq  \what\cS\}\cap\{\hat{s}\lesssim \bar{s}\}$, with $\bar{s}= \phi s$. Theorem \ref{thm:lasso} and Lemma \ref{lem:shat} establish that the event occurs with probability at least $1-O(N^{-\nu})$ for some $\nu>0$. 

We set $\bep=r_c\biota-\bR\bw^*$ with $r_c=\bar{\rho}(1+\theta)/\theta$ and $\hat{r}_c=\bar{\rho}(1+\hat{\theta})/\hat{\theta}$, where $\hat{\theta}$ is defined as a simple plug-in estimator consisting of the sample mean vector and inverse of the sample covariance matrix of screened data, $\bR_{\what\cS}$. On the event $\{\cS\subseteq  \what\cS\}$, since $\bw^*$ is sparse, we have $r_c\biota=\bR_{\what\cS}\bw_{\what\cS}^*+\bep$. 
The post-Lasso OLS estimator is then written as $\what{\bw}=(\what{\bw}_{\what\cS}',\bzero_{N-\hat{s}}')'$, where
\begin{align*}
\what{\bw}_{\what\cS}&=\hat{r}_c(\bR_{\what{\cS}}'\bR_{\what{\cS}})^{-1}\bR_{\what{\cS}}'\biota \\
&=\bw_{\what\cS}^*+(\bR_{\what{\cS}}'\bR_{\what{\cS}})^{-1}\bR_{\what{\cS}}'\bep + (\hat{r}_c-r_c)(\bR_{\what{\cS}}'\bR_{\what{\cS}})^{-1}\bR_{\what{\cS}}'\biota. 
\end{align*}
By the Cauchy-Schwarz inequality, we obtain 
\begin{align}
\|\what{\bw}_{\what\cS}-\bw_{\what\cS}^*\|_2
\leq \|(\bR_{\what{\cS}}'\bR_{\what{\cS}})^{-1}\|_2\|\bR_{\what{\cS}}'\bep\|_2+|\hat{r}_c-r_c|\|(\bR_{\what{\cS}}'\bR_{\what{\cS}})^{-1}\|_2\|\bR_{\what{\cS}}'\biota\|_2. \label{ineq:postl}
\end{align}

We evaluate each part using $\hat{s}\lesssim \bar{s}$. First, Weyl's inequality yields
\begin{align*}
\|(\bR_{\what{\cS}}'\bR_{\what{\cS}})^{-1}\|_2
&= \frac{1}{T}\frac{1}{\lambda_{\min}(\bR_{\what{\cS}}'\bR_{\what{\cS}}/T)} \\
&\leq \frac{1}{T}\frac{1}{\{\lambda_{\min}(\E[\bR'\bR/T]_{\what{\cS}\what{\cS}})-\|(\bR'\bR/T - \E[\bR'\bR/T])_{\what{\cS}\what{\cS}}\|_2}.
\end{align*}
By an application of the Cauchy interlacing theorem and the proof of Lemma \ref{lem:Omega_sp}, we have
\begin{align*}
\lambda_{\min}(\E[\bR'\bR/T]_{\what{\cS}\what{\cS}})
\geq \lambda_{\min}(\E[\bR'\bR/T])=\lambda_{\min}(\bmu\bmu'+\bB\bB'+\bSigma_e)\geq c.
\end{align*}
Moreover, we obtain
\begin{align*}
\|(\bR'\bR/T - \E[\bR'\bR/T])_{\what{\cS}\what{\cS}}\|_2
&\lesssim \bar{s} \|\bR'\bR/T - \E[\bR'\bR/T]\|_{\max}
\lesssim \bar{s} \sqrt{\frac{\log N}{T}}.
\end{align*}
Thus we have 
\begin{align*}
\|(\bR_{\what{\cS}}'\bR_{\what{\cS}})^{-1}\|_2
&\lesssim \frac{1}{T}\frac{1}{c-\bar{s}\sqrt{\log N/T}}\asymp \frac{1}{T},
\end{align*}
where we have used $\bar{s}\sqrt{\log N/T}=o(1)$. 
Similarly, Lemma \ref{lem:concent_1}(g) achieves
\begin{align*}
\|\bR_{\what{\cS}}'\bep/T\|_2
\lesssim \sqrt{\bar{s}}\|\bR'\bep/T\|_\infty 
\lesssim \sqrt{\bar{s}}(\sqrt{\phi}\wedge \sqrt{s})\sqrt{\frac{\log N}{T}}.
\end{align*}
Thus the first term of the upper bound in \eqref{ineq:postl} is at most 
\begin{equation*}
\|(\bR_{\what{\cS}}'\bR_{\what{\cS}})^{-1}\|_2\|\bR_{\what{\cS}}'\bep\|_2
\lesssim \sqrt{\bar{s}}(\sqrt{\phi}\wedge \sqrt{s})\sqrt{\frac{\log N}{T}}
\asymp \phi \sqrt{\frac{s\log N}{T}}
\end{equation*} 

For the second term in \eqref{ineq:postl}, since $\|(\bR_{\what{\cS}}'\bR_{\what{\cS}})^{-1}\|_2\asymp 1/T$ and $\|\bR_{\what{\cS}}'\biota\|_2 \asymp T$, we note that 
\begin{align}
|\hat{r}_c-r_c|\|(\bR_{\what{\cS}}'\bR_{\what{\cS}})^{-1}\|_2\|\bR_{\what{\cS}}'\biota\|_2 
&\lesssim 
|\hat{r}_c-r_c| \leq \frac{\bar{\rho}|\hat{\theta}-\theta|}{|\hat{\theta}\theta|} \notag \\
&\leq \frac{\bar{\rho}|\hat{\theta}-\theta|}{|\theta|}\frac{1}{\left||\theta|-|\hat{\theta}-\theta|\right|}. \label{2nd_term}
\end{align}
Proposition \ref{prop:SR} states that $|\theta|$ is of constant order. Thus, the upper bound in \eqref{2nd_term} is the same as $|\hat{\theta} - \theta|$ up to a positive constant factor, and Lemma \ref{lem:thetahat} directly gives the bound, 
\begin{equation*}
|\hat{\theta} - \theta| 
\lesssim s\sqrt{\frac{\log N}{T}}
\end{equation*}

Therefore, by combining the inequalities obtained so far, \eqref{ineq:postl} is bounded by 
\begin{equation*}
\phi\sqrt{\frac{s\log N}{T}}
+ s \sqrt{\frac{\log N}{T}} 
\asymp (\phi \vee \sqrt{s}) \sqrt{\frac{s\log N}{T}}.
\end{equation*} 
This completes the proof. 
\end{proof}

\subsection{Proof of Corollary \ref{cor:sr}}

\begin{proof}
Let $SR(\bw) = \bmu' \bw/\sqrt{\bw' \bSigma \bw}$ be the Sharpe ratio (SR) for a portfolio weight $\bw$. For the optimal weight $\bw^* = (\bar{\rho}/\theta)\bSigma^{-1}\bmu$, the maximum squared SR is $SR(\bw^*)^2 = \bmu' \bSigma^{-1} \bmu = \theta$. 

We first evaluate the difference in the squared SR. By definition, we have
\begin{align*}
    SR(\bw^*)^2 - SR(\widehat{\bw})^2 
    &= \theta - \frac{(\bmu' \widehat{\bw})^2}{\widehat{\bw}' \bSigma \widehat{\bw}} 
    = \frac{\theta \widehat{\bw}' \bSigma \widehat{\bw} - (\bmu' \widehat{\bw})^2}{\widehat{\bw}' \bSigma \widehat{\bw}}.
\end{align*}
From the optimality condition $\bSigma \bw^* = (\bar{\rho}/\theta) \bmu$, we can write $\bmu = (\theta/\bar{\rho}) \bSigma \bw^*$. Substituting this into the squared term yields $(\bmu' \widehat{\bw})^2 = \left\{ (\theta/\bar{\rho}) \bw^{*\prime} \bSigma \widehat{\bw} \right\}^2$. Furthermore, noting that $\bw^{*\prime} \bSigma \bw^* = (\bar{\rho}^2/\theta^2) \bmu' \bSigma^{-1} \bmu = \bar{\rho}^2/\theta$, we can replace $\theta^2/\bar{\rho}^2$ with $\theta/(\bw^{*\prime} \bSigma \bw^*)$. Thus, the numerator can be factored as
\begin{align*}
    \theta \widehat{\bw}' \bSigma \widehat{\bw} - (\bmu' \widehat{\bw})^2 
    &= \theta \left\{ \widehat{\bw}' \bSigma \widehat{\bw} - \frac{(\bw^{*\prime} \bSigma \widehat{\bw})^2}{\bw^{*\prime} \bSigma \bw^*} \right\}.
\end{align*}

Let $\bDelta = \widehat{\bw} - \bw^*$ be the estimation error. We expand the quadratic forms with respect to $\bDelta$:
\begin{align*}
    \widehat{\bw}' \bSigma \widehat{\bw} &= \bw^{*\prime} \bSigma \bw^* + 2\bw^{*\prime} \bSigma \bDelta + \bDelta' \bSigma \bDelta, \\
    (\bw^{*\prime} \bSigma \widehat{\bw})^2 &= (\bw^{*\prime} \bSigma \bw^* + \bw^{*\prime} \bSigma \bDelta)^2 = (\bw^{*\prime} \bSigma \bw^*)^2 + 2(\bw^{*\prime} \bSigma \bw^*)(\bw^{*\prime} \bSigma \bDelta) + (\bw^{*\prime} \bSigma \bDelta)^2.
\end{align*}
Substituting these expansions back into the curly bracketed term, we have
\begin{align*}
    \widehat{\bw}' \bSigma \widehat{\bw} - \frac{(\bw^{*\prime} \bSigma \widehat{\bw})^2}{\bw^{*\prime} \bSigma \bw^*} 
    = \bDelta' \bSigma \bDelta - \frac{(\bw^{*\prime} \bSigma \bDelta)^2}{\bw^{*\prime} \bSigma \bw^*}.
\end{align*}
By the Cauchy-Schwarz inequality, this term is non-negative and is bounded above by $\bDelta' \bSigma \bDelta$. Consequently, we obtain
\begin{align*}
    SR(\bw^*)^2 - SR(\widehat{\bw})^2 \le \frac{\theta}{\widehat{\bw}' \bSigma \widehat{\bw}} \bDelta' \bSigma \bDelta.
\end{align*}

To bound the absolute difference, we divide both sides by $SR(\bw^*) + SR(\widehat{\bw})$:
\begin{align*}
    |SR(\bw^*) - SR(\widehat{\bw})| \le \frac{\theta}{(SR(\bw^*) + SR(\widehat{\bw})) \widehat{\bw}' \bSigma \widehat{\bw}} \bDelta' \bSigma \bDelta.
\end{align*}
Since $\widehat{\bw} \xrightarrow{p} \bw^*$, the term $(SR(\bw^*) + SR(\widehat{\bw})) \widehat{\bw}' \bSigma \widehat{\bw}$ converges in probability to a strictly positive constant $2\sqrt{\theta} (\bar{\rho}^2 / \theta)$. Hence, the fractional multiplier is $O_p(1)$.

For the quadratic form, we have $\bDelta' \bSigma \bDelta \le \lambda_{\max}(\bSigma) \|\bDelta\|_2^2$. Under Condition 1 (or the factor model assumptions), $\lambda_{\max}(\bSigma) = O(\phi)$. We have also achieved the upper bound of $\|\bDelta\|_2$ by Theorem 2. We consequently obtain
\begin{align*}
    |SR(\widehat{\bw}) - SR(\bw^*)| = O_p\left( \phi (\phi \vee \sqrt{s})^2 \frac{s \log N}{T} \right),
\end{align*}
which completes the proof.
\end{proof}

\subsection{Proof of Theorem \ref{thm:when}}
\begin{proof}
	We will achieve the result by deriving an appropriate upper bound for $\|\bbeta(1)_{\cA^c} \|_\infty$ that converges to zero as $N\to \infty$ under the assumed conditions. First, we have $\E[\br_t\br_t']=\bmu\bmu'+\bB\bSigma_f\bB'+\bSigma_e$, so that the SMW formula gives 
	\begin{align}
		\bOmega 
		= \bTheta - \frac{\bTheta \bmu \bmu'\bTheta}{1+\bmu'\bTheta\bmu}
	\end{align}
    with $\bTheta=(\bB\bSigma_f\bB'+\bSigma_e)^{-1}$. 
	We then evaluate the $\ell_\infty$ norm of the subvector of $\bbeta(1)=\bOmega\bmu$ restricted on $\cA^c$, where $\cA=\supp(\bmu)$. We have 
\begin{align}
\|\bbeta(1)_{\cA^c} \|_\infty 
= \left\| \left(\bTheta \bmu - \frac{\bTheta \bmu \bmu'\bTheta\bmu}{1+\bmu'\bTheta\bmu}\right)_{\cA^c} \right\|_\infty 
= \frac{\|(\bTheta\bmu)_{\cA^c} \|_\infty}{1+\bmu'\bTheta\bmu}
\leq \|(\bTheta\bmu)_{\cA^c} \|_\infty, \label{ineq:Ome0mu1}
\end{align}
where the last inequality holds since $\bmu'\bTheta\bmu$ is nonnegative.

Derive an upper bound of $\|(\bTheta\bmu)_{\cA^c} \|_\infty$. By the definition of $\bmu$, we obtain 
\begin{align}
\bTheta \bmu = \bTheta\bB\bmu_f + \bTheta\bmu_e. \label{eq:Tm}
\end{align}
Letting $\bM=(\bSigma_f^{-1} + \bB'\bSigma_e^{-1}\bB)^{-1}$ and applying the SMW formula to $\bTheta$, we see that 
	\begin{align}
	\bTheta\bB &= \bSigma_e^{-1} \bB( \bI - \bM\bB'\bSigma_e^{-1}\bB) \\
	&= \bSigma_e^{-1} \bB\bM\bSigma_f^{-1}, \label{eq:TB}
\end{align}
which holds by the identity, $\bI-(\bU+\bV)^{-1}\bV=(\bU+\bV)^{-1}\bU$, and 
	\begin{align}
	\bTheta\bmu_e = \bSigma_e^{-1}\bmu_e - \bSigma_e^{-1}\bB\bM\mathbf{h} \label{eq:Tme}
\end{align}
with $\mathbf{h}=\bB'\bSigma_e^{-1}\bmu_e$. 
We plug \eqref{eq:TB} and \eqref{eq:Tme} into \eqref{eq:Tm}. Taking the $\ell_\infty$-norm to the subvector of $\bTheta\bmu$ with the triangle inequality and using compatibility of the induced $\ell_\infty$ norm in Lemma \ref{lem:matnorm}(a) yield
\begin{align}
\|(\bTheta\bmu)_{\cA^c}\|_\infty 
&\leq \|(\bSigma_e^{-1} \bB)_{\cA^c}\bM\bSigma_f^{-1}\bmu_f\|_\infty 
+ \|(\bSigma_e^{-1}\bmu_e)_{\cA^c}\|_\infty + \|(\bSigma_e^{-1}\bB)_{\cA^c}\bM\textbf{h}\|_\infty \\
&\leq \|(\bSigma_e^{-1})_{\cA^c\cA}\|_{2,\infty}\|\bmu_e\|_2 \\
&\quad \quad + \|(\bSigma_e^{-1})_{\cA^c\cA} \bB_{\cA\cdot}\|_{2,\infty}\|\bM\|_2\left(\|\bSigma_f^{-1}\|_2\|\bmu_f\|_2  
+ \|\textbf{h}\|_2\right).
\end{align}
In this upper bound, we evaluate each component by Conditions \ref{con:WFmodel} and \ref{con:when}. Then we have $\|(\bSigma_e^{-1})_{\cA^c\cA}(\bmu_e)_{\cA}\|_{\infty}=o(1)$, $\|(\bSigma_e^{-1})_{\cA^c\cA} \bB_{\cA\cdot}\|_{2,\infty}=O(1)$, and 
\begin{align}
\|\bSigma_f^{-1}\|_2\|\bmu_f\|_2 + \|\textbf{h}\|_2
=O(1/\sqrt{\phi}) + o(\sqrt{\phi}).
\end{align}
Moreover, Weyl's inequality with $1/\lambda_{\min}(\bB'\bB)\lesssim 1/\sqrt{\phi}$ in Condition \ref{con:when} yields 
\begin{align}
\|\bM\|_{2} &= \frac{1}{\lambda_{\min}(\bSigma_f^{-1} + \bB'\bSigma_e^{-1}\bB)} 
\leq \frac{1}{\lambda_{\min}(\bSigma_f^{-1}) +\lambda_{\min}(\bB'\bSigma_e^{-1}\bB)} \\
&\lesssim \frac{1}{\lambda_{\min}(\bB'\bB)}
\lesssim\frac{1}{\sqrt{\phi}}.\label{ineq:M}
\end{align}
By these inequalities and \eqref{ineq:Ome0mu1}, we conclude that $\|\bbeta(1)_{\cA^c} \|_\infty 
=o(1)$. This completes the proof.
\end{proof}

\subsection{Proof of Theorem \ref{thm:defac}}

\begin{proof}
To compare the feasible screening based on $\what{\bU}$ with the idealized analysis based on $\bU$, write $\what{\bU}=\bU+\bDelta$. Then, for each $j\in[N]$, we have 
$\bU_{\cdot j}'\what{\bv}=\what{\bU}_{\cdot j}'\what{\bv}-\bDelta_{\cdot j}'\what{\bv}$, 
so that the triangle inequality gives
\begin{equation}\label{eq:U_score_bound}
	\Big\|
	\frac{1}{T}\bU'\what{\bv}
	\Big\|_\infty
	\le
	\Big\|
	\frac{1}{T}\what{\bU}'\what{\bv}
	\Big\|_\infty
	+
	\Big\|
	\frac{1}{T}\bDelta'\what{\bv}
	\Big\|_\infty
	\le
	\lambda_{\textsf L}^{\textsf{full}}
	+
	\max_{j\in[N]}\frac{\|\bDelta_{\cdot j}\|_2}{\sqrt{T}}\cdot \frac{\|\what{\bv}\|_2}{\sqrt{T}},
\end{equation}
where the last inequality follows by the KKT condition for $\what{\bbeta}_{\textsf L}^{\textsf{full}}$ and the Cauchy--Schwarz inequality. 
Meanwhile, the optimality of $\what{\bbeta}_{\textsf L}^{\textsf{full}}$ with evaluating the objective function at $\bbeta=\mathbf{0}$ yields
\begin{equation}
\frac{1}{2T}\|\what{\bv}\|_2^2+\lambda_{\textsf L}^{\textsf{full}}\|\what{\bbeta}_{\textsf L}^{\textsf{full}}\|_1
\le
\frac{1}{2T}\|\bz\|_2^2=O(1).
\end{equation}
By Lemma \ref{lem:U-U}, we establish the upper bound,  
\begin{equation}\label{eq:U_score_bound}
\max_{j\in[N]}\frac{\|\bDelta_{\cdot j}\|_2}{\sqrt{T}}
= O\left(\sqrt{\frac{\phi\log N}{T}}\right)
\end{equation}
with high probability. Combining the results completes the proof.
\end{proof}

\section{Lemmas and the Proofs}

\begin{lem}\label{lem:concent_0}
Let $g(\bu)=\|\bB'\bu\|_2\vee\|\bu\|_2$. 
Under Conditions \ref{con:dgp_fe}, for any $\bu,\bv\in\mathbb{R}^N$, we have the following:
\begin{align}
(a)&\quad \bu'(\br_tz_t-\alpha\bmu) \sim \text{\normalfont ind.\ subE} \left( cg(\bu)^2\right), \\
(b)&\quad \bu'(\br_t\br_t'-\bOmega^{-1})\bv \sim \text{\normalfont ind.\ subE} \left( cg(\bu)^2g(\bv)^2  \right), \\
(c)&\quad \bu'\br_t(1-\br_t'\bOmega \bmu) \sim \text{\normalfont ind.\ subE} \left( c g(\bu)^2g(\bOmega\bmu)^2 \right).
\end{align}
\end{lem}
\begin{proof}
The proof follows from a standard argument as in Chapter 2 of \citesuppl{Vershynin2018}. Denote by $c$ a generic positive constant that takes different values accordingly. Throughout the proof, we use Condition \ref{con:dgp_fe}(b) to have $\alpha\asymp 1$ and $\tau \lesssim 1$ without specific mention. 

Prove (a). Condition \ref{con:dgp_fe} gives $\bu'\bB\bff_t\sim \text{ind.\ subG}(K^2\|\bB'\bu\|_2^2)$ and $\bu'\be_t\sim \text{ind.\ subG}(K^2\|\bu\|_2^2)$, so that we have
\begin{align}
\bu'(\br_t-\bmu) = \bu'(\bB\bff_t+\be_t)\sim \text{ind.\ subG}\left(cg(\bu)^2\right). \label{u(r-mu)}
\end{align}
Since $z_t\sim \text{ind.}N(\alpha, \tau^2)$ by Condition \ref{con:dgp_fe}(b), which implies $z_t-\alpha \sim \text{ind.\ subG}(c)$, this together with \eqref{u(r-mu)} yields
\begin{align}
\bu'(\br_t-\bmu)(z_t-\alpha) \sim \text{ind.\ subE}\left(c g(\bu)^2\right). 
\end{align}
We also obtain $\bu'\bmu(z_t-\alpha) \sim \text{ind.\ subG}(c \|\bu\|_2^2)$ because $|\bu'\bmu|\lesssim \|\bu\|_2$ is implied by Condition \ref{con:WFmodel}. 
Because subG is automatically subE, combining them yields
\begin{align}
\bu'(\br_tz_t-\alpha\bmu) 
&= \bu'(\br_t-\bmu)\alpha + \bu'(\br_t-\bmu)(z_t-\alpha) + \bu'\bmu(z_t-\alpha) \\
&\sim \text{\normalfont ind.\ subE} \left( c\left\{ \|\bu\|_2^2+g(\bu)^2\right\} \right),
\end{align}
where $\|\bu\|_2^2+g(\bu)^2 \lesssim g(\bu)^2$. This gives the result. 

Prove (b) in a similar manner. From \eqref{u(r-mu)}, we have 
\begin{align}
\bu'\left\{(\br_t-\bmu)(\br_t-\bmu)'- \bSigma \right\}\bv 
&\sim \text{\normalfont ind.\ subE} \left( cg(\bu)^2g(\bv)^2 \right)
\end{align}
and
\begin{align}
\bu'(\br_t-\bmu)\bmu'\bv + \bu'\bmu(\br_t-\bmu)'\bv 
\sim \text{\normalfont ind.\ subG} \left( c\left\{ (\bv'\bmu)^2 g(\bu)^2 + (\bu'\bmu)^2 g(\bv)^2 \right\} \right),
\end{align}
where $(\bv'\bmu)^2 g(\bu)^2 + (\bu'\bmu)^2 g(\bv)^2\lesssim \|\bv\|_2^2g(\bu)^2+\|\bu\|_2^2g(\bv)^2$. 
Combining them, we obtain
\begin{align}
\bu'(\br_t\br_t'-\bOmega^{-1})\bv 
&= \bu'\left\{(\br_t-\bmu)(\br_t-\bmu)'- \bSigma \right\}\bv +\bu'(\br_t-\bmu)\bmu'\bv+\bu'\bmu(\br_t-\bmu)'\bv \\
&\sim \text{\normalfont ind.\ subE} \left( cg(\bu)^2g(\bv)^2 \right).
\end{align}

Prove (c). We see that $\bu'\br_t(1-\br_t'\bOmega \bmu) = \bu'(\br_t-\bmu) - \bu'(\br_t\br_t'-\bOmega^{-1}) \bOmega \bmu$. By \eqref{u(r-mu)} and (b) above with $\bv=\bOmega\bmu$, we have
\begin{align*}
\bu'(\br_t-\bmu) 
\sim \text{ind.\ subG}\left(cg(\bu)^2\right)~\text{and}~
\bu'(\br_t\br_t'-\bOmega^{-1}) \bOmega \bmu
\sim \text{\normalfont ind.\ subE} \left( c g(\bu)^2g(\bOmega\bmu)^2 \right).
\end{align*}
Combining the two gives the result by the same reason as above. This completes all the proofs.
\end{proof}

\begin{lem}\label{lem:concent_1}
Under Conditions \ref{con:WFmodel} and \ref{con:dgp_fe}, the following inequalities (a)--(f) simultaneously hold with probability at least $1-O(N^{-\nu})$ for any fixed $\nu>2$, where $\phi>0$ slowly diverges if factors exist ($\bB\not=\bzero$) and $\phi=1$ if they do not ($\bB=\bzero$). 
\begin{align}
(a)&\quad \left\|\bOmega\bR'\bz/T-\bbeta\right\|_\infty 
\lesssim \sqrt{\frac{\phi\log N}{T}}, \\
(b)&\quad \left\|\bOmega\bR'\bR/T-\bI \right\|_{\max}
\lesssim \sqrt{\frac{\phi\log N}{T}}, \\
(c)&\quad \left\|\bOmega\bR'\bR \bbeta/T-\bbeta \right\|_{\infty} 
\lesssim \sqrt{\frac{\phi^2\log N}{T}},\\
(d)&\quad \left\|\bR'\bz/T-\alpha\bmu \right\|_\infty
\lesssim \sqrt{\frac{\log N}{T}}, \\
(e)&\quad \left\|\bR'\bR \bbeta/T-\alpha\bmu \right\|_{\infty}
\lesssim \sqrt{\frac{\phi\log N}{T}}, \\
(f)&\quad \left\|\bR'\bR/T-\bOmega^{-1} \right\|_{\max}
\lesssim \sqrt{\frac{\log N}{T}}, \\
(g)&\quad  \left\|\bR'(r_c\biota - \bR\bw^*) \right\|_{\infty}
\lesssim \sqrt{\frac{\phi\log N}{T}}. 
\end{align}
\end{lem}
\begin{rem}\normalfont
In Lemma \ref{lem:concent_1} (and other related places), as we can see from the following proof, some of $\phi$'s can be replaced by $\phi_{\cS}:=\phi \wedge s$, which mathematically gives a sharper upper bound. However, in realistic models that lead to a sparse $\bw^*$, we have $\phi \wedge s = \phi$. For this reason, and also for notational simplicity, we adopt this ``looser'' upper bound.
\end{rem}
\begin{proof}
Prove (a). Set $\bu=\bomega_{i\cdot}'$ in Lemma \ref{lem:concent_0}(a). 
Some algebra reduces to
\begin{align}
\max_{i\in[N]}\left\{ g(\bomega_{i\cdot}')^2\right\} 
&= \max_{i\in[N]}\left\{(\bomega_{i\cdot}\bB\bB'\bomega_{i\cdot}')\vee(\bomega_{i\cdot}\bomega_{i\cdot}')\right\} \\
&\lesssim \|\bOmega\|_{2,\infty}^2(\|\bB\|_2^2\vee 1) \\
&\lesssim \phi\vee 1 
=:M^2, 
\end{align}
where the last inequality holds by Lemma \ref{lem:Omega_sp} implies that $\|\bOmega\|_{2,\infty}\leq \|\bOmega\|_2\leq c$.  
Therefore, Bernstein's inequality for a sum of subE random variables together with the union bound yields for any $x>0$,
\begin{align}
\Pro\left( \left\|\bOmega\bR'\bz/T-\bbeta\right\|_\infty \geq x\right)
&\leq \sum_{i=1}^N\Pro\left( \left|\frac{1}{T}\sum_{t=1}^T\bomega_{i\cdot}(\br_tz_t-\alpha \bmu)\right| \geq x\right) \\
&\leq 2N\exp \left( -c \left[\frac{x^2}{c^2M^2}\right]\wedge \left[\frac{x}{cM }\right] T \right).
\end{align}
Taking $x\asymp M\sqrt{(\nu+1)\log N /T}$ makes the upper bound of this probability $O(N^{-\nu})$, which gives the result. 

The remaining part of this proof will complete with a similar argument. 
Prove (b). Set $\bu=\bomega_{i\cdot}'$ and $\bv=\vec\be_j$, where $\vec\be_j$ denotes the $N$-vector with a $1$ in the $j$th coordinate and $0$'s elsewhere, in Lemma \ref{lem:concent_0}(b). Because $g(\vec\be_j)$ is uniformly bounded, we see that 
\begin{align}
\max_{i,j\in[N]} \left\{ g(\bomega_{i\cdot}')^2g(\vec\be_j)^2 \right\} 
&\lesssim \|\bOmega\|_{2,\infty}^2\left\{\|\bB\|_2^2\vee 1\right\} \\
&\lesssim \phi \vee 1. \label{Omega_sp}
\end{align}
This with Bernstein's inequality and the union bound establishes the result. 

Prove (c). Set $\bu=\bomega_{i\cdot}'$ and $\bv=\bbeta$ in Lemma \ref{lem:concent_0}(b). By Condition \ref{con:sparse}, we have $\|\bB_{\cS}\|_2\leq \|\bB\|_2 \wedge s\max_{i\in{N}}\|\bb_{i\cdot}'\|_2\lesssim \sqrt{\phi_{\cS}}$ with $\phi_{\cS}=\phi\wedge s$. A similar argument leads to
\begin{align}
\max_{i\in[N]}\left\{ g(\bomega_{i\cdot}')^2g(\bbeta)^2 \right\} 
&\lesssim \|\bOmega\|_{2,\infty}^2\|\bbeta\|_2^2 (\|\bB\|_2^2\vee 1)(\|\bB_\cS\|_2^2\vee 1) \\
&\lesssim (\|\bB\|_2^2\vee 1)(\|\bB_\cS\|_2^2\vee 1) \\
&\lesssim (\phi \vee 1) (\phi_{\cS} \vee 1),
\end{align}
where we have used $\|\bOmega\|_{2,\infty}\leq \|\bOmega\|_2\lesssim 1$ and $\|\bbeta\|_2 \lesssim \|\bOmega\|_2\|\bmu\|_2\lesssim1$ as in the proof of (a). Thus Bernstein's inequality and the union bound immediately establish the result. 

Prove (d). Set $\bu=\vec\be_i$ in Lemma \ref{lem:concent_0}(a). Then we have
\begin{align}
&\max_{i\in[N]}g(\vec\be_i)^2\lesssim 1.
\end{align}
This with Bernstein's inequality and the union bound immediately establishes the result. 

Prove (e). Set $\bu=\vec\be_i$ and $\bv=\bbeta$ in Lemma \ref{lem:concent_0}(b). Then we have
\begin{align}
\max_{i\in[N]}\left\{ g(\vec\be_i)^2g(\bbeta)^2 \right\} 
\lesssim 
\|\bB_\cS\|_2^2\vee 1\lesssim \phi_{\cS}\vee 1.
\end{align}
This with Bernstein's inequality and the union bound immediately establishes the result. 

Prove (f). Set $\bu=\vec\be_i$ and $\bv=\vec\be_j$ in Lemma \ref{lem:concent_0}(b). Then we have
\begin{align}
&\max_{i,j\in[N]}g(\vec\be_i)^2g(\vec\be_j)^2 \lesssim 1.
\end{align}
This with Bernstein's inequality and the union bound immediately establishes the result.  

Prove (g). Set $\bu=\bomega_{i\cdot}'$ in Lemma \ref{lem:concent_0}(c). By the sparsity of $\bw^*\asymp\bOmega\bmu$ in Condition \ref{con:sparse} together with $\|\bOmega\bmu\|_2=O(1)$ implied by Condition \ref{con:WFmodel}, we have
\begin{align}
&\max_{i\in[N]}g(\vec\be_i)^2g(\bOmega\mu)^2 \lesssim \|\bB'\bOmega\bmu\|_2^2\vee \|\bOmega\bmu\|_2^2 \lesssim \|\bB_\cS\|_2^2\vee 1\lesssim \phi_{\cS}\vee 1.
\end{align}
This is the same as (e). Finally, note that $\phi_{\cS}= \phi\wedge s \leq \phi$. 
This completes all the proofs. 
\end{proof}

\begin{lem}\label{lem:Omega_sp}
Under Condition \ref{con:WFmodel}, 
we have $\|\bOmega\|_2\in[c,1/c]$ and $\|\bOmega^{-1}\|_2\lesssim \phi$.
\end{lem}
\begin{proof}
Let $\lambda_k(\cdot)$ denote the $k$th largest eigenvalue, and define $\lambda_{\max}(\cdot)=\lambda_1(\cdot)$ and $\lambda_{\min}(\cdot)=\lambda_N(\cdot)$. Since $\bOmega=\E[\br_t\br_t']^{-1}$ is symmetric and positive definite, we have 
\begin{align}
\|\bOmega\|_2=\lambda_{\max}(\bOmega)=\frac{1}{\lambda_{\min}(\bmu\bmu'+\bB\bB'+\bSigma_e)}.
\end{align}
By Weyl's inequality, Condition \ref{con:dgp_fe} entails that
\begin{align}
\frac{1}{\lambda_{\min}(\bmu\bmu'+\bB\bB'+\bSigma_e) }
&\leq \frac{1}{\lambda_{\min}(\bmu\bmu')+\lambda_{\min}(\bB\bB')+\lambda_{\min}(\bSigma_e)} \\ 
&\leq \frac{1}{\lambda_{\min}(\bSigma_e)} \leq \frac{1}{c}
\end{align}
and 
\begin{align}
\frac{1}{\lambda_{\min}(\bmu\bmu'+\bB\bB'+\bSigma_e)} 
&\geq \frac{1}{\lambda_{2}(\bmu\bmu')+\lambda_{N-1}(\bB\bB'+\bSigma_e)} \\
&\geq \frac{1}{\lambda_{r+1}(\bB\bB')+\lambda_{N-r-1}(\bSigma_e)}
\geq \frac{1}{\lambda_{\max}(\bSigma_e)}\geq c.
\end{align}

Furthermore, by Condition \ref{con:WFmodel}, we obtain $\|\bmu\|_2=O(1)$ and 
\begin{align}
\|\bOmega^{-1}\|_2&=\lambda_{\max}(\bOmega^{-1})=\lambda_{\max}(\bmu\bmu'+\bB\bB'+\bSigma_e) \\ 
&\leq \lambda_{\max}(\bmu\bmu')+\lambda_{\max}(\bB\bB')+\lambda_{\max}(\bSigma_e) \\ 
&\lesssim \phi.
\end{align}
This completes the proof. 
\end{proof}

\begin{lem}\label{lem:L_bound}
Assume the same conditions as in Theorem \ref{thm:lasso}. 
Then the Lasso estimator $\what\bbeta_{\textsf{L}}$ satisfies the following upper bounds with probability at least $1-O(N^{-\nu})$:
\begin{align}
(a)&\quad \|\what\bbeta_{\textsf{L}}-\bbeta \|_2
\lesssim \sqrt{\frac{s\phi\log N}{T}}, \\
(b)&\quad \|\what\bbeta_{\textsf{L}}-\bbeta \|_1
\lesssim \sqrt{\frac{s^2\phi\log N}{T}}.
\end{align}
\end{lem}
\begin{proof}
By the definition of the Lasso estimator, we have
\begin{align}
T^{-1}\|\bz-\bR\what{\bbeta}_{\textsf{L}}\|_2^2 + 2\lambda \|\what{\bbeta}_{\textsf{L}}\|_1 
\leq T^{-1}\|\bz-\bR\bbeta \|_2^2 + 2\lambda \|{\bbeta} \|_1. 
\end{align}
Expanding the terms and rearranging them, we have
\begin{align}
&2^{-1}(\what\bbeta_{\textsf{L}}-\bbeta )'(T^{-1}\bR'\bR)(\what\bbeta_{\textsf{L}}-\bbeta ) \\
&\leq -(T^{-1}\bR'\bR\bbeta -\alpha\bmu)'(\what\bbeta_{\textsf{L}}-\bbeta ) 
+ (\bR'\bz/T-\alpha\bmu)'(\what{\bbeta}_{\textsf{L}}-\bbeta) + \lambda(\|{\bbeta} \|_1-\|\what{\bbeta}_{\textsf{L}}\|_1 ) \\
&\leq \left(\|T^{-1}\bR'\bR\bbeta -\alpha\bmu\|_{\infty} 
+ \|\bR'\bz/T-\alpha\bmu\|_\infty \right) \|\what\bbeta_{\textsf{L}}-\bbeta \|_1
+ \lambda(\|{\bbeta} \|_1-\|\what{\bbeta}_{\textsf{L}}\|_1 ). 
\end{align}
In this upper bound, Lemma \ref{lem:concent_1}(d)--(e) gives
\begin{align}
\|\bR'\bR\bbeta/T -\alpha\bmu\|_{\infty} 
+ \|\bR'\bz/T-\alpha\bmu\|_\infty  
\lesssim \sqrt{\frac{\phi\log N}{T}}
\asymp \lambda
\end{align}
with probability at least $1-O(N^{-\nu})$. 
Therefore, by the triangle inequality, we obtain
\begin{align}
2^{-1}(\what\bbeta_{\textsf{L}}-\bbeta )'(T^{-1}\bR'\bR)(\what\bbeta_{\textsf{L}}-\bbeta )
&\lesssim \lambda\|\what\bbeta_{\textsf{L}}-\bbeta \|_1
+ \lambda(\|{\bbeta} \|_1-\|\what{\bbeta}_{\textsf{L}}\|_1 ) \label{ineq:for_lem0} \\ 
&\lesssim\lambda\|\what\bbeta_{\textsf{L}}-\bbeta \|_1,
\end{align}
By decomposability of the $\ell_1$-norm and the same argument as in \citesuppl{Negahban2012},  we can actually establish $\|\what\bbeta_{\cS^c}^{\textsf{L}}-\bbeta_{\cS^c}\|_1\lesssim \|\what\bbeta_{\cS}^{\textsf{L}}-\bbeta_{\cS}\|_1$. Thus the following inequality from \eqref{ineq:for_lem0} holds:
\begin{align}
\|\what\bbeta_{\textsf{L}}-\bbeta\|_1
&= \|\what\bbeta_{\cS^c}^{\textsf{L}}-\bbeta_{\cS^c}\|_1 + \|\what\bbeta_{\cS}^{\textsf{L}}-\bbeta_{\cS}\|_1 \\
&\lesssim\|\what\bbeta_{\cS}^{\textsf{L}}-\bbeta_{\cS}\|_1 
\leq \sqrt{s}\|\what\bbeta_{\cS}^{\textsf{L}}-\bbeta_{\cS}\|_2
\leq \sqrt{s}\|\what\bbeta_{\textsf{L}}-\bbeta\|_2.\label{lem:Nega}
\end{align}
Combining the bounds with Lemma \ref{lem:LB}, we conclude that 
\begin{align}
\|\what\bbeta_{\textsf{L}}-\bbeta \|_2^2
&\lesssim \frac{\lambda}{\lambda_{\min}(\bOmega^{-1}) - O\left(s\sqrt{T^{-1}\log N}\right)}\|\what\bbeta_{\textsf{L}}-\bbeta \|_1 \label{ineq:ell1} \\
&\lesssim \frac{\sqrt{s}\lambda}{\lambda_{\min}(\bOmega^{-1}) - O\left(s\sqrt{T^{-1}\log N}\right)}\|\what\bbeta_{\textsf{L}}-\bbeta \|_2
\end{align}
 with probability at least $1-O(N^{-\nu})$. Dividing both sides by $\|\what\bbeta_{\textsf{L}}-\bbeta \|_2$ gives the result of (a). The result of (b) is then obtained by \eqref{lem:Nega} and \eqref{ineq:ell1}. Finally, using $\lambda_{\min}(\bOmega^{-1})=1/\|\bOmega\|_2\asymp 1$ by Lemma \ref{lem:Omega_sp} completes the proof. 
\end{proof}

\begin{lem}\label{lem:LB}
Assume the same conditions as in Theorem \ref{thm:lasso}. Then the Lasso estimator $\what\bbeta_\textsf{L}$ satisfies the following inequality with probability at least $1-O(N^{-\nu})$:
\begin{align}
(\what\bbeta_\textsf{L}-\bbeta)'(T^{-1}\bR'\bR)(\what\bbeta_\textsf{L}-\bbeta)
\geq \left\{ \lambda_{\min}(\bOmega^{-1}) - O\left(\sqrt{\frac{s^2\log N}{T}}\right) \right\}\|\what\bbeta_\textsf{L}-\bbeta\|_2^2.
\end{align}
\end{lem}
\begin{proof}
For $\bdelta=\what\bbeta_\textsf{L}-\bbeta$ with $\cS=\supp(\bbeta)$, Lemma \ref{lem:concent_1} and \eqref{lem:Nega}  yield
\begin{align}
\bdelta'(T^{-1}\bR'\bR)\bdelta
&= \bdelta'\bOmega^{-1}\bdelta 
+ \bdelta'(T^{-1}\bR'\bR-\bOmega^{-1})\bdelta \\
&\geq \lambda_{\min}(\bOmega^{-1})\|\bdelta\|_2^2 
- \|T^{-1}\bR'\bR-\bOmega^{-1}\|_{\max}\|\bdelta\|_1^2 \\
&\geq \left\{ \lambda_{\min}(\bOmega^{-1}) - O\left(\sqrt{\frac{s^2\log N}{T}}\right) \right\}\|\bdelta\|_2^2
\end{align}
 with probability at least $1-O(N^{-\nu})$, where $\lambda_{\min}(\bOmega^{-1})=1/\|\bOmega\|_2\asymp 1$ by Lemma \ref{lem:Omega_sp}. This completes the proof. 
\end{proof}

\begin{lem}\label{lem:Omega}
Assume the same conditions as in Theorem \ref{thm:lasso}. 
\begin{align}
(a)&\quad \|\bOmega\|_\infty \lesssim \sqrt{s}, \\
(b)&\quad \|(\bOmega^{-1})_{\cS^c\cS}\|_{2,\infty} \lesssim \sqrt{\phi}.
\end{align}
\end{lem}
\begin{proof}
Prove (a). By $\E[\br_t\br_t']=\bmu\bmu'+\bB\bSigma_f\bB'+\bSigma_e$, the SMW formula gives 
\begin{align}
\bOmega = \bTheta - \frac{\bTheta \bmu \bmu'\bTheta}{1+\bmu'\bTheta\bmu}
\end{align}
with $\bTheta:=\bSigma^{-1}=(\bB\bSigma_f\bB'+\bSigma_e)^{-1}$. 
Then, since $1+\bmu'\bTheta\bmu\geq 1$ by the positive definiteness of $\bTheta$, we obtain 
\begin{align}
\|\bOmega\|_\infty \leq \|\bTheta\|_\infty + \|\bTheta \bmu\|_\infty \|\bTheta\bmu\|_1.
\end{align}
We evaluate each norm. Condition \ref{con:WFmodel} implies $\|\bTheta\bmu\|_2\asymp \|\bw^*\|_2=O(1)$, which gives $\|\bTheta\bmu\|_\infty=O(1)$ and $\|\bTheta\bmu\|_1 \leq  \sqrt{s}\|\bTheta\bmu\|_2=O(\sqrt{s})$.

We next consider $\|\bTheta\|_\infty$. By the SMW formula, we have  
\begin{align*}
\bTheta = \bSigma_e^{-1} - \bSigma_e^{-1}\bB(\bSigma_f^{-1}+\bB'\bSigma_e^{-1}\bB)\bB'\bSigma_e^{-1}.
\end{align*}
Condition \ref{con:sparse}(b) and (c) give $\|\bB'\|_\infty=\|\bB\|_1=O(\phi)$, $\|\bB\|_\infty=O(1)$, and $\|\bSigma_e^{-1}\|_\infty\asymp 1$. Thus the norm is bounded as
\begin{align*}
\|\bTheta\|_\infty &\leq \|\bSigma_e^{-1}\|_\infty + \|\bSigma_e^{-1}\|_\infty \|\bB\|_\infty \|(\bSigma_f^{-1}+\bB'\bSigma_e^{-1}\bB)^{-1}\|_\infty\|\bB'\|_\infty \|\bSigma_e^{-1}\|_\infty \\
&\lesssim 1 + \|(\bSigma_f^{-1}+\bB'\bSigma_e^{-1}\bB)^{-1}\|_\infty \phi.
\end{align*}
By Condition \ref{con:sparse}(b), we have $\lambda_{\min}(\bB'\bB)\gtrsim \sqrt{\phi}$. Thus the norm in this upper bound is further bounded by
\begin{align*}
\|(\bSigma_f^{-1}+\bB'\bSigma_e^{-1}\bB)^{-1}\|_2
&= \frac{1}{\lambda_{\min}(\bSigma_f^{-1}+\bB'\bSigma_e^{-1}\bB)} \\
&\leq \frac{1}{\lambda_{\min}(\bSigma_f^{-1})+\lambda_{\min}(\bB'\bB)\lambda_{\min}(\bSigma_e^{-1})}
\lesssim \frac{1}{\sqrt{\phi}}.
\end{align*}
Combining the results with $\phi \leq s$ in Condition \ref{con:sparse}(a) leads to 
$\|\bOmega\|_\infty \lesssim \sqrt{\phi} + \sqrt{s}\leq \sqrt{s}$. 

Prove (b). By Conditions \ref{con:WFmodel}(a)--(c) and \ref{con:sparse}(a)--(b), we have
\begin{align*}
\|(\bOmega^{-1})_{\cS^c\cS}\|_{2,\infty}
&\leq \|\bmu_{\cS^c}\bmu_{\cS}'\|_{2,\infty} + \|\bB_{\cS^c}\bSigma_f\bB_{\cS}\|_{2,\infty} + \|\bSigma_e\|_{2,\infty} \\
&\leq \|\bmu_{\cS^c}\|_\infty \|\bmu_{\cS}\|_2 +  \|\bB_{\cS^c}\|_{2,\infty} \|\bSigma_f\|_2 \|\bB_{\cS}\|_{2} + \|\bSigma_e\|_{2} \\
&\lesssim \|\bB_{\cS}\|_{2}
\lesssim \sqrt{\phi}\wedge \sqrt{s} \leq \sqrt{\phi}.
\end{align*}
This completes the proof. 
\end{proof}

\begin{lem}\label{lem:shat}
Assume the same conditions as in Theorem \ref{thm:post}. 
We have $\hat{s}\lesssim \phi s$. 
\end{lem}
\begin{proof}
We work on the event, $\{\cS\subseteq \what\cS\}$. 
The KKT condition for $\hat{\beta}_j^{\textsf{L}}\not=0$ in \eqref{KKT} gives
\begin{align}
\lambda_{\textsf{L}}\sqrt{\hat{s}} &= \|\bR_{\what\cS}'(\bz-\bR\what\bbeta)/T\|_2 \notag\\
&\leq \|\bR_{\what\cS}'\bz/T - \alpha \bmu_{\what\cS}\|_2 
+ \|\bR_{\what\cS}'\bR\bbeta/T - \alpha\bmu_{\what\cS}\|_2
+ \|\bR_{\what\cS}'\bR(\what\bbeta-\bbeta)/T\|_2 \notag\\
&\leq \sqrt{\hat{s}}\|\bR'\bz/T - \alpha \bmu\|_\infty 
+ \sqrt{\hat{s}}\|\bR'\bR\bbeta/T - \alpha\bmu\|_\infty 
+ \sqrt{\|\bR_{\what\cS}'\bR_{\what\cS}/T\|_2} \left\| \frac{\bR(\what\bbeta-\bbeta)}{\sqrt{T}} \right\|_2 \notag\\
&\lesssim \sqrt{\hat{s}}\sqrt{\frac{\log N}{T}} + \sqrt{\hat{s}}\sqrt{\frac{\phi\log N}{T}} \notag\\ 
&\qquad + \sqrt{ \phi + \sup_{\cU\subseteq[N]:|\cU|\leq T}\|(\bR'\bR/T-\bOmega^{-1})_{\cU\cU}\|_2}\left\| \frac{\bR(\what\bbeta-\bbeta)}{\sqrt{T}} \right\|_2, \label{ineq:shat}
\end{align}
where the third inequality holds by Lemmas \ref{lem:concent_1}(d), \ref{lem:concent_1}(e), \ref{lem:Omega_sp}, and the fact that the Lasso active set is of the size at most $T$. 

For the prediction error, the basic inequality of the Lasso optimization gives
\begin{align*}
    \left\| \frac{\bR(\what\bbeta-\bbeta)}{\sqrt{T}} \right\|_2^2 \lesssim \lambda_{\textsf{L}} \|\what\bbeta-\bbeta\|_1 \lesssim \lambda_{\textsf{L}} \sqrt{s} \|\what\bbeta-\bbeta\|_2 \lesssim \frac{s\phi\log N}{T}.
\end{align*}

To uniformly bound the spectral norm, we rely on the discretization via an $\epsilon$-net following \cite{Vershynin2018}. For any $\cU \subseteq [N]$ with $|\cU| = T$, let $\mathbb{S}_{\cU}^{T-1}$ denote the unit sphere in $\mathbb{R}^N$ restricted to the support $\cU$. For each such $\cU$, let $\mathcal{N}_{\cU}$ denote a minimal $1/4$-net of $\mathbb{S}_{\cU}^{T-1}$ with respect to the Euclidean metric. Moreover, let $\cN = \cup_{\cU\subseteq [N]:|\cU|=T}\cN_{\cU}$. Therefore, we have
\begin{align*}
\sup_{\cU\subset[N]:|\cU|\leq T}\|(\bR'\bR/T-\bOmega^{-1})_{\cU\cU}\|_2 
&= \sup_{|\cU|\leq T} \sup_{\bx\in \mathbb{S}_{\cU}^{T-1}}|\bx'(\bR'\bR/T-\bOmega^{-1})\bx| \\
&\leq 2 \sup_{\by\in \cN} |\by'(\bR'\bR/T-\bOmega^{-1})\by|
\end{align*}
By a standard volumetric argument, we have $|\cN|\leq \binom{N}{T} 9^T \leq (9\e N/T)^T$. Hence, Bernstein's inequality implied by Lemma \ref{lem:concent_0}(b) with the union bound, we obtain   
\begin{align*}
\Pro \left( 2 \sup_{\by\in \cN} |\by'(\bR'\bR/T-\bOmega^{-1})\by| > x \right) 
&\leq 2 \left(\frac{9\e N}{T}\right)^T \exp\left(-cT \frac{x^2}{M^2}\right)\\
&= 2 \exp\left\{T\log\left(\frac{9\e N}{T}\right)-\frac{cTx^2}{M^2}\right\},
\end{align*}
where $M=\|\bB\|_2^2\vee 1\lesssim \phi$. Taking $x\asymp M\sqrt{\log (N/T)}$ leads to the upper bound less than or equal to $O(N^{-\nu})$ for some $\nu>0$. We consequently have 
\begin{align}
\sup_{\cU\subset[N]:|\cU|\leq T}\|(\bR'\bR/T-\bOmega^{-1})_{\cU\cU}\|_2 
&\lesssim \phi\sqrt{\log (N/T)}.\label{ineq:unfRR}
\end{align}

Combining \eqref{ineq:shat}, the prediction error bound, and \eqref{ineq:unfRR}, we obtain
\begin{align*}
\lambda_{\textsf{L}}\sqrt{\hat{s}} 
&\leq c \sqrt{\hat{s}}\sqrt{\phi}\sqrt{\frac{\log N}{T}} 
+ c\sqrt{\phi + \phi\sqrt{\log \left(\frac{N}{T}\right)}} \sqrt{\frac{s\phi\log N}{T}}
\end{align*}
for some constant $c>0$. Recall that $\lambda_{\textsf{L}}\asymp \sqrt{\phi}\sqrt{\log N/T}$. Therefore, we can take another constant $C>c$ such that
\begin{align*}
(C-c)\sqrt{\hat{s}} \leq c\sqrt{\phi + \phi\sqrt{\log \left(\frac{N}{T}\right)}}\sqrt{s} .
\end{align*}
If $N\geq T$, squaring both sides gives $\hat{s}\lesssim \phi s\sqrt{\log N}$. 

Finally, we obtain a sharper upper bound for $\hat{s}$. Since we have established that $\hat{s} \leq \bar{k}$ with $\bar{k} \asymp \phi s \sqrt{\log N}$ holds with high probability, we can restrict the range of the supremum in the preceding $1/4$-net argument from $|\cU|\leq T$ to $|\cU|\leq \bar{k}$. Then, for the new $\cN$, it holds that $|\cN| \leq (9\e N/\bar{k})^{\bar{k}}$. Accordingly, recalculating the bound in \eqref{ineq:unfRR} under $\sup_{\cU\subset[N]:|\cU|\leq \bar{k}}$ yields an upper bound of $o(\phi)$. Using this to re-evaluate \eqref{ineq:shat} by the same argument, we obtain $\sqrt{\hat{s}}\lesssim \sqrt{\phi(1+o(1))}\sqrt{s}$. Consequently, we have $\hat{s}\lesssim \phi s$, which completes the proof.
\end{proof}

\begin{lem}\label{lem:thetahat}
Assume the same conditions as in Theorem \ref{thm:post}. 
We have 
\begin{equation}
|\hat{\theta} - \theta| 
\lesssim s\sqrt{\frac{\log N}{T}}. 
\end{equation}
\end{lem}
\begin{proof}
We evaluate the convergence rate of $|\hat{\theta} - \theta|$ strictly utilizing the sparsity of the true weight vector. Conditional on the event $\{\cS \subseteq \what{\cS}\}$, we notice that $\bSigma_{\what{\cS}}^{-1} \bmu_{\what{\cS}}$ has at most $s$ non-zero elements since the true signal is confined to $\cS$. By the bounded eigenvalue condition and $\|\bmu\|_2 = O(1)$, we have $\|\bSigma_{\what{\cS}}^{-1} \bmu_{\what{\cS}}\|_2 = O(1)$, which implies $\|\bSigma_{\what{\cS}}^{-1} \bmu_{\what{\cS}}\|_1 \le \sqrt{s}\|\bSigma_{\what{\cS}}^{-1} \bmu_{\what{\cS}}\|_2 = O(\sqrt{s})$.

From the sub-Gaussianity assumptions and Weyl's inequality, we further have the following basic bounds that hold with high probability:
\begin{align*}
    \|\what{\bmu}_{\what{\cS}} - \bmu_{\what{\cS}}\|_\infty &= O\left(\sqrt{\frac{\log N}{T}}\right), ~~ \|\what{\bmu}_{\what{\cS}} - \bmu_{\what{\cS}}\|_2 \le \sqrt{\bar{s}}\|\what{\bmu}_{\what{\cS}} - \bmu_{\what{\cS}}\|_\infty = O\left(\sqrt{\frac{\bar{s}\log N}{T}}\right), \\
    \|\what{\bSigma}_{\what{\cS}} - \bSigma_{\what{\cS}}\|_{\max} &= O\left(\sqrt{\frac{\log N}{T}}\right), ~~ \|\what{\bSigma}_{\what{\cS}} - \bSigma_{\what{\cS}}\|_2 \le \bar{s}\|\what{\bSigma}_{\what{\cS}} - \bSigma_{\what{\cS}}\|_{\max} = O\left(\bar{s}\sqrt{\frac{\log N}{T}}\right), \\
    \|\what{\bSigma}_{\what{\cS}}^{-1}\|_2 &= O(1).
\end{align*}
Note that $\|\what{\bSigma}_{\what{\cS}} - \bSigma_{\what{\cS}}\|_2 = o_p(1)$ holds by the assumed condition.

We decompose the estimation error $\hat{\theta} - \theta = \what{\bmu}_{\what{\cS}}' \what{\bSigma}_{\what{\cS}}^{-1} \what{\bmu}_{\what{\cS}} - \bmu_{\what{\cS}}' \bSigma_{\what{\cS}}^{-1} \bmu_{\what{\cS}}$ into three terms:
\begin{align}
    \hat{\theta} - \theta &= (\what{\bmu}_{\what{\cS}} - \bmu_{\what{\cS}})' \what{\bSigma}_{\what{\cS}}^{-1} (\what{\bmu}_{\what{\cS}} - \bmu_{\what{\cS}}) 
    + 2 \bmu_{\what{\cS}}' \what{\bSigma}_{\what{\cS}}^{-1} (\what{\bmu}_{\what{\cS}} - \bmu_{\what{\cS}}) 
    + \bmu_{\what{\cS}}' (\what{\bSigma}_{\what{\cS}}^{-1} - \bSigma_{\what{\cS}}^{-1}) \bmu_{\what{\cS}} \notag \\
    &=: T_1 + T_2 + T_3. \label{eq:theta_decomp}
\end{align}

For the first term $T_1$, the spectral norm bound directly yields
\begin{equation*}
    |T_1| = |(\what{\bmu}_{\what{\cS}} - \bmu_{\what{\cS}})' \what{\bSigma}_{\what{\cS}}^{-1} (\what{\bmu}_{\what{\cS}} - \bmu_{\what{\cS}})| \le \|\what{\bSigma}_{\what{\cS}}^{-1}\|_2 \|\what{\bmu}_{\what{\cS}} - \bmu_{\what{\cS}}\|_2^2 = O_p\left(\frac{\bar{s} \log N}{T}\right).
\end{equation*}

For the second term $T_2$, we apply the identity $\what{\bSigma}_{\what{\cS}}^{-1} = \bSigma_{\what{\cS}}^{-1} - \bSigma_{\what{\cS}}^{-1} (\what{\bSigma}_{\what{\cS}} - \bSigma_{\what{\cS}}) \what{\bSigma}_{\what{\cS}}^{-1}$ to rewrite it as $T_2=2 \bmu_{\what{\cS}}' \bSigma_{\what{\cS}}^{-1} (\what{\bmu}_{\what{\cS}} - \bmu_{\what{\cS}}) - 2 \bmu_{\what{\cS}}' \bSigma_{\what{\cS}}^{-1} (\what{\bSigma}_{\what{\cS}} - \bSigma_{\what{\cS}}) \what{\bSigma}_{\what{\cS}}^{-1} (\what{\bmu}_{\what{\cS}} - \bmu_{\what{\cS}})$. We bound the main part using H\"{o}lder's inequality and the sparsity of $\bSigma_{\what{\cS}}^{-1} \bmu_{\what{\cS}}$:
\begin{equation*}
    |2 \bmu_{\what{\cS}}' \bSigma_{\what{\cS}}^{-1} (\what{\bmu}_{\what{\cS}} - \bmu_{\what{\cS}})| \le 2\|\bSigma_{\what{\cS}}^{-1} \bmu_{\what{\cS}}\|_1 \|\what{\bmu}_{\what{\cS}} - \bmu_{\what{\cS}}\|_\infty = O_p\left(\sqrt{\frac{s \log N}{T}}\right).
\end{equation*}
Since $\|\what{\bSigma}_{\what{\cS}}^{-1} (\what{\bmu}_{\what{\cS}} - \bmu_{\what{\cS}})\|_2 \le \|\what{\bSigma}_{\what{\cS}}^{-1}\|_2 \|\what{\bmu}_{\what{\cS}} - \bmu_{\what{\cS}}\|_2 = O_p(\sqrt{\bar{s} \log N / T})$, the remainder is bounded by
\begin{align*}
    &|2 \bmu_{\what{\cS}}' \bSigma_{\what{\cS}}^{-1} (\what{\bSigma}_{\what{\cS}} - \bSigma_{\what{\cS}}) \what{\bSigma}_{\what{\cS}}^{-1} (\what{\bmu}_{\what{\cS}} - \bmu_{\what{\cS}})| \\
    &\le 2\|\bSigma_{\what{\cS}}^{-1} \bmu_{\what{\cS}}\|_1 \max_{j \in \cS} \left|(\what{\bSigma}_{\what{\cS}} - \bSigma_{\what{\cS}})_{j\cdot} \what{\bSigma}_{\what{\cS}}^{-1} (\what{\bmu}_{\what{\cS}} - \bmu_{\what{\cS}})\right| \\
    &\le 2\|\bSigma_{\what{\cS}}^{-1} \bmu_{\what{\cS}}\|_1 \sqrt{\bar{s}}\|\what{\bSigma}_{\what{\cS}} - \bSigma_{\what{\cS}}\|_{\max} \|\what{\bSigma}_{\what{\cS}}^{-1} (\what{\bmu}_{\what{\cS}} - \bmu_{\what{\cS}})\|_2 = O_p\left(\bar{s}\sqrt{s}\frac{\log N}{T}\right).
\end{align*}

For the third term $T_3$, using $\what{\bSigma}_{\what{\cS}}^{-1} - \bSigma_{\what{\cS}}^{-1} = -\bSigma_{\what{\cS}}^{-1}(\what{\bSigma}_{\what{\cS}} - \bSigma_{\what{\cS}}) \bSigma_{\what{\cS}}^{-1} + \bSigma_{\what{\cS}}^{-1}(\what{\bSigma}_{\what{\cS}} - \bSigma_{\what{\cS}}) \what{\bSigma}_{\what{\cS}}^{-1}(\what{\bSigma}_{\what{\cS}} - \bSigma_{\what{\cS}}) \bSigma_{\what{\cS}}^{-1}$, we obtain that the third term equals $-\bmu_{\what{\cS}}' \bSigma_{\what{\cS}}^{-1} (\what{\bSigma}_{\what{\cS}} - \bSigma_{\what{\cS}}) \bSigma_{\what{\cS}}^{-1} \bmu_{\what{\cS}} + \bmu_{\what{\cS}}' \bSigma_{\what{\cS}}^{-1} (\what{\bSigma}_{\what{\cS}} - \bSigma_{\what{\cS}}) \what{\bSigma}_{\what{\cS}}^{-1} (\what{\bSigma}_{\what{\cS}} - \bSigma_{\what{\cS}}) \bSigma_{\what{\cS}}^{-1} \bmu_{\what{\cS}}$. The main part is tightly bounded by exploiting the $s$-sparsity:
\begin{align*}
    |\bmu_{\what{\cS}}' \bSigma_{\what{\cS}}^{-1} (\what{\bSigma}_{\what{\cS}} - \bSigma_{\what{\cS}}) \bSigma_{\what{\cS}}^{-1} \bmu_{\what{\cS}}| 
    &\le \|\bSigma_{\what{\cS}}^{-1} \bmu_{\what{\cS}}\|_1 \|(\what{\bSigma}_{\what{\cS}} - \bSigma_{\what{\cS}}) \bSigma_{\what{\cS}}^{-1} \bmu_{\what{\cS}}\|_\infty \\
    &\le \|\bSigma_{\what{\cS}}^{-1} \bmu_{\what{\cS}}\|_1^2 \|\what{\bSigma}_{\what{\cS}} - \bSigma_{\what{\cS}}\|_{\max} 
    = O_p\left(s \sqrt{\frac{\log N}{T}}\right).
\end{align*}
Notice that the elements of $(\what{\bSigma}_{\what{\cS}} - \bSigma_{\what{\cS}}) \bSigma_{\what{\cS}}^{-1} \bmu_{\what{\cS}}$ are uniformly bounded as $\|(\what{\bSigma}_{\what{\cS}} - \bSigma_{\what{\cS}}) \bSigma_{\what{\cS}}^{-1} \bmu_{\what{\cS}}\|_\infty \le \|\what{\bSigma}_{\what{\cS}} - \bSigma_{\what{\cS}}\|_{\max} \|\bSigma_{\what{\cS}}^{-1} \bmu_{\what{\cS}}\|_1 = O_p(\sqrt{s \log N / T})$. Since this vector is in $\mathbb{R}^{\hat{s}}$, its $\ell_2$-norm squared is $\|(\what{\bSigma}_{\what{\cS}} - \bSigma_{\what{\cS}}) \bSigma_{\what{\cS}}^{-1} \bmu_{\what{\cS}}\|_2^2 \le \bar{s} \|(\what{\bSigma}_{\what{\cS}} - \bSigma_{\what{\cS}}) \bSigma_{\what{\cS}}^{-1} \bmu_{\what{\cS}}\|_\infty^2 = O_p(\bar{s} s \log N / T)$. The remainder is then
\begin{align*}
    &|\bmu_{\what{\cS}}' \bSigma_{\what{\cS}}^{-1} (\what{\bSigma}_{\what{\cS}} - \bSigma_{\what{\cS}}) \what{\bSigma}_{\what{\cS}}^{-1} (\what{\bSigma}_{\what{\cS}} - \bSigma_{\what{\cS}}) \bSigma_{\what{\cS}}^{-1} \bmu_{\what{\cS}}| \\
    &\le \|\what{\bSigma}_{\what{\cS}}^{-1}\|_2 \|(\what{\bSigma}_{\what{\cS}} - \bSigma_{\what{\cS}}) \bSigma_{\what{\cS}}^{-1} \bmu_{\what{\cS}}\|_2^2 = O_p\left(\bar{s}s\frac{\log N}{T}\right).
\end{align*}

Combining all terms, the main components are $O_p(\sqrt{s \log N / T})$ and $O_p(s \sqrt{\log N / T})$, where $O_p(s \sqrt{\log N / T})$ dominates. All higher-order remainder terms are bounded by the main terms multiplied by $\bar{s}\sqrt{\log N / T} = o_p(1)$. Thus, the final convergence rate is obtained. This completes the proof. 
\end{proof}

\begin{lem}\label{lem:matnorm}
Let $\bM: m\times n$. Let $\bA,\bB,\bC$ be matrices that can define the product $\bA\bB\bC$. Then we have the following results. 
\begin{align}
(a)&\quad \|\bM\bx\|_{q} 
\leq \|\bM\|_{p,q}\|\bx\|_{p}, \\
(b)&\quad \|\bA\bB\bC\|_{p,q} 
\leq \|\bA\|_{\beta,q}\|\bB\|_{\alpha,\beta}\|\bC\|_{p,\alpha}, \\
(c)&\quad \|\bM\|_{\alpha,\beta} \leq m^{0\vee(1/\beta-1/2)}n^{0\vee(1/2-1/\alpha)} \|\bM\|_{2}, \\
(d)&\quad \alpha \leq \beta \text{ implies } \|\bM\|_{\alpha,q}\leq \|\bM\|_{\beta,q} \text{ and } \|\bM\|_{p, \alpha}\geq \|\bM\|_{p,\beta}, \\
(e)&\quad \|\bA\|_{\infty,1} \leq \sum_{i,j}|a_{ij}|, \\
(f)&\quad \|\bA\|_{1,\infty} = \max_{ij}|a_{ij}|.
\end{align}
\end{lem}
\begin{proof}
Proof of (a):  For any $\by\not=\bzero$, set $\bx=\by/\|\by\|_\alpha$. We have
\begin{align}
\|\bM\by\|_{\beta}
&= \|\bM\bx\|_{\beta}\|\by\|_\alpha
\leq \sup_{\|\bu\|_\alpha=1}\|\bM\bu\|_{\beta}\|\by\|_\alpha
\equiv \|\bM\|_{\alpha,\beta}\|\by\|_\alpha,
\end{align}
giving the result. 

(b) For any $\bx$, we have 
\begin{align}
\|\bA\bB\bC\bx\|_{q} 
&\leq \|\bA\|_{\beta,q}\|\bB\bC\bx\|_{\beta} 
\leq \|\bA\|_{\beta,q}\|\bB\|_{\alpha,\beta}\|\bC\bx\|_{\alpha} \\
&\leq \|\bA\|_{\beta,q}\|\bB\|_{\alpha,\beta}\|\bC\|_{p,\alpha}\|\bx\|_{p}. 
\end{align}
Taking supremum over $\|\bx\|_p\leq 1$ gives the result. 

(c) We have 
\begin{align}
\|\bM\bx\|_{\beta}
&\leq m^{\max\{0,1/\beta-1/2\}} \|\bM\bx\|_{2} \\
&\leq m^{\max\{0,1/\beta-1/2\}} \|\bM\|_{2}\|\bx\|_{2} \\
&\leq m^{\max\{0,1/\beta-1/2\}} n^{\max\{0,1/2-1/\alpha\}}\|\bM\|_{2}\|\bx\|_{\alpha}.
\end{align}
Taking supremum over $\|\bx\|_\alpha\leq 1$ gives the result. 

(d) The result follows by monotonicity of the vector $\ell_r$ norms. 

(e) We have
\begin{align}
\|\bA\|_{\infty,1} = \sup_{\|\bx\|_\infty\leq 1}\sum_i|(\bA\bx)_i|
= \sup_{\|\bx\|_\infty\leq 1}\sum_i\left|\sum_j a_{ij}x_j\right|
\leq \sup_{\|\bx\|_\infty\leq 1}\sum_{i,j}|a_{ij}||x_{j}|.
\end{align}
Taking $x_j=1$ for all $j$ yields the result. 
\end{proof}

\begin{lem}\label{lem:U-U}
Assume the same conditions as in Theorem \ref{thm:defac}. We have 
\begin{align*}
\frac{1}{\sqrt{T}}\max_{j\in[N]}\|\what{\bU}_{\cdot j}-{\bU}_{\cdot j}\|_2
= O_p\left(\sqrt{\frac{\phi\log N}{T}}\right).
\end{align*}
\end{lem}
\begin{proof}
By the definition of the residuals in the modified PS$^2$ with the model $\bR=\bX\bA'+\bU$, we obtain 
\begin{align*}
\what{\bU} = \bR - \bX \what\bA'
= \bU-\bX(\what\bA - \bA)'.
\end{align*}
The OLS estimator is given by
\begin{align*}
\what\bA' = (\wtilde\bX'\wtilde\bX)^{-1} \wtilde\bX' \bR
= \bA' + (\wtilde\bX'\wtilde\bX)^{-1} \wtilde\bX' \bU
= \bA' + (\wtilde\bX'\wtilde\bX)^{-1} \wtilde\bX' (\bU-\biota_T\bmu_u'),
\end{align*}
where the last equality holds by $\wtilde\bX' \biota_T=\bzero$. Thus we have 
\begin{align*}
\what{\bU} - \bU
= -\bX (\wtilde\bX'\wtilde\bX)^{-1}\wtilde\bX'(\bU-\biota_T\bmu_u').
\end{align*}
Therefore, by Condition \ref{con:obsx}, we achieve
\begin{align*}
\frac{1}{\sqrt{T}}\max_{j\in[N]}\|\hat{\bu}_{\cdot j}-{\bu}_{\cdot j}\|_2
&\leq \frac{1}{\sqrt{T}}\|\bX\|_2\left\|\left(\frac{1}{T} \wtilde\bX'\wtilde\bX \right)^{-1}\right\|_2 \max_{j\in[N]}\left\|\frac{1}{T}\wtilde\bX'(\bu_{\cdot j}- \mu_j^u\biota_T) \right\|_2 \\
&\lesssim \max_{j\in[N]}\left\|\frac{1}{T}\wtilde\bX'(\bu_{\cdot j}- \mu_j^u\biota_T) \right\|_2.
\end{align*}

Conditioning on $\widetilde{\bX}$, for each $(j,k)$ the summands
$\{\widetilde x_{sk}(u_{sj}-\mu_{u,j})\}_{s=1}^T$ are independent and
sub-Gaussian with its variance proxy proportional to $\phi$ under Condition \ref{con:dgp_fe}. A standard tail bound for weighted sums of independent sub-Gaussian variables
implies that, for any $t>0$,
\begin{align*}
\Pro \left(
\left|
\frac{1}{T}\sum_{s=1}^T \widetilde x_{sk}(u_{sj}-\mu_{u,j})
\right|
\ge t
\;\Big|\;\widetilde{\bX}
\right)
\leq
2\exp\!\left(
-\frac{cTt^2}{\phi}
\right).
\end{align*}
Applying a union bound over $(j,k)\in[N]\times[K]$ with choosing $t\asymp \sqrt{\phi\log(NK)/T}$ gives the stated probability bound. This completes the proof. 
\end{proof}





\vspp 

\section{Theory of Exact Recovery via Adaptive Lasso Screening} \label{suppl:adaL}

This provides an auxiliary screening result based on the adaptive Lasso. The adaptive Lasso screening involves replacing the penalty 
term $\lambda_\textsf{L} \|\bbeta\|_1$ with $\lambda_\textsf{aL} \|\bgamma \circ \bbeta\|_1$ 
in \eqref{feasibleLasso}, where $\lambda_\textsf{aL}$ is a positive regularization coefficient, $\bgamma\in\bbR^N$ is a (data-dependent) vector of weights, such as $\gamma_j=1/|\hat{\beta}_j^\textsf{L}|$, and $\circ$ indicates the Hadamard product.  
Given the adaptive Lasso estimates $\what{\bbeta}_\textsf{aL}$, the set of discoveries is defined as $\what\cS_\textsf{aL}=\{j: |\hat{\beta}_j^\textsf{aL}| > 0\}$. 

Unlike the sure-screening guarantees used in the main text, the adaptive Lasso can deliver exact support recovery under stronger minimum-signal and design conditions. We include this result for completeness and as a theoretical benchmark; the second-step post-selection analysis in the main text continues to focus on post-Lasso estimation.

\begin{thm}[adaptive Lasso screening]\label{thm:alasso}
Suppose that Conditions \ref{con:WFmodel}--\ref{con:dgp_fe}  hold and the minimum signal satisfies
\begin{align}
&\min_{j\in\cS}|\beta_j|
\gg \phi s \sqrt{\frac{\log N}{T}}. \label{con:aw-min} 
\end{align}
If the weight vector is constructed by $\gamma_j=1/|\hat{\beta}_j^\textsf{L}|$ with the Lasso regularization coefficient \eqref{cond:lasso_lamb}, and if the regularization coefficient is set to be
\begin{align}
\lambda_\textsf{aL} &\asymp \min_{j\in\cS}|\beta_j| \sqrt{\frac{\phi\log N}{T}}, \label{cond:alasso_lamb} 
\end{align}
then $\cS = \what\cS_\textsf{aL}$ occurs with high probability. 
\end{thm}

\subsection{Proof of Theorem \ref{thm:alasso}}

\begin{proof}
In this proof, let $n=N\wedge T$. For $\bb\in\bbR^N$ and $\bgamma=(1/|\hat{\beta}_1^\textsf{L}|,\dots,1/|\hat{\beta}_N^\textsf{L}|)'$, define the objective function
\begin{align}
Q_n(\bb) = T^{-1}\|\bz-\bR\bb\|_2^2 + 2\lambda_{\textsf{aL}} \|\bgamma \circ \bb\|_1.
\end{align}
Then the adaptive Lasso estimator is obtained as $\what\bbeta_\textsf{aL}=\argmin_{\bb\in\bbR^N}Q_n(\bb)$. We also 
define the oracle estimator as $\what\bbeta^o=\argmin_{\bb\in\cB}Q_n(\bb)$, where $\cB=\{\bb\in\bbR^N:\bb_{\cS^c}=\bzero\}$ for given $\cS$. If we can say $\what\bbeta_\textsf{aL}=\what\bbeta^o$ asymptotically, we conclude $\Pro(\what\cS_\textsf{aL}=\cS)\to 1$. 
The proof consists of two steps; we first prove the consistency of the oracle estimator defined on $\cB$ in Step 1, and then check if it is indeed an asymptotic minimizer of $\bb\mapsto Q_n(\bb)$ on $\bbR^N$ in Step 2. 

Step 1. Prove that $\what\bbeta_{\cS}^o$ is consistent for $\bbeta_{\cS}$ in $\ell_2$-norm with the convergence rate, $h_n=\sqrt{s\phi\log N/T}$. To this end, we show that for any $\ep>0$ there exists (large) constant $C>0$ such that 
\begin{align}
\inf_{\|\bu\|_2=C} Q_n(\bbeta_{\cS} +h_n\bu,\bzero_{N-s}) > Q_n(\bbeta_{\cS} ,\bzero_{N-s}) \label{aLobj}
\end{align}
holds with high probability. Let $D_n(\bu)=Q_n(\bbeta_{\cS} +h_n\bu,\bzero_{N-s}) - Q_n(\bbeta_{\cS} ,\bzero_{N-s})$ for $\bu\in\bbR^{s}$. After some algebra, applying H\"older's inequality gives 
\begin{align}
D_n(\bu)
&= h_n^2\bu'(T^{-1}\bR_{\cS}'\bR_{\cS})\bu 
+ 2h_n(T^{-1}\bR_{\cS}'\bR\bbeta -\alpha\bmu_{\cS})'\bu 
- 2h_n(T^{-1}\bz'\bR_{\cS}-\alpha \bmu_{\cS}')\bu \\
&\qquad- 2\lambda_{\textsf{aL}}(\|\bgamma_{\cS}\circ\bbeta_{\cS} \|_1
-\|\bgamma_{\cS}\circ(\bbeta_{\cS} +h_n\bu)\|_1 ) \\
&\geq h_n^2\bu'(T^{-1}\bR_{\cS}'\bR_{\cS})\bu 
- 2h_n\|T^{-1}\bR_{\cS}'\bR\bbeta -\alpha\bmu_{\cS}\|_\infty\|\bu\|_1 \\
&\qquad- 2h_n\|T^{-1}\bR_{\cS}'\bz-\alpha \bmu_{\cS}\|_\infty\|\bu\|_1 
- 2h_n\lambda_{\textsf{aL}}\|\bgamma_{\cS}\circ\bu\|_1. 
\end{align}
Hereafter, we use $\|\bOmega\|_2\in[c,1/c]$ for some constant $c>0$, obtained by Lemma \ref{lem:Omega_sp}. Evaluate each term of the lower bound of $D_n(\bu)$. 
For the first term, using H\"older's inequality and Lemma \ref{lem:concent_1}, we have
\begin{align}
h_n^2\bu'\left(T^{-1}\bR_{\cS}'\bR_{\cS}\right)\bu
&\geq h_n^2\bu'\left(\E\br_t^{\cS}{\br_t^{\cS}}'\right)\bu 
- h_n^2\left|\bu'(T^{-1}\bR_{\cS}'\bR_{\cS}-\E\br_t^{\cS}{\br_t^{\cS}}')\bu\right|\\
&\geq h_n^2\lambda_{\min}(\bOmega^{-1})\|\bu\|_2^2 
- h_n^2\|T^{-1}\bR_{\cS}'\bR_{\cS}-\E\br_t^{\cS}{\br_t^{\cS}}'\|_{\max}\|\bu\|_1^2 \\
&\geq h_n^2\|\bu\|_2^2/\|\bOmega\|_2 
-h_n^2O\left(\sqrt{s^2\log N/T}\right)\|\bu\|_2^2 \\
&\asymp h_n^2\left\{ 1 - o(1) \right\}\|\bu\|_2^2
\end{align}
with probability at least $1-O(N^{-\nu})$ for any fixed constant $\nu>2$ determined in Lemma \ref{lem:concent_1}(f). 
Consider the last term.  Lemma \ref{lem:weight_alasso} together with condition \eqref{cond:alasso_lamb} entails
\begin{align}
\lambda_\textsf{aL}\|\bgamma_{\cS}\|_2 \lesssim \frac{\sqrt{s}\lambda_\textsf{aL}}{\min_{j\in\cS}|\beta_j|}  
\asymp \sqrt{\frac{s\phi\log N}{T}}
= h_n
\end{align}
with probability at least $1-O(N^{-\nu})$. 
Thus, we have
\begin{align}
- 2h_n\lambda_{\textsf{aL}}\|\bgamma_{\cS}\circ\bu\|_1
\geq - 2h_n \lambda_{\textsf{aL}} \|\bgamma_{\cS}\|_2\|\bu\|_2 
\gtrsim -h_n^2\|\bu\|_2
\end{align}
with probability at least $1-O(N^{-\nu})$. Consider the remaining two terms. Lemma \ref{lem:concent_1}(d)(e) with  $\|\bu\|_1\leq \sqrt{s}\|\bu\|_2$ yield
\begin{align}
&- 2h_n \left( \|T^{-1}\bR_{\cS}'\bR\bbeta  -\alpha\bmu_{\cS}\|_\infty 
+ \|T^{-1}\bR_{\cS}'\bz-\alpha \bmu_{\cS}\|_\infty \right)\|\bu\|_1 \\
&\gtrsim - h_n \left( \sqrt{\frac{\phi\log N}{T}} + \sqrt{\frac{\log N}{T}} \right)\sqrt{s}\|\bu\|_2 \\
&\asymp - h_n \sqrt{\frac{s\phi\log N}{T}} \|\bu\|_2 
\gtrsim - h_n^2\|\bu\|_2.
\end{align}
Putting together the pieces, we obtain
\begin{align}
\inf_{\|\bu\|_2=C}D_n(\bu)
&\gtrsim\left\{ 1 - o(1) \right\} h_n^2C^2-O(h_n^2)C.
\end{align}
Therefore, taking sufficiently large $C$ makes $\inf_{\|\bu\|_2=C}D_n(\bu)$ become positive, and hence \eqref{aLobj} holds with probability at least $1-O(N^{-\nu})$.

Step 2. Prove that the oracle estimator, $\what\bbeta^o=(\what{\bbeta}_{\cS}^{o\prime},\bzero_{N-s}')'$, is indeed a minimizer of $\bb\mapsto Q_n(\bb)$ over $\bb\in\bbR^N$, leading to $\what\bbeta^o=\what\bbeta_\textsf{aL}$. In view of the KKT condition, it is sufficient to prove
\begin{align}
    \|\bgamma_{\cS^c}^{-}\circ\bR_{\cS^c}'(\bz-\bR_{\cS}\what{\bbeta}_{\cS}^o)/T\|_\infty < \lambda_\textsf{aL}. \label{ineq:step2}
\end{align}
Recall $\|\what\bbeta_{\textsf{L}}-\bbeta\|_\infty
\lesssim u_n$ in \eqref{infty_norm} and condition \eqref{con:aw-min}, which is expressed as
\begin{align}
\min_{j\in\cS}|\beta_j|\gg u_n \sqrt{s}\|(\bOmega^{-1})_{\cS^c\cS}\|_{2,\infty}.
\end{align}
Combining these inequalities and using condition \eqref{cond:alasso_lamb} give
\begin{align*}
\|\what\bbeta_{\textsf{L}}-\bbeta\|_\infty
\lesssim u_n
\ll \frac{\min_{j\in\cS}|\beta_j|}{\sqrt{s}\|(\bOmega^{-1})_{\cS^c\cS}\|_{2,\infty}}
\asymp \frac{\lambda_\textsf{aL}\sqrt{s}}{ h_n\sqrt{s}\|(\bOmega^{-1})_{\cS^c\cS}\|_{2,\infty}}, 
\end{align*}
which implies
\begin{align}
\|\what\bbeta_{\textsf{L}}-\bbeta\|_\infty \|(\bOmega^{-1})_{\cS^c\cS}\|_{2,\infty}h_n
\ll \lambda_\textsf{aL}. \label{ineq:la/min}
\end{align}
By Lemma \ref{lem:kkt_alasso} and \eqref{ineq:la/min}, we obtain
\begin{align}
\|\bgamma_{\cS^c}^{-}\circ\bR_{\cS^c}'(\bz-\bR_{\cS}\what{\bbeta}_{\cS}^o)/T\|_\infty
&\lesssim \|\what\bbeta_{\textsf{L}}-\bbeta\|_\infty \|(\bOmega^{-1})_{\cS^c\cS}\|_{2,\infty} h_n \\
&\ll   \lambda_\textsf{aL}.
\end{align}
Therefore, regardless of the positive constant factor, \eqref{ineq:step2} eventually becomes true. This completes the proof. 
\end{proof}

\begin{lem}\label{lem:weight_alasso}
Let $\bgamma=\left(1/|\hat{\beta}_1^{\textsf{L}}|,\dots,1/|\hat{\beta}_N^{\textsf{L}}|\right)'$, where $\hat{\beta}_j^{\textsf{L}} $ is the lasso estimator considered in Theorem \ref{thm:lasso}. Under the same conditions as in Theorem \ref{thm:alasso}, we have
\begin{align}
\Pro\left(\|\bgamma_\cS\|_2 \leq \frac{2\sqrt{s}}{\min_{j\in\cS}|\beta_j|}\right) \geq 1-O(N^{-\nu}).
\end{align}
\end{lem}
\begin{proof}
For any $x> 0$, we have
\begin{align}
\Pro\left(\|\bgamma_\cS\|_2 > x\right)
&\leq \Pro\left(\frac{\sqrt{s}}{\min_{j\in\cS}|\hat{\beta}_j^{\textsf{L}}|} > x\right) \\
&= \Pro\left( \min_{j\in\cS}|\hat{\beta}_j^{\textsf{L}}|<\sqrt{s}/x\right)
\leq \Pro\left( \|\what\bbeta-\bbeta\|_\infty > \min_{j\in\cS}|\beta_j| -\sqrt{s}/x \right).
\end{align}
For $u_n=\sqrt{s\phi\log N/T}$, we have established \eqref{infty_norm}: 
\begin{align}
\Pro\left( \|\what\bbeta_{\textsf{L}}-\bbeta\|_\infty \lesssim u_n \right)\geq 1-O(N^{-\nu}).
\end{align}
The minimum signal condition \eqref{con:aw-min} is rewritten as
\begin{align}
\min_{j\in\cS}|\beta_j|\gg u_n\sqrt{s}\|(\bOmega^{-1})_{\cS^c\cS}\|_{2,\infty} \geq u_n. 
\end{align}
Taking $x=2\sqrt{s}/\min_{j\in\cS}|\beta_j|$ gives 
\begin{align}
\Pro\left(\|\bgamma_\cS\|_2 > 2\sqrt{s}/\min_{j\in\cS}|\beta_j|\right)
&\leq \Pro\left( \|\what\bbeta-\bbeta\|_\infty > \min_{j\in\cS}|\beta_j|/2 \right) \\
&\leq \Pro\left( \|\what\bbeta-\bbeta\|_\infty \gg u_n \right)
\end{align}
Therefore, the above upper bound is at most $O(N^{-\nu})$ by \eqref{infty_norm}. This completes the proof. 
\end{proof}

\begin{lem}\label{lem:kkt_alasso}
Let $\bgamma^-=(|\hat{\beta}_1^\textsf{L}|,\dots,|\hat{\beta}_N^\textsf{L}|)'$. Under the same conditions as in Theorem \ref{thm:alasso}, we have
\begin{align}
\|\bgamma_{\cS^c}^{-}\circ\bR_{\cS^c}'(\bz-\bR_{\cS}\what{\bbeta}_{\cS}^o)/T\|_\infty 
\lesssim \|\what\bbeta_\textsf{L}-\bbeta\|_\infty \|(\bOmega^{-1})_{\cS^c\cS}\|_{2,\infty} \sqrt{\frac{N^{\varphi_1}s\log N}{T}}.
\end{align}
\end{lem}
\begin{proof}
Let 
\begin{align*}
\cN=\{\bb=(\bb_{\cS}',\bb_{\cS^c}')'\in\bbR^N: \bb_{\cS^c}=\bzero_{N-s}, \|\bb_{\cS}-\bbeta_{\cS} \|_2\lesssim h_n\}
\end{align*}
with $h_n=\sqrt{N^{\varphi_1}s\log N/T}$ as in the proof of Theorem \ref{thm:alasso}. Then the event,  $\what\bbeta^o\in\cN$, occurs with probability at least $1-O(N^{-\nu})$ by Step 1 in the proof of Theorem \ref{thm:alasso}. Thus on this event, by the triangle inequality and Lemma \ref{lem:concent_1}, we have
\begin{align}
&\|\bR_{\cS^c}'(\bz-\bR_{\cS}\what{\bbeta}_{\cS}^o)/T\|_\infty\\
&\leq \|\bR_{\cS^c}'\bz/T-\alpha\bmu_{\cS^c}\|_\infty + \|\bR_{\cS^c}'\bR\bbeta /T-\alpha\bmu_{\cS^c}\|_\infty 
+ \|\bR_{\cS^c}'\bR_{\cS}/T\|_{2,\infty}\|\what\bbeta_{\cS}^o-\bbeta_{\cS} \|_2 \\
&\leq \|\bR'\bz/T-\alpha\bmu\|_\infty
+ \|\bR'\bR\bbeta/T-\alpha\bmu\|_\infty \\
&\qquad + \left(\|(\bOmega^{-1})_{\cS^c\cS}\|_{2,\infty} 
+ \sqrt{s}\|\bR'\bR/T - \bOmega^{-1}\|_{\max}\right)
\sup_{\bb\in\cN}\|\bb_{\cS}-\bbeta_{\cS} \|_2 \\
&\lesssim \sqrt{\frac{\log N}{T}} + (N^{\varphi_1/2}\wedge \sqrt{s})\sqrt{\frac{\log N}{T}} + \left(\|(\bOmega^{-1})_{\cS^c\cS}\|_{2,\infty} +\sqrt{\frac{s\log N}{T}}\right)h_n \\
&\lesssim \|(\bOmega^{-1})_{\cS^c\cS}\|_{2,\infty}  h_n,
\end{align}
where the last inequality holds by $\sqrt{s\log N/T}\lesssim 1$ due to Condition \ref{con:sparse}. 
Meanwhile, since $\bbeta_{\cS^c}=\bzero$ by the definition, we have
\begin{align}
\|\bgamma_{\cS^c}^{-}\|_\infty = \|\what\bbeta_{\cS^c}^\textsf{L}-\bbeta_{\cS^c}\|_\infty
\leq \|\what\bbeta_\textsf{L}-\bbeta\|_\infty. 
\end{align}
Consequently, combining these bounds gives
\begin{align}
\|\bgamma_{\cS^c}^{-}\circ\bR_{\cS^c}'(\bz-\bR_{\cS}\what{\bbeta}_{\cS}^o)/T\|_\infty
&\leq \|\bgamma_{\cS^c}^{-}\|_\infty \|\bR_{\cS^c}'(\bz-\bR_{\cS}\what{\bbeta}_{\cS}^o)/T\|_\infty \\
&\lesssim \|\what\bbeta_\textsf{L}-\bbeta\|_\infty \|(\bOmega^{-1})_{\cS^c\cS}\|_{2,\infty}  h_n.
\end{align}
This completes the proof.
\end{proof}

\section{Multicollinearity Caused by a Strong Signal}\label{sim:multico}

Suppose that the excess return vector $\br_t\in\mathbb{R}^N$ is independently generated from the Gaussian single factor model:
\begin{align}\label{model:fact0}
\br_t 
\sim \text{ind.\ } \mathcal{N}\left(\ba \mu_x+\bmu_u,~ \sigma_f^2 \ba\ba'+\bSigma_u \right), 
\end{align}
where $\sigma_f^2=1$ and $\bSigma_u=\bI_N$ are factor and error variances, respectively. Suppose that $\mu_x=1$, $\bmu_u=0.1\biota$, and $\ba=(c/\sqrt{N})\biota$ for $c\in\{0.1,1,5,10,20\}$ with $\biota$ a vector of $N$ ones. Following Section \ref{subsec:lasso}, we consider Lasso screening of $z_t\sim \text{ind.\ }\mathcal{N}(1,\, 0.1)$ on $\br_t$ for $T\in\{100, 250,500,1000,1500\}$ and $N=500$.  In either case, the target $\bbeta(1)=\bOmega\bmu$ contains no zero entries, so ideally the Lasso estimate should also contain no zeros. However, when the signal is strong, this fails completely. The behavior is depicted in Figure \ref{fig:false}, where we draw the power surface of the Lasso,  $100\times|\what\cS_{\textsf{L}}|/N~\%$, against each $T$ and signal strength $c$. 
\begin{figure}[H]
\centering
\includegraphics[width=9cm]{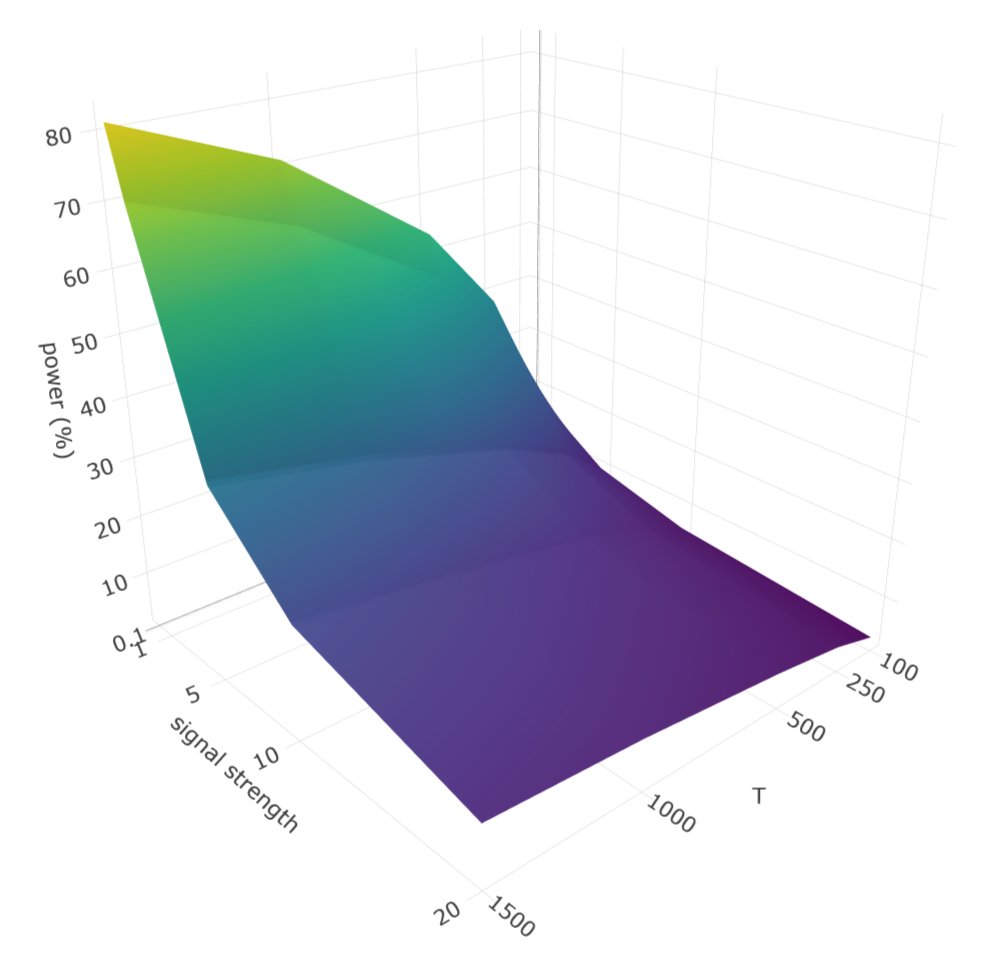}
\caption{Power of the Lasso when all the $N=500$ elements of $\bbeta(1)$ are nonzeros for each pair of $T$ and signal strength. Under this setting, the Lasso should find all entries to be nonzero; however, when the signal is strong, it identifies only a small fraction as nonzero even with reasonably large sample sizes.}\label{fig:false}
\end{figure}

\newpage

\section{Supplementary Materials for Empirical Results in Section \ref{emp:mainsec}} \label{emp:add emp main}

\subsection{Additional Empirical Results in Section \ref{emp:mainsec}} \label{emp:add emp}

Complementing the main findings in Table \ref{tab:meanrSR}, Table \ref{tab:meanrSR_appendix} presents two additional sets of results regarding the 
out-of-sample performance of the portfolios. The left panel reports the results of the long-window stress test for the FPS$^2$, FMAXSER, QIS, FF3, 
and EW portfolios when $J=500$. The right panel compares the performance of the FPS$^2$ portfolio against the market benchmark (Mkt).

\begin{table}[h!]

   \centering
  \caption{Long-Window Stress Test (Left) and Performance Comparison (Right)}

\begin{adjustbox}{scale=0.65}

 \begin{threeparttable}

\vsp 

  \renewcommand{\arraystretch}{1.3}
  
    \begin{tabular}{lrrrrrrrrrrrrr}
          &       &       &       &       &       &       &       &       &       &       &       &       &  \\
\cmidrule{1-6}\cmidrule{9-14}          &       &       &       &       &       &       &       &       &       &       &       &       &  \\
    \multicolumn{6}{l}{($i$) $\widehat{\mathrm{SR}} $ in the overall out-of-sample period (-2025/11)} &       &       & \multicolumn{6}{l}{($i$) $\widehat{\mathrm{SR}} $ in the overall out-of-sample period (-2025/11)} \\
          & FPS$^2$ & FM    & EW    & QIS   & FF3   &       &       &       & FPS$^2$ & Market  &       &       &  \\
    $ J = 500 $ & 0.111 & 0.035 & 0.063 & 0.064 & 0.100 &       &       & \multicolumn{1}{l}{$ J = 100 $} & 0.082 & 0.073 &       &       &  \\
          &       &       &       &       &       &       &       & \multicolumn{1}{l}{$ J = 200 $} & 0.102 & 0.080 &       &       &  \\
    \multicolumn{6}{l}{($ii$) $\widehat{\mathrm{SR}} $ in the pre-pandemic period (-2019/12)} &       &       & \multicolumn{1}{l}{$ J = 300 $} & 0.100 & 0.079 &       &       &  \\
          & FPS$^2$ & FM    & EW    & QIS   & FF3   &       &       &       &       &       &       &       &  \\
    $ J = 500 $ & 0.127 & 0.041 & 0.098 & 0.050 & 0.102 &       &       & \multicolumn{6}{l}{($ii$) $\widehat{\mathrm{SR}} $ in the pre-pandemic period (-2019/12)} \\
          &       &       &       &       &       &       &       &       & FPS$^2$ & Market  &       &       &  \\
    \multicolumn{6}{l}{($iii$) $\widehat{\mathrm{SR}} $ in the post-pandemic period (2020/1-2025/11)} &       &       & \multicolumn{1}{l}{$ J = 100 $} & 0.119 & 0.068 &       &       &  \\
          & FPS$^2$ & FM    & EW    & QIS   & FF3   &       &       & \multicolumn{1}{l}{$ J = 200 $} & 0.122 & 0.077 &       &       &  \\
    $ J = 500 $ & 0.083 & 0.031 & 0.024 & 0.082 & 0.100 &       &       & \multicolumn{1}{l}{$ J = 300 $} & 0.117 & 0.076 &       &       &  \\
          &       &       &       &       &       &       &       &       &       &       &       &       &  \\
    \multicolumn{6}{l}{($iv$) $\widehat{\mathrm{SR}}_{\tau_c}$  in the overall out-of-sample period (-2025/11)} &       &       & \multicolumn{6}{l}{($iii$) $\widehat{\mathrm{SR}} $ in the post-pandemic period (2020/1-2025/11)} \\
          & FPS$^2$ & FM    & EW    & QIS   & FF3   &       &       &       & FPS$^2$ & Market  &       &       &  \\
    $ J = 500 $ & 0.055 & 0.017 & 0.062 & 0.024 & 0.095 &       &       & \multicolumn{1}{l}{$ J = 100 $} & -0.009 & 0.086 &       &       &  \\
    {[Turnover]} & [0.826] & [0.412] & [0.027] & [0.108] & [0.403] &       &       & \multicolumn{1}{l}{$ J = 200 $} & 0.049 & 0.086 &       &       &  \\
          &       &       &       &       &       &       &       & \multicolumn{1}{l}{$ J = 300 $} & 0.068 & 0.086 &       &       &  \\
    \multicolumn{6}{l}{($v$) $\widehat{\mathrm{SR}}_{\tau_c}$ in the pre-pandemic period (-2019/12)} &       &       &       &       &       &       &       &  \\
          & FPS$^2$ & FM    & EW    & QIS   & FF3   &       &       & \multicolumn{6}{l}{($iv$) $\widehat{\mathrm{SR}}_{\tau_c}$  in the overall out-of-sample period (-2025/11)} \\
    $ J = 500 $ & 0.065 & 0.022 & 0.097 & 0.008 & 0.095 &       &       &       & FPS$^2$ & Market  &       &       &  \\
    {[Turnover]} & [1.150] & [0.168] & [0.022] & [0.163] & [0.264] &       &       & \multicolumn{1}{l}{$ J = 100 $} & 0.019 & 0.073 &       &       &  \\
          &       &       &       &       &       &       &       & \multicolumn{1}{l}{$ J = 200 $} & 0.042 & 0.080 &       &       &  \\
    \multicolumn{6}{l}{($vi$) $\widehat{\mathrm{SR}}_{\tau_c}$ in the post-pandemic period (2020/1-2025/11)} &       &       & \multicolumn{1}{l}{$ J = 300 $} & 0.044 & 0.079 &       &       &  \\
          & FPS$^2$ & FM    & EW    & QIS   & FF3   &       &       &       &       &       &       &       &  \\
    $ J = 500 $ & 0.037 & 0.015 & 0.023 & 0.043 & 0.097 &       &       & \multicolumn{6}{l}{($v$) $\widehat{\mathrm{SR}}_{\tau_c}$ in the pre-pandemic period (-2019/12)} \\
    {[Turnover]} & [0.844] & [0.263] & [0.028] & [0.238] & [0.095] &       &       &       & FPS$^2$ & Market  &       &       &  \\
          &       &       &       &       &       &       &       & \multicolumn{1}{l}{$ J = 100 $} & 0.050 & 0.068 &       &       &  \\
\cmidrule{1-6}          &       &       &       &       &       &       &       & \multicolumn{1}{l}{$ J = 200 $} & 0.057 & 0.077 &       &       &  \\
          &       &       &       &       &       &       &       & \multicolumn{1}{l}{$ J = 300 $} & 0.052 & 0.076 &       &       &  \\
          &       &       &       &       &       &       &       &       &       &       &       &       &  \\
          &       &       &       &       &       &       &       & \multicolumn{6}{l}{($vi$) $\widehat{\mathrm{SR}}_{\tau_c}$ in the post-pandemic period (2020/1-2025/11)} \\
          &       &       &       &       &       &       &       &       & FPS$^2$ & Market  &       &       &  \\
          &       &       &       &       &       &       &       & \multicolumn{1}{l}{$ J = 100 $} & -0.057 & 0.086 &       &       &  \\
          &       &       &       &       &       &       &       & \multicolumn{1}{l}{$ J = 200 $} & 0.002 & 0.086 &       &       &  \\
          &       &       &       &       &       &       &       & \multicolumn{1}{l}{$ J = 300 $} & 0.030 & 0.086 &       &       &  \\
          &       &       &       &       &       &       &       &       &       &       &       &       &  \\
\cmidrule{9-14}    
\end{tabular}%

\vsp 

\begin{tablenotes}
\item This table complements the findings presented in Table \ref{tab:meanrSR}. The left panel reports the gross and net Sharpe ratios across the respective 
sample periods evaluated with an extended rolling window of $J=500$. The right panel displays the gross and net Sharpe ratios for the return 
series of the FPS$^2$ portfolio and the market (Mkt) factor across different rolling window sizes.
\end{tablenotes}
  
\label{tab:meanrSR_appendix}%

\end{threeparttable}

\end{adjustbox}

\end{table}%

\subsection{Stock List of Empirical Application in Section \ref{emp:mainsec}} \label{emp:stock list}

The complete universe of the 968 stock tickers used in our empirical application is provided below:

\clearpage 



\section*{Supplementary material}

| No. | Ticker | No. | Ticker | No. | Ticker | No. | Ticker | No. | Ticker | No. | Ticker | No. | Ticker | No. | Ticker |
|---|---|---|---|---|---|---|---|---|---|---|---|---|---|---|---|
| 1 | A | 51 | AME | 101 | BAC | 151 | CAG | 201 | CMI | 251 | DALRQ | 301 | EIX | 351 | FIS |
| 2 | AABA | 52 | AMG | 102 | BALL | 152 | CAH | 202 | CMS | 252 | DASH | 302 | EKDKQ | 352 | FISV |
| 3 | AAL | 53 | AMGN | 103 | BAX | 153 | CAL | 203 | CMVT | 253 | DAY | 303 | EL | 353 | FITB |
| 4 | AAMRQ | 54 | AMH | 104 | BBBY | 154 | CAR | 204 | CMX | 254 | DCNAQ | 304 | ELV | 354 | FL |
| 5 | AAP | 55 | AMP | 105 | BBT | 155 | CARR | 205 | CNC | 255 | DD | 305 | EMC | 355 | FLIR |
| 6 | AAPL | 56 | AMT | 106 | BBWI | 156 | CAT | 206 | CNP | 256 | DDOG | 306 | EME | 356 | FLMIQ |
| 7 | ABBV | 57 | AMZN | 107 | BBY | 157 | CB | 207 | CNW | 257 | DDS | 307 | EMN | 357 | FLR |
| 8 | ABKFQ | 58 | AN | 108 | BC | 158 | CBB | 208 | CNX | 258 | DE | 308 | EMR | 358 | FLS |
| 9 | ABMD | 59 | ANDV | 109 | BCO | 159 | CBE | 209 | CNXT | 259 | DECK | 309 | ENPH | 359 | FLTWQ |
| 10 | ABNB | 60 | ANDW | 110 | BCR | 160 | CBH | 210 | COF | 260 | DELL | 310 | ENRNQ | 360 | FMC |
| 11 | ABT | 61 | ANET | 111 | BDK | 161 | CBOE | 211 | COIN | 261 | DFS | 311 | ENS | 361 | FMCC |
| 12 | ABX | 62 | ANF | 112 | BDX | 162 | CBRE | 212 | COL | 262 | DG | 312 | EOG | 362 | FMY |
| 13 | ACAS | 63 | ANRZQ | 113 | BEAM | 163 | CCB | 213 | COO | 263 | DGX | 313 | EOP | 363 | FNMA |
| 14 | ACGL | 64 | ANSS | 114 | BEN | 164 | CCEP | 214 | COP | 264 | DHI | 314 | EPAM | 364 | FOSL |
| 15 | ACKH | 65 | AON | 115 | BF.B | 165 | CCI | 215 | COR | 265 | DHR | 315 | EQ | 365 | FOX |
| 16 | ACN | 66 | AOS | 116 | BFO | 166 | CCK | 216 | COST | 266 | DI | 316 | EQIX | 366 | FOXA |
| 17 | ACV | 67 | APA | 117 | BG | 167 | CCL | 217 | COTY | 267 | DIS | 317 | EQR | 367 | FRT |
| 18 | ADBE | 68 | APD | 118 | BGEN | 168 | CCTYQ | 218 | COV | 268 | DISCA | 318 | EQT | 368 | FSH |
| 19 | ADCT | 69 | APH | 119 | BHF | 169 | CCU | 219 | CPAY | 269 | DISCK | 319 | ERIE | 369 | FSL |
| 20 | ADI | 70 | APO | 120 | BHMSQ | 170 | CDNS | 220 | CPB | 270 | DISH | 320 | ES | 370 | FSLR |
| 21 | ADM | 71 | APOL | 121 | BIIB | 171 | CDW | 221 | CPNLQ | 271 | DJ | 321 | ESRX | 371 | FTI |
| 22 | ADP | 72 | APP | 122 | BIO | 172 | CE | 222 | CPRI | 272 | DLR | 322 | ESS | 372 | FTNT |
| 23 | ADSK | 73 | APTV | 123 | BK | 173 | CEG | 223 | CPRT | 273 | DLTR | 323 | ETFC | 373 | FTV |
| 24 | ADT | 74 | AR | 124 | BKNG | 174 | CELG | 224 | CPT | 274 | DLX | 324 | ETN | 374 | FWLT |
| 25 | AEE | 75 | ARC | 125 | BKR | 175 | CEN | 225 | CR | 275 | DNB | 325 | ETR | 375 | G |
| 26 | AEP | 76 | ARE | 126 | BLDR | 176 | CEPH | 226 | CRL | 276 | DO | 326 | ETSY | 376 | GAPTQ |
| 27 | AES | 77 | ARG | 127 | BLK | 177 | CERN | 227 | CRM | 277 | DOC | 327 | EVHC | 377 | GAS |
| 28 | AET | 78 | ARNC | 128 | BLS | 178 | CF | 228 | CRWD | 278 | DOV | 328 | EVRG | 378 | GD |
| 29 | AFL | 79 | ASC | 129 | BMC | 179 | CFG | 229 | CSCO | 279 | DOW | 329 | EW | 379 | GDDY |
| 30 | AGN | 80 | ASH | 130 | BMS | 180 | CFN | 230 | CSE | 280 | DPHIQ | 330 | EXC | 380 | GDW |
| 31 | AIG | 81 | ASND | 131 | BMY | 181 | CFS | 231 | CSGP | 281 | DPZ | 331 | EXE | 381 | GE |
| 32 | AIT | 82 | ASO | 132 | BNI | 182 | CG | 232 | CSR | 282 | DRE | 332 | EXPD | 382 | GEHC |
| 33 | AIV | 83 | AT | 133 | BNL | 183 | CHD | 233 | CSRA | 283 | DRI | 333 | EXPE | 383 | GEN |
| 34 | AIZ | 84 | ATGE | 134 | BOAT | 184 | CHIR | 234 | CSX | 284 | DTE | 334 | EXR | 384 | GENZ |
| 35 | AJG | 85 | ATI | 135 | BOL | 185 | CHRS | 235 | CTAS | 285 | DUK | 335 | F | 385 | GGP |
| 36 | AKAM | 86 | ATO | 136 | BR | 186 | CHRW | 236 | CTB | 286 | DVA | 336 | FANG | 386 | GHC |
| 37 | AKS | 87 | ATVI | 137 | BRCM | 187 | CHTR | 237 | CTLT | 287 | DVN | 337 | FAST | 387 | GILD |
| 38 | AL | 88 | AVB | 138 | BRK.B | 188 | CI | 238 | CTRA | 288 | DXC | 338 | FBF | 388 | GIS |
| 39 | ALB | 89 | AVGO | 139 | BRL | 189 | CIEN | 239 | CTSH | 289 | DXCM | 339 | FCN | 389 | GL |
| 40 | ALGN | 90 | AVP | 140 | BRO | 190 | CIN | 240 | CTVA | 290 | DYN | 340 | FCPT | 390 | GLD |
| 41 | ALK | 91 | AVY | 141 | BSC | 191 | CINF | 241 | CTXS | 291 | EA | 341 | FCX | 391 | GLK |
| 42 | ALL | 92 | AW | 142 | BSX | 192 | CITGQ | 242 | CVC | 292 | EBAY | 342 | FDC | 392 | GLW |
| 43 | ALLE | 93 | AWE | 143 | BTUUQ | 193 | CL | 243 | CVG | 293 | EC | 343 | FDO | 393 | GM |
| 44 | ALTR | 94 | AWK | 144 | BUD | 194 | CLF | 244 | CVH | 294 | ECH | 344 | FDS | 394 | GMCR |
| 45 | ALXN | 95 | AXON | 145 | BVSN | 195 | CLX | 245 | CVS | 295 | ECL | 345 | FDX | 395 | GME |
| 46 | AM | 96 | AXP | 146 | BWA | 196 | CMA | 246 | CVX | 296 | ED | 346 | FE | 396 | GNRC |
| 47 | AMAT | 97 | AYE | 147 | BX | 197 | CMCSA | 247 | CXO | 297 | EFU | 347 | FFIV | 397 | GNT |
| 48 | AMCC | 98 | AYI | 148 | BXP | 198 | CMCSK | 248 | CZR | 298 | EFX | 348 | FG | 398 | GNW |
| 49 | AMCR | 99 | AZO | 149 | C | 199 | CME | 249 | D | 299 | EG | 349 | FHN | 399 | GOOG |
| 50 | AMD | 100 | BA | 150 | CA | 200 | CMG | 250 | DAL | 300 | EHC | 350 | FICO | 400 | GOOGL |

| No. | Ticker | No. | Ticker | No. | Ticker | No. | Ticker | No. | Ticker | No. | Ticker | No. | Ticker | No. | Ticker |
|---|---|---|---|---|---|---|---|---|---|---|---|---|---|---|---|
| 401 | GPC | 451 | KKR | 501 | IBM | 551 | MA | 601 | MTB | 651 | NWSA | 701 | PLL | 751 | ROP |
| 402 | GPN | 452 | ICE | 502 | KLAC | 552 | MAA | 602 | MTCH | 652 | NXPI | 702 | PLTR | 752 | ROST |
| 403 | GR | 453 | IDXX | 503 | KMB | 553 | MAC | 603 | MTD | 653 | NXTL | 703 | PM | 753 | RRC |
| 404 | GRA | 454 | IEX | 504 | KMG | 554 | MAR | 604 | MTG | 654 | NYT | 704 | PMCS | 754 | RRD |
| 405 | GRMN | 455 | IFF | 505 | KMI | 555 | MAS | 605 | MTL | 655 | O | 705 | PNC | 755 | RSG |
| 406 | GRN | 456 | IKN | 506 | KMX | 556 | MAT | 606 | MTLQQ | 656 | ODFL | 706 | PNR | 756 | RSHCQ |
| 407 | GS | 457 | ILMN | 507 | KO | 557 | MAY | 607 | MTW | 657 | ODP | 707 | PNW | 757 | RTN |
| 408 | GT | 458 | IMNX | 508 | KR | 558 | MBI | 608 | MU | 658 | OGN | 708 | PODD | 758 | RTX |
| 409 | GTE | 459 | INCLF | 509 | KRB | 559 | MCD | 609 | MUR | 659 | OI | 709 | POOL | 759 | RVTY |
| 410 | GTW | 460 | INCY | 510 | KRFT | 560 | MCHP | 610 | MWI | 660 | OKE | 710 | PPG | 760 | RX |
| 411 | GWW | 461 | INGR | 511 | KRI | 561 | MCK | 611 | MWW | 661 | OM | 711 | PPL | 761 | RYAN |
| 412 | H | 462 | INTC | 512 | KSE | 562 | MCO | 612 | MXIM | 662 | OMC | 712 | PRGO | 762 | RYI |
| 413 | HAL | 463 | INTU | 513 | KSS | 563 | MD | 613 | MYL | 663 | OMX | 713 | PRU | 763 | S |
| 414 | HAR | 464 | INVH | 514 | KSU | 564 | MDLZ | 614 | NAV | 664 | ON | 714 | PSA | 764 | SANM |
| 415 | HAS | 465 | IP | 515 | KVUE | 565 | MDP | 615 | NAVI | 665 | ORCL | 715 | PSX | 765 | SAPE |
| 416 | HBAN | 466 | IPG | 516 | L | 566 | MDT | 616 | NBL | 666 | ORLY | 716 | PTC | 766 | SB |
| 417 | HBI | 467 | IPGP | 517 | LDG | 567 | MEA | 617 | NBR | 667 | OTIS | 717 | PTV | 767 | SBAC |
| 418 | HCA | 468 | IQV | 518 | LDOS | 568 | MEDI | 618 | NC | 668 | OWENQ | 718 | PVH | 768 | SBNY |
| 419 | HCBK | 469 | IR | 519 | LEG | 569 | MEE | 619 | NCC | 669 | OXY | 719 | PVN | 769 | SBUX |
| 420 | HCP | 470 | IRM | 520 | LEHMQ | 570 | MEL | 620 | NCLH | 670 | PAC | 720 | PWER | 770 | SCG |
| 421 | HD | 471 | ISRG | 521 | LEN | 571 | MER | 621 | NDAQ | 671 | PALM | 721 | PWR | 771 | SCHW |
| 422 | HES | 472 | IT | 522 | LH | 572 | MET | 622 | NDSN | 672 | PANW | 722 | PX | 772 | SCI |
| 423 | HET | 473 | ITT | 523 | LHX | 573 | META | 623 | NE | 673 | PARA | 723 | PXD | 773 | SDS |
| 424 | HFC | 474 | ITW | 524 | LII | 574 | MFE | 624 | NEE | 674 | PAS | 724 | PYPL | 774 | SE |
| 425 | HI | 475 | IVZ | 525 | LIN | 575 | MGM | 625 | NEM | 675 | PAYC | 725 | PZE | 775 | SEBL |
| 426 | HIG | 476 | J | 526 | LKQ | 576 | MHK | 626 | NFB | 676 | PAYX | 726 | QCOM | 776 | SEDG |
| 427 | HII | 477 | JAVA | 527 | LLTC | 577 | MHS | 627 | NFLX | 677 | PBCT | 727 | QEP | 777 | SEE |
| 428 | HLT | 478 | JBHT | 528 | LLY | 578 | MI | 628 | NFX | 678 | PBG | 728 | QLGC | 778 | SFS |
| 429 | HNZ | 479 | JBL | 529 | LM | 579 | MIR | 629 | NI | 679 | PBI | 729 | QRVO | 779 | SGID |
| 430 | HOG | 480 | JCI | 530 | LMT | 580 | MJN | 630 | NKE | 680 | PCAR | 730 | QTRN | 780 | SGP |
| 431 | HOLX | 481 | JEF | 531 | LNC | 581 | MKC | 631 | NKTR | 681 | PCH | 731 | R | 781 | SHLD |
| 432 | HON | 482 | JH | 532 | LNT | 582 | MKTX | 632 | NLC | 682 | PCG | 732 | RAI | 782 | SHW |
| 433 | HOOD | 483 | JHF | 533 | LO | 583 | MLM | 633 | NLSN | 683 | PCP | 733 | RATL | 783 | SIAL |
| 434 | HOT | 484 | JKHY | 534 | LOR | 584 | MMI | 634 | NMK | 684 | PD | 734 | RBK | 784 | SIG |
| 435 | HP | 485 | JNJ | 535 | LOW | 585 | MMM | 635 | NOC | 685 | PDCO | 735 | RCL | 785 | SJM |
| 436 | HPC | 486 | JNPR | 536 | LPX | 586 | MNR | 636 | NOV | 686 | PDG | 736 | RDC | 786 | SLB |
| 437 | HPE | 487 | JNS | 537 | LRCX | 587 | MNST | 637 | NOVL | 687 | PEG | 737 | RDS.A | 787 | SLG |
| 438 | HPQ | 488 | JNY | 538 | LSI | 588 | MO | 638 | NOW | 688 | PENN | 738 | REG | 788 | SLM |
| 439 | HRB | 489 | JOY | 539 | LU | 589 | MOH | 639 | NRG | 689 | PEP | 739 | REGN | 789 | SMCI |
| 440 | HRL | 490 | JPM | 540 | LUB | 590 | MOLX | 640 | NSC | 690 | PETM | 740 | RF | 790 | SMI |
| 441 | HSH | 491 | JWN | 541 | LULU | 591 | MOS | 641 | NSI | 691 | PFE | 741 | RHI | 791 | SNA |
| 442 | HSIC | 492 | K | 542 | LUMN | 592 | MPC | 642 | NSM | 692 | PFG | 742 | RHT | 792 | SNPS |
| 443 | HSP | 493 | KATE | 543 | LUV | 593 | MPWR | 643 | NTAP | 693 | PG | 743 | RIG | 793 | SNT |
| 444 | HST | 494 | KBH | 544 | LVLT | 594 | MRK | 644 | NTRS | 694 | PGL | 744 | RL | 794 | SNV |
| 445 | HSY | 495 | KDP | 545 | LVS | 595 | MRNA | 645 | NUE | 695 | PGR | 745 | RMD | 795 | SO |
| 446 | HUBB | 496 | KEY | 546 | LW | 596 | MRO | 646 | NVDA | 696 | PH | 746 | RML | 796 | SOV |
| 447 | HUM | 497 | KEYS | 547 | LXK | 597 | MS | 647 | NVR | 697 | PHB | 747 | ROH | 797 | SPGI |
| 448 | HWM | 498 | KG | 548 | LYB | 598 | MSCI | 648 | NWL | 698 | PHM | 748 | ROK | 798 | SPG |
| 449 | IAC | 499 | KHC | 549 | LYV | 599 | MSFT | 649 | NWS | 699 | PKG | 749 | ROL | 799 | SRCL |
| 450 | IBKR | 500 | KIM | 550 | M | 600 | MSI | 650 | — | 700 | PLD | 750 | — | 800 | SRE |

| No. | Ticker | No. | Ticker | No. | Ticker | No. | Ticker |
|---|---|---|---|---|---|---|---|
| 801 | SRR | 851 | TPR | 901 | VIAV | 951 | WYE |
| 802 | SSP | 852 | TRB | 902 | VICI | 952 | WYNN |
| 803 | STE | 853 | TRGP | 903 | VLO | 953 | X |
| 804 | STJ | 854 | TRIP | 904 | VLTO | 954 | XEC |
| 805 | STLD | 855 | TRMB | 905 | VMC | 955 | XEL |
| 806 | STT | 856 | TROW | 906 | VNO | 956 | XLNX |
| 807 | STX | 857 | TRV | 907 | VNT | 957 | XOM |
| 808 | STZ | 858 | TRW | 908 | VO | 958 | XRAY |
| 809 | SUB | 859 | TSCO | 909 | VRSK | 959 | XRX |
| 810 | SUN | 860 | TSLA | 910 | VRSN | 960 | XTO |
| 811 | SUNEQ | 861 | TSN | 911 | VRTS | 961 | XYL |
| 812 | SVU | 862 | TSS | 912 | VRTX | 962 | XYZ |
| 813 | SWK | 863 | TT | 913 | VST | 963 | YUM |
| 814 | SWKS | 864 | TTD | 914 | VSTNQ | 964 | ZBH |
| 815 | SWN | 865 | TTWO | 915 | VTR | 965 | ZBRA |
| 816 | SWY | 866 | TWC | 916 | VTRS | 966 | ZION |
| 817 | SXCL | 867 | TWTR | 917 | VTSS | 967 | ZTS |
| 818 | SYF | 868 | TWX | 918 | VZ | 968 | SPTR |
| 819 | SYK | 869 | TX | 919 | WAB | | |
| 820 | SYY | 870 | TXN | 920 | WAI | | |
| 821 | T | 871 | TXT | 921 | WAMUQ | | |
| 822 | TA | 872 | TXU | 922 | WAT | | |
| 823 | TAP | 873 | TYL | 923 | WB | | |
| 824 | TDC | 874 | UA | 924 | WBA | | |
| 825 | TDG | 875 | UAA | 925 | WBD | | |
| 826 | TDY | 876 | UAL | 926 | WCG | | |
| 827 | TE | 877 | UAWGQ | 927 | WCOEQ | | |
| 828 | TECH | 878 | UBER | 928 | WDAY | | |
| 829 | TEG | 879 | UCC | 929 | WDC | | |
| 830 | TEL | 880 | UCL | 930 | WEC | | |
| 831 | TEN | 881 | UDR | 931 | WELL | | |
| 832 | TER | 882 | UHS | 932 | WEN | | |
| 833 | TEX | 883 | UIS | 933 | WFC | | |
| 834 | TFC | 884 | UK | 934 | WFM | | |
| 835 | TFX | 885 | ULTA | 935 | WHR | | |
| 836 | TGNA | 886 | UN | 936 | WLTW | | |
| 837 | TGT | 887 | UNH | 937 | WM | | |
| 838 | THC | 888 | UNM | 938 | WMB | | |
| 839 | THY | 889 | UNP | 939 | WMT | | |
| 840 | TIE | 890 | UPC | 940 | WNDXQ | | |
| 841 | TIN | 891 | UPS | 941 | WOR | | |
| 842 | TJX | 892 | URBN | 942 | WPX | | |
| 843 | TKO | 893 | URI | 943 | WRB | | |
| 844 | TKR | 894 | USB | 944 | WRK | | |
| 845 | TLAB | 895 | USBC | 945 | WSM | | |
| 846 | TMO | 896 | UST | 946 | WST | | |
| 847 | TMUS | 897 | V | 947 | WTW | | |
| 848 | TNB | 898 | VAR | 948 | WU | | |
| 849 | TOY | 899 | VFC | 949 | WWY | | |
| 850 | TPL | 900 | VIAB | 950 | WY | | |



\clearpage 

\bibliographystylesuppl{chicago}
\bibliographysuppl{portfolio}

\end{document}